\documentclass[a4paper,twocolumn,11pt,unpublished]{quantumarticle}
\pdfoutput=1

\usepackage[utf8]{inputenc}
\usepackage[T1]{fontenc}
\usepackage[english]{babel}
\usepackage{amsmath}
\usepackage{amssymb}
\usepackage{amsthm}
\usepackage{bm}
\usepackage{booktabs}
\usepackage{graphicx}
\usepackage{subcaption}
\usepackage{mathtools}
\usepackage{microtype}
\usepackage[numbers,sort&compress]{natbib}
\usepackage{listings}
\usepackage{xcolor}
\usepackage[
    colorlinks=true,
    linkcolor=refBlue,
    citecolor=citeRed,
    urlcolor=refBlue
]{hyperref}

\definecolor{refBlue}{RGB}{0,82,136}
\definecolor{citeRed}{RGB}{160,65,0}

\graphicspath{{figures/}}

\definecolor{codebackground}{HTML}{F5F8FB}
\definecolor{codeframe}{HTML}{B8C8D8}
\definecolor{codekeyword}{HTML}{7254A3}
\definecolor{codestring}{HTML}{27864A}
\definecolor{codecomment}{HTML}{627D98}
\lstdefinestyle{pythonsource}{%
  language=Python,
  basicstyle=\fontsize{6}{6.5}\selectfont\ttfamily,
  keywordstyle=\color{codekeyword}\bfseries,
  stringstyle=\color{codestring},
  commentstyle=\color{codecomment}\itshape,
  backgroundcolor=\color{codebackground},
  frame=single,
  rulecolor=\color{codeframe},
  framesep=5pt,
  xleftmargin=0.35em,
  xrightmargin=0.35em,
  aboveskip=0.65\baselineskip,
  belowskip=0.65\baselineskip,
  showstringspaces=false,
  columns=fullflexible,
  keepspaces=true,
  breaklines=true,
  breakatwhitespace=false,
  tabsize=4,
  upquote=true,
}

\newtheorem{theorem}{Theorem}[section]
\newtheorem{lemma}[theorem]{Lemma}
\newtheorem{proposition}[theorem]{Proposition}
\newtheorem{corollary}[theorem]{Corollary}
\theoremstyle{definition}
\newtheorem{definition}[theorem]{Definition}

\theoremstyle{remark}
\newtheorem{remark}[theorem]{Remark}

\newcommand{\C}{\mathbb{C}}

\newcommand{\Ftwo}{\mathbb{F}_2}
\newcommand{\Nin}{n_{\mathrm{in}}}

\newcommand{\ket}[1]{\lvert #1 \rangle}
\newcommand{\bra}[1]{\langle #1 \rvert}

\newcommand{\ketbra}[2]{\ket{#1}\!\bra{#2}}
\newcommand{\proj}[1]{\ketbra{#1}{#1}}
\newcommand{\Pauli}[1]{\mathcal{P}_{#1}}
\newcommand{\Clifford}[1]{\mathcal{C}_{#1}}
\newcommand{\Hilb}[1]{\mathcal{H}_{#1}}
\newcommand{\QubitSpace}[1]{(\C^2)^{\otimes #1}}
\newcommand{\tr}{\operatorname{Tr}}
\newcommand{\opket}[1]{\lvert #1\rangle\!\rangle}

\newcommand{\llbracket}{\mathopen{[\![}}
\newcommand{\rrbracket}{\mathclose{]\!]}}
\newcommand{\interpret}[1]{\llbracket #1\rrbracket}

\newcommand{\fp}[1]{}

\title{Composing Detector Error Models}

\author{Vadym Kliuchnikov}
\email{vkliuchnikov@nvidia.com}
\affiliation{NVIDIA Corporation}

\author{Justin Lietz}
\email{jlietz@nvidia.com}
\affiliation{NVIDIA Corporation}

\author{Fernando Pastawski}
\email{fpastawski@nvidia.com}
\affiliation{NVIDIA Corporation}


\begin{document}

\maketitle

\begin{abstract}
Fault-tolerant quantum programs thread together reusable logical operations, 
yet correcting errors requires comparing measurements across operation boundaries. 
Correlations tie each operation's error analysis to the computation around it.
Adaptive computation makes this a runtime challenge:
measurement results determine which operation comes next, so the error analysis must keep pace with execution.

We introduce extended detector error models (EDEMs) for stabilizer circuits with stochastic Pauli faults.
EDEMs compose in sequence and in parallel, mirroring the composition of physical circuit realizations and their detector contracts.
We prove how detectors, the measurement parities used to diagnose errors, can be split across circuit boundaries, 
and establish conditions under which composition recovers the assembled circuit’s detector error model. 
This structure supports symbolic precompilation and analysis of entire families of circuits. 
It also defines error models for decoding windows and the boundary information through which committed corrections affect subsequent windows. 
The same interface thus connects the design of individual logical operations to the decoding of an unfolding quantum program.

\end{abstract}

\section{Introduction}
\label{sec:introduction}

\subsection{Overview and main result}
\label{sec:introduction-overview}

Fault-tolerant quantum computation protects logical information with stabilizer codes and repeatedly measures redundant parities to diagnose noise~\cite{Dennis2002topologicalmemory,Fowler2012surfacecodes,Terhal2015quantumerrorcorrection}.
In the absence of noise, certain parities of measurement outcomes are deterministic; these parities are called \emph{detectors}~\cite{Gidney2021stimfaststabilizer,McEwen2023relaxinghardware,Derks2025designingfault}.
A \emph{detector error model} (DEM) records how each elementary fault flips detectors and logical observables, together with the fault probabilities.
Decoders such as minimum-weight matching~\cite{Higgott2025sparseblossom}, union-find~\cite{Delfosse2021unionfind}, and belief propagation with post-processing~\cite{Panteleev2021degeneratequantumldpc} are configured from this data, and established tools construct it from the complete physical circuit~\cite{Gidney2021stimfaststabilizer}.

The fault-tolerant computations we target are adaptive: measurement outcomes and logical control decisions determine parts of the physical circuit during execution. 
Consequently, the complete physical circuit is not available in advance.
Logical operations are chosen adaptively, based on decoded outcomes that fix Pauli or Clifford frame corrections~\cite{Knill2005realisticallynoisy,Chamberland2018cliffordframe,Horsman2012latticesurgery,Litinski2019gameofsurfacecodes,Gidney2019flexiblelayout}, and decoding must keep pace with the physical circuit as it executes~\cite{Terhal2015quantumerrorcorrection,Battistel2023realtimedecoding,GoogleQuantumAI2025belowthreshold,Caune2026realtimedecoding}.
The DEM must therefore be constructed incrementally, in tandem with the circuit.
Physical circuits for such computations are assembled from \emph{gadgets}: circuit fragments that provably implement specific logical operations on encoded information~\cite{Aliferis2006accuracythreshold,Kliuchnikov2023stabilizercircuitverification,Beverland2024faulttolerance}.
This suggests a division of labor: analyze each gadget type once, offline, and assemble the DEM of any logical circuit from the per-gadget results, without revisiting the gadgets' physical circuits.

The obstruction is that detectors do not respect gadget boundaries.
A detector typically compares syndrome information extracted by different gadgets, so a gadget partition of the circuit does not induce a partition of its detectors.
We address this in two steps.
First, we prove that detectors can be cut at circuit boundaries: every detector of a serial composition splits into parts local to the pieces once the boundary is extended by a small number of extra outcome bits (Sec.~\ref{sec:cutting-detectors}).
The proof is constructive and preserves a chosen detector basis, and the composed basis is complete whenever the component bases are.
At gadget boundaries, these extra bits can be taken to be \emph{virtual syndromes}: outcomes of virtual stabilizer measurements inserted at the boundary, following Ref.~\cite{DEQ, WuDeq}.
Second, we define the \emph{extended detector error model} (EDEM) of a gadget (Sec.~\ref{sec:gadget-dems}): a linear map over $\Ftwo$ that takes the gadget's elementary faults and its incoming interface data---virtual-syndrome flips and boundary-error coordinates~(Sec.~\ref{sec:gadget-profiles})---to detector and observable flips and the corresponding outgoing interface data.
EDEMs compose the way gadgets do: wiring the per-gadget maps according to the logical circuit and contracting the result recovers the DEM of the full circuit, for the detector basis induced by the gadget detector contracts, provided connected ports carry matching codes and code presentations~(Sec.~\ref{sec:gadget-profiles}).

Because composition is linear-algebraic, constructing a DEM becomes evaluating a wired diagram of sparse block matrices, and the same few blocks appear many times.
This makes a symbolic approach effective (Sec.~\ref{sec:symbolic-construction}): assign a variable to each distinct block, contract the diagram once into polynomial expressions in these variables, and evaluate each distinct expression once.
The symbolic form also supports structural arguments; for syndrome-extraction gadgets satisfying mild assumptions, we derive closed-form expressions establishing the banded structure of the repeated-extraction part of memory-experiment DEMs.

\subsection{Consequences and applications}
\label{sec:introduction-applications}

\emph{Modular logical instruction-set design.}
An EDEM provides a compositional error-model interface for a logical gadget. 
Gadget implementations can therefore be analyzed, compared, and optimized independently, then assembled into larger computations without flattening and reanalyzing their physical circuits, provided that connected code and code-presentation interfaces match~(Sec.~\ref{sec:gadget-profiles}).
In this way, EDEMs support the fault-tolerance and decoding layer of a logical instruction-set architecture built from reusable, independently characterized components.
They also provide a self-similar structure naturally suited to describing concatenated fault-tolerant protocols.

\emph{Incremental construction.}
Per-gadget EDEMs can be computed and stored ahead of time, once per gadget type.
Extending the DEM by one more gadget call at runtime then touches only these stored maps and the interface data of its ports, which suits adaptive computations whose instruction stream is decided on the fly (Sec.~\ref{sec:sliding-window-decoding}).

\emph{Sliding-window decoding.}
The same interfaces define the decoding problem of a window: the EDEM of any convex set of gadget calls (Sec.~\ref{sec:sliding-window-decoding}).
Causality makes the channel-check matrix block lower triangular with respect to the gadget circuit, and detector span, read off the per-gadget blocks, makes it block local, reducing to block banded for chain-like schedules.
Committed corrections act on later windows only through their interface outputs, so residual errors are carried forward as boundary and virtual-syndrome data rather than as retained detector history.

\emph{Compact representation.}
The composed DEM is a block matrix with few distinct blocks, so it can be stored as a grid of pointers into a small set of blocks.
This reduces memory traffic, for example when transferring DEMs to accelerators for decoding.

\emph{Amortization across experiments.}
Families of experiments that share a fault-tolerance protocol share polynomial expressions, so symbolic evaluations can be cached and reused across
many circuits using the same fault-tolerant instruction set, and
many circuit instantiations, such as sweeps over distances, rounds, or noise parameters (Secs.~\ref{sec:symbolic-caching} and~\ref{sec:examples-implementation}).

\subsection{Contributions and organization}
\label{sec:contributions}

After Sec.~\ref{sec:preliminaries}, which fixes notation, reviews stabilizer codes, stabilizer circuits, DEMs, gadgets, and the modularization of detectors by virtual syndromes~\cite{DEQ, WuDeq}, and introduces the channel-level conventions used throughout (boundary constraints, detector and outcome maps, and the diagram notation for linear maps), our contributions are:
\begin{itemize}
    \item a detector-cutting theorem showing that every detector of a composed stabilizer circuit splits across a serial cut, constructively and preserving completeness of the detector basis (Sec.~\ref{sec:cutting-detectors});
    \item detector maps with virtual-syndrome interfaces and extended detector error models, together with the composition rule that recovers the DEM of a gadget composition from the per-gadget EDEMs (Secs.~\ref{sec:cutting-detectors} and~\ref{sec:gadget-dems});
    \item a symbolic evaluation method for EDEM diagrams, illustrated by a closed-form derivation of the structure of repeated-syndrome-extraction DEMs (Sec.~\ref{sec:symbolic-construction});
    \item a Rust prototype that implements the formalism for four code families and reproduces the DEMs Stim derives from the flat circuits (Sec.~\ref{sec:examples-implementation});
    \item window EDEMs as sub-diagram contractions, a commit-and-fold rule and interface dispositions for sequential decoding, a per-gadget criterion for detector span, and a distinction between quantum execution and decoding dependencies, with interface requirements for schedules that use different orders (Sec.~\ref{sec:sliding-window-decoding}).
\end{itemize}
Limitations and next steps are discussed in Sec.~\ref{sec:conclusion}.

\paragraph{Relation to prior work.}
Modular, analyze-each-gadget-once error analysis is not new.
The \texttt{deq} system~\cite{DEQ, WuDeq} and the framework of Ref.~\cite{Beverland2024faulttolerance} modularize detectors with virtual stabilizer measurements, analyze each gadget type offline, and assemble decoding problems, including sliding windows, online; our detector decomposition builds directly on these ideas.
Stim~\cite{Gidney2021stimfaststabilizer} provides stabilizer flows~\cite{Gidney2024magicstatecultivation}, and the Stim project's \texttt{stimflow} library~\cite{GidneyStimflow} composes chunks carrying these relations and completes detectors across piece boundaries.
In that workflow, composition produces a circuit from which the DEM is then derived monolithically.
Our contribution relative to both is the uniform linear-map form of the gadget error model, the accompanying detector-cutting proofs, and the symbolic layer this uniformity enables.
The linear-algebraic view of fault effects follows the spacetime and outcome-code perspectives on Clifford circuits~\cite{Delfosse2023spacetimecodes,Kliuchnikov2023stabilizercircuitverification,Beverland2024faulttolerance}; DEMs as design objects appear in Ref.~\cite{Derks2025designingfault}, and detectors as design objects in Ref.~\cite{McEwen2023relaxinghardware}.
Compositional descriptions of fault-tolerant protocols at the operator level include logical blocks, fusion-based computation, and ZX-calculus approaches~\cite{Bombin2023logicalblocks,Bartolucci2023fusionbased,Bombin2024unifyingflavorsof,Rodatz2025faulttolerancebyconstruction}; these compose the protocols themselves, whereas we compose their error models.
Windowed and modular decoding are reviewed where we use them (Sec.~\ref{sec:sliding-window-decoding}).

\section{Preliminaries}
\label{sec:preliminaries}

This section establishes notation and reviews the background used throughout the paper.
Sections~\ref{sec:preliminaries-codes}, \ref{sec:preliminaries-circuits}, and~\ref{sec:preliminaries-dems} recall standard material: stabilizer codes and encoders, stabilizer circuits with Pauli faults~\cite{Kliuchnikov2023stabilizercircuitverification,Beverland2024faulttolerance}, and detector error models~\cite{Gidney2021stimfaststabilizer,Derks2025designingfault}.
Section~\ref{sec:preliminaries-gadgets} reviews gadgets~\cite{Kliuchnikov2023stabilizercircuitverification}, and Sec.~\ref{sec:preliminaries-virtual-detectors} reviews the modularization of detectors using virtual syndromes introduced in Ref.~\cite{DEQ}.
Sections~\ref{sec:stabilizer-maps}, \ref{sec:preliminaries-detectors}, and~\ref{sec:preliminaries-diagrams} fix this paper's conventions: stabilizer channels with declared classical ports and their boundary constraints, which are the stabilizer flows of Ref.~\cite{Gidney2024magicstatecultivation} written at the channel level; detectors, observables, and outcome maps in that language; and the diagram notation for linear maps over $\Ftwo$.
Apart from these conventions, the material in this section is standard; our results begin in Sec.~\ref{sec:cutting-detectors}.
A reader familiar with stabilizer codes, stabilizer circuits, and detector error models may skim Secs.~\ref{sec:preliminaries-codes}, \ref{sec:preliminaries-circuits}, and~\ref{sec:preliminaries-dems}.

\subsection{Notation and stabilizer codes}
\label{sec:preliminaries-codes}

We recall standard facts about stabilizer codes; see, e.g., Ref.~\cite{Terhal2015quantumerrorcorrection}.
Let $\Pauli{n}$ denote the $n$-qubit Pauli group, with phases $\{\pm1,\pm i\}$ retained because the signs of Pauli observables encode measurement outcomes and syndromes.
A \emph{Pauli observable} is a Hermitian element of $\Pauli{n}$, and $\Clifford{n}$ denotes the $n$-qubit Clifford group.
For a bitvector $v\in\Ftwo^{\mathsf{R}}$ indexed by a register $\mathsf{R}$, we write $X_{\mathsf{R}}^{v}:=\bigotimes_{j\in\mathsf{R}}X_j^{v_j}$, and similarly $Z_{\mathsf{R}}^{v}$.
The \emph{phase-free Pauli group} $\Pauli{n}/\langle i\rangle$ is identified with $\Ftwo^{2n}$ through $(x,z)\mapsto X^{x}Z^{z}$.

A \emph{stabilizer group} $S\subseteq\Pauli{n}$ is an abelian group of Pauli observables that does not contain $-I$.
Its \emph{codespace} is the joint $+1$ eigenspace of $S$, with projector $\Pi_S=|S|^{-1}\sum_{s\in S}s$; if $S$ has $n-k$ independent generators, the codespace encodes $k$ logical qubits, and for $k=0$ it is spanned by a single \emph{stabilizer state}.
Physical representatives of logical Pauli operators are the Paulis that commute with $S$, and two representatives that differ by an element of $S$ act identically on the codespace.
A \emph{logical interpretation} of the code fixes an isomorphism
\begin{equation}
    \iota_S:\Pauli{k}\longrightarrow S^\perp/S,
    \qquad
    \iota_S(\ell)=[\bar\ell]:=\bar\ell S,
    \label{eq:logical-pauli-quotient}
\end{equation}
where $S^\perp$ denotes the group of Paulis that commute with $S$; it suffices to fix representatives $\bar X_1,\bar Z_1,\ldots,\bar X_k,\bar Z_k$.
Throughout the paper, every code has a fixed choice of generators $s_1,\ldots,s_{n-k}$ and of logical representatives; gadgets compose only when these choices agree on connected ports (Sec.~\ref{sec:preliminaries-gadgets}), and the code presentations of Sec.~\ref{sec:gadget-profiles} additionally allow overcomplete generator lists.

The code and its logical interpretation specify an \emph{encoding isometry} $E_{\boldsymbol{0}}:\QubitSpace{k}\to\QubitSpace{n}$, characterized by
\begin{equation}
    \begin{aligned}
        E_{\boldsymbol{0}}^\dagger E_{\boldsymbol{0}} &= I, \\
        sE_{\boldsymbol{0}} &= E_{\boldsymbol{0}} && (s\in S), \\
        \bar{\ell}E_{\boldsymbol{0}} &= E_{\boldsymbol{0}}\ell && (\ell\in\Pauli{k}).
    \end{aligned}
    \label{eq:encoding-isometry}
\end{equation}
As in Ref.~\cite{Kliuchnikov2023stabilizercircuitverification}, fix an \emph{encoding unitary}, a Clifford unitary $E:\QubitSpace{k}\otimes\QubitSpace{n-k}\to\QubitSpace{n}$ extending $E_{\boldsymbol{0}}$, whose second input is the \emph{syndrome register} associated with the chosen generators.
Its restrictions $E_\sigma:=E(I\otimes\ket{\sigma})$ to syndrome values $\sigma\in\Ftwo^{n-k}$ satisfy
\begin{equation}
    \begin{aligned}
        s_iE_\sigma &= (-1)^{\sigma_i}E_\sigma && (i=1,\ldots,n-k), \\
        \bar{\ell}E_\sigma &= E_\sigma\ell && (\ell\in\Pauli{k}),
    \end{aligned}
    \label{eq:syndrome-sector-isometries}
\end{equation}
so that $E_{\boldsymbol{0}}$ is the encoding isometry and the other $E_\sigma$ identify every syndrome sector with the same logical space.
The \emph{encoder} is the channel with a $k$-qubit quantum input and a declared classical syndrome input that prepares a logical state in the requested syndrome sector,
\begin{equation}
    \mathcal{E}\bigl(\rho_L\otimes\proj{\sigma}\bigr)
    :=
    E_\sigma\rho_L E_\sigma^\dagger,
    \label{eq:encoder}
\end{equation}
and the \emph{unencoder} is its inverse, the channel
\begin{equation}
    \mathcal{U}(\rho)
    :=
    \sum_{\sigma\in\Ftwo^{n-k}}
    E_\sigma^\dagger\rho E_\sigma
    \otimes
    \proj{\sigma},
    \label{eq:unencoder}
\end{equation}
whose quantum output is the $k$-qubit logical state and whose declared classical output is the syndrome of the chosen generators.
Indeed, $\mathcal{U}\circ\mathcal{E}$ is the identity, and $\mathcal{E}\circ\mathcal{U}$ acts as the identity on every state supported in a single syndrome sector.
Reading out the syndrome of the unencoder is the same as measuring the generators $s_1,\ldots,s_{n-k}$, so an unencoder followed by an encoder is a nondestructive measurement of the generators.

\begin{lemma}[Pull-back of Paulis through the encoding unitary]
    \label{lem:encoder-pullback}
    For every $P\in\Pauli{n}$, there are a logical Pauli $\ell_P\in\Pauli{k}$ and bitvectors $c(P),z(P)\in\Ftwo^{n-k}$ such that, up to a phase,
    \begin{equation}
        E^\dagger P E = \ell_P\otimes X^{c(P)}Z^{z(P)},
        \label{eq:encoder-pullback}
    \end{equation}
    where $c_i(P)=1$ precisely when $P$ anticommutes with $s_i$.
    The pair $(c(P),\ell_P)$ depends only on the coset $PS$.
    On states whose syndrome register is in $\ket{\boldsymbol{0}}$, the factor $Z^{z(P)}$ acts trivially.
\end{lemma}
\begin{proof}
    By Eq.~\eqref{eq:syndrome-sector-isometries}, $E^\dagger s_iE=I\otimes Z_i$ and $E^\dagger\bar\ell E=\ell\otimes I$.
    The Clifford conjugate $E^\dagger PE$ is a Pauli, and it anticommutes with $I\otimes Z_i$ exactly when $P$ anticommutes with $s_i$, which fixes its $X$-part on the syndrome register to be $X^{c(P)}$.
    Multiplying $P$ by $s\in S$ changes only the $Z$-part on the syndrome register, and $Z^{z}\ket{\boldsymbol 0}=\ket{\boldsymbol 0}$.
\end{proof}
We call $c(P)$ the \emph{syndrome} and $\ell_P$ the \emph{logical label} of $P$; Secs.~\ref{sec:gadget-profiles} and~\ref{sec:symbolic-memory} use them as the coordinates of boundary errors.

\subsection{Stabilizer channels and boundary constraints}
\label{sec:preliminaries-vectorization}
\label{sec:stabilizer-maps}
\label{sec:stabilizer-channel-constraints}

Stabilizer circuits with measurements, discards, and classically controlled Paulis implement quantum channels, and we formulate our discussion in terms of these channels.
For an operator $Q:\Hilb{A}\to\Hilb{B}$, let $\opket{Q}:=(Q\otimes I)\ket{\Omega_A}$ be its vectorization, where $\ket{\Omega_A}=\sum_x\ket{x}_A\ket{x}_{A'}$ is the unnormalized maximally entangled vector with a reference copy $A'$; then $\opket{LQR}=(L\otimes R^T)\opket{Q}$, so left and right Pauli actions on $Q$ become Pauli stabilizer constraints on $\opket{Q}$.
We call a nonzero $Q$ a \emph{stabilizer operator} if $\opket{Q}$ is proportional to a stabilizer state; Clifford unitaries, stabilizer projectors, and encoding isometries are examples.
A \emph{stabilizer channel} is a completely positive trace-preserving map $\Phi$ whose Choi operator $J(\Phi):=(\Phi\otimes\operatorname{id}_{A'})(\proj{\Omega_A})$ is a stabilizer operator~\cite{Yashin2025}.
Equivalently, $\Phi$ is fixed up to a scalar by $2(n_{\mathrm{in}}+n_{\mathrm{out}})$ independent equations of the four-slot form
\begin{equation}
    P_{\mathrm{out,L}}
    \Phi\!\left(
        P_{\mathrm{in,L}}\rho P_{\mathrm{in,R}}
    \right)
    P_{\mathrm{out,R}}
    =
    \Phi(\rho),
    \label{eq:stabilizer-superoperator-constraint}
\end{equation}
holding for every $\rho$, and we write $\mathsf{Stab}(\Phi)$ for the group of all such equations.
For a stabilizer operator $Q$, the map $\operatorname{Ad}_Q:\rho\mapsto Q\rho Q^\dagger$ is a stabilizer channel whenever it is trace preserving; tensor products and serial compositions of stabilizer channels are stabilizer channels~\cite{Yashin2025}.
Appendix~\ref{app:further-stabilizer-preliminaries} records the Choi-operator details, verifies that the elementary operations of Sec.~\ref{sec:preliminaries-circuits} are stabilizer channels, and compares this notion with the trace-non-increasing \emph{stabilizer maps} and the circuits of Ref.~\cite{Kliuchnikov2023stabilizercircuitverification}.
The one structural fact about stabilizer channels used in the main text is that Pauli errors propagate through them as Pauli errors.

The \emph{adjoint} $\Phi^\dagger$ of a channel $\Phi$ is the unital completely positive map defined by $\tr\bigl(\Phi(X)^\dagger Y\bigr)=\tr\bigl(X^\dagger\Phi^\dagger(Y)\bigr)$ for all operators $X$ and $Y$; it preserves Hermiticity and reverses the order of composition, $(\Phi_2\circ\Phi_1)^\dagger=\Phi_1^\dagger\circ\Phi_2^\dagger$.
\begin{lemma}[Pauli covariance]
    \label{lem:pauli-covariance}
    Let $\Phi$ be a stabilizer channel from register $A$ to register $B$.
    For every Pauli $R$ on $A$, there is a Pauli $Q$ on $B$ such that
    \begin{equation}
        \Phi\circ\operatorname{Ad}_R
        =
        \operatorname{Ad}_Q\circ\Phi .
        \label{eq:pauli-covariance}
    \end{equation}
    Dually, for every Pauli $P$ on $B$, the adjoint $\Phi^\dagger(P)$ is either zero or a Pauli on $A$ up to a sign.
\end{lemma}
The first statement is proved via a stabilizer dilation in Appendix~\ref{app:further-stabilizer-preliminaries} (Corollary~\ref{cor:stabilizer-channel-pauli-covariance}); the dual statement is Theorem~1 of Ref.~\cite{Yashin2025}.

\paragraph{Classical ports.}
We model classical bits as qubits so that a single formalism covers quantum wires and measurement records.
With $\mathcal{T}_Z(\rho):=(\rho+Z\rho Z)/2$ the dephasing channel, which is a stabilizer channel, a register $A_C$ of input ports or $B_C$ of output ports is \emph{declared classical} when
\begin{equation}
    \Phi\circ\mathcal{T}_{Z}^{\otimes A_C}
    =
    \Phi,
    \qquad
    \mathcal{T}_{Z}^{\otimes B_C}\circ\Phi
    =
    \Phi,
    \label{eq:classical-interface}
\end{equation}
respectively, where the twirl acts as the identity on the remaining ports.
Since $\mathcal{T}_Z=(\operatorname{id}+\operatorname{Ad}_Z)/2$, these conditions are the stabilizer equations
\begin{equation}
    \begin{aligned}
        \Phi(Z_i\rho Z_i)&=\Phi(\rho) && (i\in A_C), \\
        Z_j\Phi(\rho)Z_j&=\Phi(\rho) && (j\in B_C),
    \end{aligned}
    \label{eq:classical-interface-stabilizers}
\end{equation}
which are elements of $\mathsf{Stab}(\Phi)$.
A declaration records that a port admits a classical implementation and fault model; it does not alter the map.
The consequence we use below is that conjugating a declared classical output by $Z$ leaves the image of $\Phi$ invariant, so measurement outcomes are affected by faults only through bit flips.

\paragraph{Boundary constraints.}
Suppose $\Phi$ maps an input register $A$ to a quantum output register $B$ and a declared classical outcome register with port set $\mathsf{M}$, and write $m\in\Ftwo^{\mathsf{M}}$ for an outcome vector.
For $a\in\Ftwo^{\mathsf{M}}$ and $b\in\Ftwo$, let $Z_{\mathsf{M}}(a,b):=(-1)^bZ_{\mathsf{M}}^{a}$, so that $Z_{\mathsf{M}}(a,b)\ket{m}=(-1)^{a^Tm+b}\ket{m}$.
\begin{definition}[Boundary constraint]
    \label{def:boundary-constraint}
    A tuple $(P_{\mathrm{in}},P_{\mathrm{out}},a,b)$ of a Pauli $P_{\mathrm{in}}$ on $A$, a Pauli $P_{\mathrm{out}}$ on $B$, and an affine outcome parity is a \emph{boundary constraint} of $\Phi$ when
    \begin{equation}
        \bigl(P_{\mathrm{out}}\otimes Z_{\mathsf{M}}(a,b)\bigr)
        \Phi(P_{\mathrm{in}}\rho)
        =
        \Phi(\rho)
        \quad\text{for every }\rho.
        \label{eq:boundary-constraint}
    \end{equation}
    The boundary constraints form the \emph{boundary-constraint group} $R_{\mathsf{M}}(\Phi)$, the subgroup of one-sided elements of $\mathsf{Stab}(\Phi)$, i.e., of the equations~\eqref{eq:stabilizer-superoperator-constraint} with trivial right-action slots.
\end{definition}
The Pauli parts determine how to interpret a boundary constraint.
When $P_{\mathrm{in}}=P_{\mathrm{out}}=I$, it is a deterministic outcome parity; when only $P_{\mathrm{out}}=I$, it states that $\Phi$ measures the input Pauli $P_{\mathrm{in}}$ with outcome $a^Tm+b$; when only $P_{\mathrm{in}}=I$, it states that $P_{\mathrm{out}}$ stabilizes the output with sign $(-1)^{a^Tm+b}$; and when both are nontrivial, it transports a logical correlation from input to output.
A boundary constraint is a \emph{stabilizer flow} $P_{\mathrm{in}}\to P_{\mathrm{out}}$ with outcome parity $a^Tm+b$ in the sense of Ref.~\cite{Gidney2024magicstatecultivation}, as implemented in Stim~\cite{Gidney2021stimfaststabilizer}, written at the channel level; equivalently, it is a check of the outcome code or spacetime code of the circuit~\cite{Delfosse2023spacetimecodes,Kliuchnikov2023stabilizercircuitverification}.
The Paulis $P_{\mathrm{out}}$ of constraints with $P_{\mathrm{in}}=I$ form the \emph{output stabilizer group} of $\Phi$, with outcome-dependent signs.
Likewise, the Paulis $P_{\mathrm{in}}$ of constraints with $P_{\mathrm{out}}=I$ form its \emph{measured group}.
Boundary constraints compose serially: if $(P_{\mathrm{in}},P,a_1,b_1)\in R_{\mathsf{M}_1}(\Phi_1)$ and $(P,P_{\mathrm{out}},a_2,b_2)\in R_{\mathsf{M}_2}(\Phi_2)$, then $(P_{\mathrm{in}},P_{\mathrm{out}},(a_1,a_2),b_1+b_2)\in R_{\mathsf{M}_1\sqcup\mathsf{M}_2}(\Phi_2\circ\Phi_1)$, because $\Phi_2(P\sigma)=(P_{\mathrm{out}}\otimes Z_{\mathsf{M}_2}(a_2,b_2))\Phi_2(\sigma)$ for every $\sigma$.
In particular, if $\Phi_1$ prepares $P$ with sign $a_1^Tm_1+b_1$ and $\Phi_2$ measures $P$ with outcome $a_2^Tm_2+b_2$, then $a_1^Tm_1+a_2^Tm_2+b_1+b_2$ is a deterministic parity of $\Phi_2\circ\Phi_1$.

\subsection{Detectors, observables, and outcome maps}
\label{sec:preliminaries-detectors}

In the absence of faults, certain parities of measurement outcomes are deterministic; these parities are called \emph{detectors}~\cite{Gidney2021stimfaststabilizer,McEwen2023relaxinghardware,Derks2025designingfault}.
\begin{definition}[Detectors]
    \label{def:detector}
    A \emph{detector} of $\Phi$ is an affine expression $\langle a,m\rangle+b$ in its outcomes, with $\langle a,m\rangle:=a^Tm$, that equals zero for every outcome vector that $\Phi$ can produce from any input state; equivalently, $(I,I,a,b)\in R_{\mathsf{M}}(\Phi)$.
    If $\Phi$ has declared classical input ports $\mathsf{S}$, the expression may also involve the classical inputs $s$, taking the form $\langle a_0,s\rangle+\langle a,m\rangle+b$, and must vanish for every value of $s$; equivalently, $(Z_{\mathsf{S}}^{a_0},I,a,b)\in R_{\mathsf{M}}(\Phi)$.
    A detector is \emph{nontrivial} if its linear part is nonzero and \emph{supported on} a set of ports if its linear part vanishes outside that set.
    A finite family of detectors, labelled by a set $\mathsf{D}$, is \emph{declared} when it is supplied to the decoder; it is a \emph{detector basis} if it is linearly independent, and \emph{complete} if it generates every deterministic parity of $\Phi$.
    The observed value of a declared detector is its \emph{detection event}.
\end{definition}
For example, take three data qubits in an arbitrary state, measure the checks $Z_1Z_2$ and $Z_2Z_3$ of the three-qubit repetition code twice, with outcomes $m^{(1)}_j$ and $m^{(2)}_j$ for $j=1,2$, and then measure the data qubits destructively with outcomes $n_1,n_2,n_3$.
The detectors $m^{(1)}_j+m^{(2)}_j$ and $m^{(2)}_j+n_j+n_{j+1}$ form a complete basis, and $n_1$ reads out the logical $Z$ of the code.

\begin{definition}[Redundant outcomes and equivalence]
    \label{def:circuit-equivalence}
    An outcome bit is \emph{redundant} given a set of other outcomes if some detector contains it and is otherwise supported on that set; equivalently, it is an affine function of those outcomes on every attainable outcome vector.
    Two stabilizer circuits $C$ and $C'$ with the same quantum ports are \emph{equivalent under an affine outcome map} $f$ from the outcomes of $C'$ to those of $C$ if $\interpret{C}=(\operatorname{id}\otimes\Phi_f)\circ\interpret{C'}$, where $\Phi_f$ applies $f$ to the outcome register and discards the rest.
\end{definition}
Inserting a measurement-only subcircuit whose outcomes are all redundant after, such as the virtual measurements of Sec.~\ref{sec:preliminaries-virtual-detectors}, yields an equivalent circuit.

For a closed experiment, deterministic parities not declared as detectors may be reported as deterministic logical observables instead.
An \emph{outcome map} reports declared logical outcomes, called logical observables, as affine functions of the measurement outcomes.
For a gadget, these are the outcomes of its logical action~(Def.~\ref{def:gadget}), and they need not be deterministic.
A \emph{deterministic outcome map} reports selected deterministic parities of these logical outcomes, identified using detectors of the logical action.
Let $\mathsf{O}$ label the declared logical observables.
Collecting the declared detectors and observables gives the affine \emph{detector map} and \emph{outcome map}
\begin{equation}
    \mathrm{DM}:
    \Ftwo^{\mathsf{M}}
    \longrightarrow
    \Ftwo^{\mathsf{D}},
    \qquad
    \mathrm{OM}:
    \Ftwo^{\mathsf{M}}
    \longrightarrow
    \Ftwo^{\mathsf{O}},
    \label{eq:detector-and-logical-affine-maps}
\end{equation}
whose components are the declared detector and observable values.
For an affine map $f:\Ftwo^{\mathsf{X}}\to\Ftwo^{\mathsf{Y}}$, we write $\Delta f$ for its linear part,
\begin{equation}
    (\Delta f)(u)
    :=
    f(u)-f(0) = f(x+u)-f(x),
    \label{eq:affine-map-linearization}
\end{equation}
and $\Delta x$ for a change in a bitvector $x$, so that $\Delta y=(\Delta f)(\Delta x)$ whenever $y=f(x)$ and $\Delta(g\circ f)=(\Delta g)\circ(\Delta f)$.
Thus $\Delta\mathrm{DM}$ and $\Delta\mathrm{OM}$ map outcome flips to detector and observable flips.
Definition~\ref{def:gadget} below extends the outcome map to gadgets, whose observables are the outcomes of a logical circuit, and Definition~\ref{def:detector-map} extends the detector map to gadgets with virtual-syndrome ports; these extensions preserve the definitions above.

\subsection{Stabilizer circuits and Pauli faults}
\label{sec:preliminaries-circuits}
\label{sec:preliminaries-faults}

Following Ref.~\cite{Kliuchnikov2023stabilizercircuitverification}, a stabilizer circuit is assembled from a fixed set $\mathcal{G}$ of \emph{elementary operations}.
Each operation $g\in\mathcal{G}$ has finite labelled sets of input and output ports, a \emph{size} for each port (the number of qubits or bits it carries), and an interpretation $\interpret{g}$ as a stabilizer channel between the corresponding registers.
The sets $\mathcal{G}$ used in this paper consist of qubit preparations, Clifford unitaries, Pauli measurements with a declared classical outcome port, discards, fair-coin allocations, and Pauli corrections controlled by affine parities of earlier outcomes; Appendix~\ref{app:further-stabilizer-preliminaries} verifies that these are stabilizer channels.
\begin{definition}[Stabilizer circuit]
\label{def:stabilizer-circuit}
A \emph{stabilizer circuit} $C$ over $\mathcal{G}$ consists of
a finite set of boxes, each labelled by an operation $g\in\mathcal{G}$ and carrying its input and output ports;
finite labelled sets of \emph{boundary input ports} and \emph{boundary output ports}, each with a specified size, which are the inputs and outputs of $C$ itself;
and a bijective \emph{wiring} that pairs every source port (a boundary input port or a box output port) with a target port (a boundary output port or a box input port) of the same size.
The directed graph on the boxes induced by the wires must be acyclic.
\end{definition}
The boundary input ports carry the quantum inputs and declared classical inputs of the circuit, and the boundary output ports carry its quantum outputs and declared classical outcomes.
A subnetwork of a stabilizer circuit is \emph{convex} if no directed path of wires leaves it and then re-enters it.
With the wires crossing its boundary treated as boundary ports, a convex subnetwork is itself a stabilizer circuit and may be treated as a single box, whose interpretation is the channel defined next.
Convexity ensures that the resulting network remains acyclic~(Def.~\ref{def:window}).
The \emph{interpretation} $\interpret{C}$ of a stabilizer circuit is the stabilizer channel from its boundary input ports to its boundary output ports obtained by composing the channels $\interpret{g}$ along the wires, in any topological order of the boxes; we write
\begin{equation}
    \Phi_{C}
    :=
    \interpret{C}
    \label{eq:stabilizer-circuit-interpretation}
\end{equation}
and reuse the brackets for gadget circuits and map diagrams below, which are networks of the same shape.
A circuit is \emph{closed} if it has no quantum boundary ports, so its only outputs are declared classical outcomes; otherwise it is \emph{open}.

Lemma~\ref{lem:compressor-exists} and the verification conditions of Secs.~\ref{sec:gadget-profiles} and~\ref{sec:examples-implementation} use the following normal form for the outcomes of a stabilizer circuit, as computed by Clifford simulation~\cite{Kliuchnikov2023stabilizercircuitverification}.
\begin{proposition}[General form of outcomes]
    \label{prop:outcome-form}
    The outcome vector of a stabilizer circuit is an affine function $m=M(r,l)$ of a uniformly random bitvector $r$ and, for an open circuit, a bitvector $l$ whose entries are eigenvalues of input Pauli observables measured by the circuit.
    The random bits of a circuit are independent of those of any circuit composed after it.
    The attainable outcome vectors form an affine code, the \emph{outcome code} of Refs.~\cite{Delfosse2023spacetimecodes,Kliuchnikov2023stabilizercircuitverification}, whose check equations are exactly the deterministic parities, i.e., the parities constant in $(r,l)$.
\end{proposition}
The proof is the usual stabilizer-tableau argument: a Pauli measurement whose observable anticommutes with a current stabilizer yields a fresh fair coin, one whose observable lies in the current stabilizer group yields an affine function of earlier coins and the constant signs, and for open circuits the remaining case records an input observable.

Following Ref.~\cite{Beverland2024faulttolerance}, we separate the structural specification of the allowed faults from the probabilities assigned to them.
An \emph{elementary Pauli fault} is a Pauli insertion on a quantum wire, a bit flip on a classical wire, or a fixed collection of such insertions representing a correlated event; in the ambient qubit representation, a bit flip is a Pauli $X$, so a circuit with any fixed set of inserted faults is again a stabilizer circuit.
Let $\mathcal{F}$ be a finite set of allowed elementary faults with label set $\mathsf{F}$, $\alpha\mapsto f_\alpha$.
A \emph{fault configuration} is a bitvector $v\in\Ftwo^{\mathsf{F}}$ selecting the faults that occur, and $\Phi_v$ denotes the channel of the circuit with those faults inserted in causal order.
The \emph{reference fault set} $\mathcal{F}_0$ consists of Pauli $X$ and $Z$ insertions at every chosen quantum fault location and $X$ flips at every chosen classical fault location; general Pauli faults are products of its elements.

\subsection{Detector error models}
\label{sec:preliminaries-dems}

We recall the detector error model of a closed experiment~\cite{Gidney2021stimfaststabilizer,Beverland2024faulttolerance,Derks2025designingfault}.
Treat each allowed fault as a Pauli conjugation of the circuit channel.
By Lemma~\ref{lem:pauli-covariance} and Eq.~\eqref{eq:classical-interface-stabilizers}, its only effect on the declared outcomes is a bit flip, so there is a matrix $A_{\mathsf{M}}\in\Ftwo^{\mathsf{M}\times\mathsf{F}}$ whose column $\alpha$ is the outcome flip produced by $f_\alpha$, and a fault configuration $v$ produces the outcome flip $\Delta m=A_{\mathsf{M}}v$.
This is the closed case of the raw fault-effect map of Sec.~\ref{sec:gadget-rfe}.
The declared \emph{channel-check matrix} $A_{\mathsf{D}}$~\cite{Beverland2024faulttolerance} and \emph{observable-flip matrix} $A_{\mathsf{O}}$, and the \emph{bitmatrix part} of the detector error model, are
\begin{equation}
    \begin{aligned}
        A_{\mathsf{D}}
        &:=(\Delta\mathrm{DM})A_{\mathsf{M}},
        &
        A_{\mathsf{O}}
        &:=(\Delta\mathrm{OM})A_{\mathsf{M}},
        \\
        M_{\mathrm{DEM}}
        &:=
        \begin{bmatrix}
            A_{\mathsf{D}} \\
            A_{\mathsf{O}}
        \end{bmatrix}
        &&\in
        \Ftwo^{(\mathsf{D}\sqcup\mathsf{O})\times\mathsf{F}},
    \end{aligned}
    \label{eq:detector-error-model-map}
\end{equation}
so that $v$ flips the detectors $A_{\mathsf{D}}v$ and the observables $A_{\mathsf{O}}v$.
In the \emph{independent stochastic Pauli noise model}, each fault $f_\alpha$ occurs independently with probability $p_\alpha$.
A \emph{detector error model} (DEM) is the bitmatrix $M_{\mathrm{DEM}}$ together with the probability vector $p\in[0,1]^{\mathsf{F}}$, equivalently the sparse list of mechanisms $(p_\alpha,(A_{\mathsf{D}})_{:,\alpha},(A_{\mathsf{O}})_{:,\alpha})_{\alpha\in\mathsf{F}}$~\cite{Gidney2021stimfaststabilizer,Derks2025designingfault}.
Faults with identical columns are customarily merged into one mechanism whose probability equals the probability that an odd number of those faults occur,
\begin{equation}
    p_{\mathrm{odd}}(\mathsf{U})
    =
    \frac{1-\prod_{\alpha\in\mathsf{U}}(1-2p_\alpha)}{2}.
    \label{eq:equivalent-fault-probability}
\end{equation}
The outcome flip assigned to a fault is unique only up to flips of uniformly random outcome bits; since deterministic parities do not depend on such bits, $A_{\mathsf{D}}$, and $A_{\mathsf{O}}$ for deterministic observables, are intrinsic to the experiment.

\subsection{Gadgets and gadget circuits}
\label{sec:preliminaries-gadgets}

\begin{samepage}
A gadget connects a logical action with its physical realization in a way that supports composition; the definition is implicit in the logical-action verification problem of Ref.~\cite{Kliuchnikov2023stabilizercircuitverification} and, with input and output code blocks, in Ref.~\cite{DEQ}.
\begin{definition}[Gadget]
\label{def:gadget}
A \emph{gadget}~(Fig.~\ref{fig:gadget}) consists of two stabilizer circuits, an \emph{action} (a logical circuit) and a \emph{realization}, zero or more input and output stabilizer codes, each with its fixed generators and logical representatives, and an affine \emph{outcome map} $\mathrm{OM}$ from the realization outcomes to the action outcomes, whose values are the gadget's logical observables.
Two conditions must hold.
First, the realization sandwiched between the encoders of its input codes and the unencoders of its output codes~(Sec.~\ref{sec:preliminaries-codes}), with the encoders' syndromes set to zero, is equivalent to the action under the outcome map~(Def.~\ref{def:circuit-equivalence}).
Second, in the sandwiched realization, the unencoders' syndromes are redundant given the realization outcomes and the encoders' syndromes, and are zero when encoders' syndromes are zero.
\end{definition}
\end{samepage}

\begin{figure}[!tbp]
    \centering
    \begin{subfigure}[b]{0.70\linewidth}
        \centering
        \includegraphics[height=9.5em,width=\linewidth,keepaspectratio]{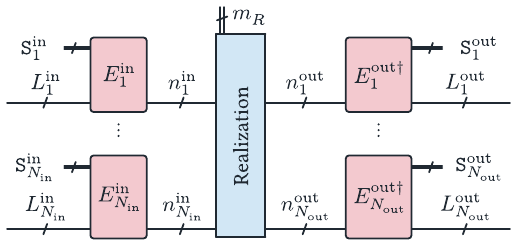}
        \caption{}
        \label{fig:realization-sandwich}
    \end{subfigure}
    \hfill
    \begin{subfigure}[b]{0.27\linewidth}
        \centering
        \includegraphics[height=9.5em,width=\linewidth,keepaspectratio]{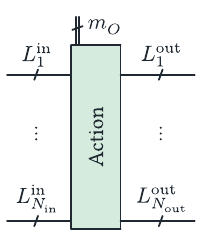}
        \caption{}
        \label{fig:action}
    \end{subfigure}
    \caption{A gadget, shown as a realization sandwiched between encoders and unencoders in subfigure~\subref{fig:realization-sandwich} and an equivalent logical action in subfigure~\subref{fig:action}.
    The outcomes $m_O$ are an affine function of $m_R$; this function is the outcome map.}
    \label{fig:gadget}
\end{figure}
The correctness of a gadget can be checked with the algorithms of Ref.~\cite{Kliuchnikov2023stabilizercircuitverification}, as Sec.~\ref{sec:gadget-profiles} discusses.
Gadgets compose when connected output and input codes match, that is, have the same generators and logical representatives; Sec.~\ref{sec:gadget-dem-composition} additionally requires matching code presentations.
We call the register carried by a gadget port a \emph{code block}.
\begin{definition}[Gadget circuit]
\label{def:gadget-circuit}
A \emph{gadget circuit} $\mathfrak{G}$ over a set of gadgets is a finite acyclic network in the sense of Definition~\ref{def:stabilizer-circuit} in which each box, called a \emph{gadget call}, is labelled by a gadget $g$, its ports are bound to the input and output codes of $g$, and every wire connects an output code to a matching input code.
The action and realization of its \emph{composite gadget} $G$ are the two stabilizer circuits obtained by wiring together the actions and the realizations, respectively, of its gadget calls.
The input and output codes of $G$ are the codes at the boundary ports of $\mathfrak{G}$, and its outcome map is
\begin{equation}
    \mathrm{OM}_G=\bigoplus_{g}\mathrm{OM}_g .
    \label{eq:composite-outcome-map}
\end{equation}
\end{definition}

Composition retains the full outcome maps, including random logical outcomes.
Let $\mathrm{DM}_{\mathrm{act}(G)}$ collect selected detectors of the composite gadget's logical action.
The corresponding \emph{deterministic outcome map} is
\begin{equation}
 \mathrm{OM}^{\mathrm{det}}_G
 = \mathrm{DM}_{\mathrm{act}(G)} \circ \mathrm{OM}_G .
 \label{eq:deterministic-outcome-map}
\end{equation}
Composition can introduce deterministic parities among individually random logical outcomes, as for the parity of the two logical readouts in a Bell experiment.

\subsection{Virtual syndromes and virtual detectors}
\label{sec:preliminaries-virtual-detectors}

\begin{figure}[ht]
    \centering
    \includegraphics[height=4em]{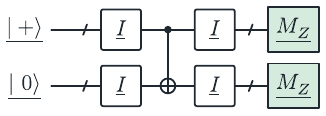}
    \caption{Physical realizations of logical Bell-pair preparation and product measurement, the running example of Secs.~\ref{sec:cutting-detectors} and~\ref{sec:gadget-dems}.
    The gadgets realize initialization in $\ket{0}$ and $\ket{+}$, syndrome extraction (a logical identity), a CNOT, and destructive logical-$Z$ measurement.
    Declared classical outputs, including syndrome-extraction and transversal-measurement outcomes, are omitted.}
    \label{fig:bell-pair}
\end{figure}

We review the modularization of detectors introduced in Ref.~\cite{DEQ, Beverland2024faulttolerance}.
Consider a physical circuit partitioned into the realizations of logical operations, for instance, the Bell-pair experiment of Fig.~\ref{fig:bell-pair}.
At every boundary between two operations, insert an unencoder immediately followed by an encoder that takes the unencoder's syndrome as its syndrome input.
By Eq.~\eqref{eq:unencoder}, this pair implements a nondestructive measurement of the code generators whose outcome is fed forward.
For the Pauli faults considered here, inserting this pair yields an equivalent circuit~(Def.~\ref{def:circuit-equivalence}).
We call the pair \emph{virtual} because it is not executed.
Its outcomes, the code syndrome at the boundary, are the \emph{output virtual syndromes} of the preceding operation and the \emph{input virtual syndromes} of the following one.
Detectors of an operation sandwiched between its virtual encoders and unencoders may involve these bits; we call them \emph{virtual detectors}.
The detectors of the full circuit are recovered by substituting, for every input virtual syndrome, the expression for the matching output virtual syndrome in terms of physical outcomes.
In the scheme of Ref.~\cite{DEQ}, a virtual detector that contains an output virtual syndrome bit contains exactly one such bit, so that it expresses that bit as an affine function of the operation's outcomes and input virtual syndromes; these functions form the \emph{output virtual syndrome map}, and substitution then proceeds in causal order.
Section~\ref{sec:virtual-detectors} shows where this per-operation scheme falls short and how the construction in Sec.~\ref{sec:cutting-detectors-formal-foundations} addresses its limitations.

\begin{figure*}[ht]
    \centering
    \begin{subfigure}[b]{0.3\textwidth}
        \centering
        \includegraphics[height=16em]{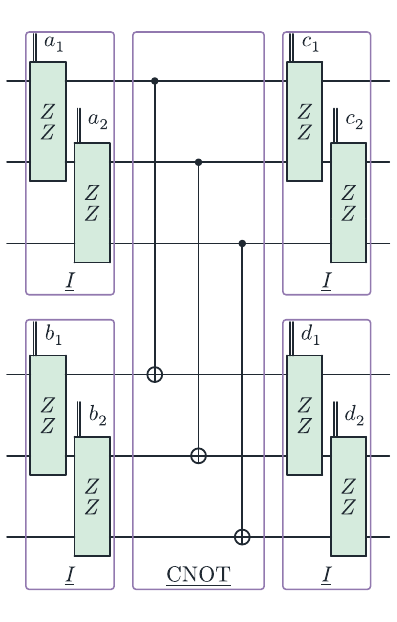}
        \caption{}
        \label{fig:repetition-code-circuit}
    \end{subfigure}
    \hfill
    \begin{subfigure}[b]{0.4\textwidth}
        \centering
        \includegraphics[height=16em]{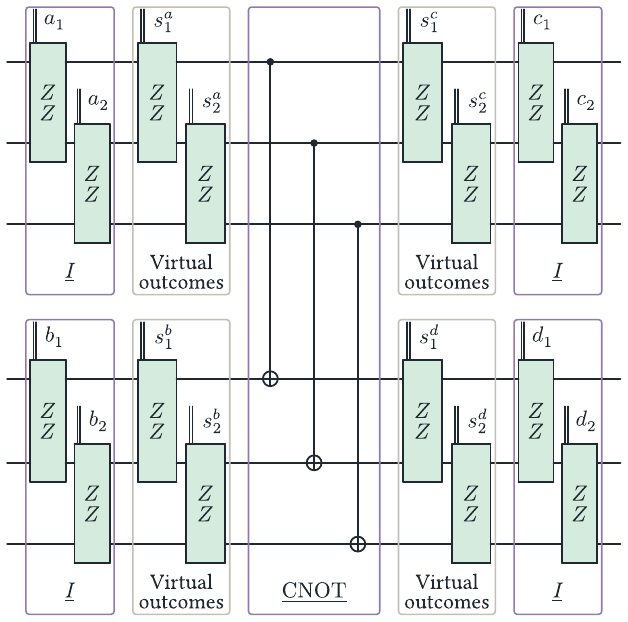}
        \caption{}
        \label{fig:repetition-code-virtual-outcomes}
    \end{subfigure}
    \hfill
    \begin{subfigure}[b]{0.27\textwidth}
        \centering
        \includegraphics[height=16em]{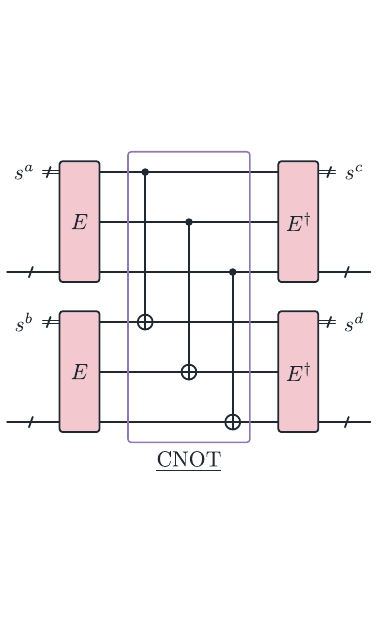}
        \caption{}
        \label{fig:cnot-encode-sandwich}
    \end{subfigure}
    \caption{A transversal CNOT between two blocks of the three-qubit repetition code, preceded by the syndrome-extraction rounds $a$ (control block) and $b$ (target block) and followed by the rounds $c$ and $d$; each round implements the logical identity.
    Subfigure~\subref{fig:repetition-code-circuit} shows the physical circuit partitioned into logical operations,
    subfigure~\subref{fig:repetition-code-virtual-outcomes} adds virtual syndromes at the boundaries,
    and subfigure~\subref{fig:cnot-encode-sandwich} shows the transversal CNOT sandwiched between virtual unencoders and encoders, with input virtual syndromes $s^a,s^b$ and output virtual syndromes $s^c,s^d$.}
    \label{fig:repetition-code-example}
\end{figure*}

Figure~\ref{fig:repetition-code-example} illustrates the scheme.
The code has the checks $Z_1Z_2$ and $Z_2Z_3$, indexed by $j\in\{1,2\}$; round $x\in\{a,b,c,d\}$ measures them with outcomes $x_j$, has input virtual syndromes $s^{x,\mathrm{in}}_j$, and provides the virtual detectors $x_j+s^{x,\mathrm{in}}_j$ together with the output virtual syndrome map $s^{x}_j=x_j$.
The CNOT realization has no physical outcomes; since it maps $Z$ on the target block to $Z\otimes Z$ and leaves $Z$ on the control block unchanged, its output virtual syndrome map is $s^c_j=s^a_j$ and $s^d_j=s^a_j+s^b_j$.
Substituting along the wires turns the virtual detectors of rounds $c$ and $d$ into the detectors
\begin{equation*}
    c_j+a_j,
    \qquad
    d_j+a_j+b_j,
\end{equation*}
of the full circuit.
Their sum $b_j+c_j+d_j$ is a detector supported on three rounds; we revisit it in Sec.~\ref{sec:virtual-detectors}.

\subsection{Diagram conventions for linear and affine maps}
\label{sec:preliminaries-diagrams}

We introduce the following conventions for maps with labelled bitvector input and output ports, closely related to~\cite{Lafont2003}.
Uppercase typewriter labels name classical bitvector ports, while lowercase mathematical symbols denote the values they carry; for example, $m$ is a value at port $\texttt{M}$ and $\Delta m$ is a value at port $\Delta\texttt{M}$.
A \emph{map diagram} is a network in the sense of Definition~\ref{def:stabilizer-circuit} whose boxes are linear or affine maps over $\Ftwo$ between labelled bitvector ports; composing the maps along the wires yields a single map $\interpret{\mathfrak{D}}$ from the boundary inputs to the boundary outputs, and this computation is called \emph{contracting} the diagram.
Inputs are drawn as wires with incoming arrows at the top and left of a box, and outputs are drawn as wires with outgoing arrows at the bottom and right.
Horizontal wires carry \emph{interface ports}, which are connected between boxes when diagrams are composed, and vertical wires carry \emph{free ports}, which survive contraction.
The block matrix of a contracted diagram has one block per pair of boundary input and output ports, and each block is a polynomial in the blocks of the component maps, with a zero block whenever no path connects the two ports (Sec.~\ref{sec:symbolic-construction}).

The map diagrams in this paper are obtained from gadget circuits by relabelling.
\begin{definition}[Map diagram of a gadget circuit]
\label{def:map-diagram-of-gadget-circuit}
Let $X$ assign to every gadget $g$ a linear or affine map $X_g$ with labelled bitvector ports.
For each input code of $g$, the assignment designates a list of input ports of $X_g$; for each output code, it designates a list of output ports.
These lists are the \emph{interface ports} of $X_g$; they stand for the code blocks at the quantum input and output ports of $g$, and the two lists assigned to matching codes have the same length.
All other ports of $X_g$ are \emph{free}.
For a gadget circuit $\mathfrak{G}$~(Def.~\ref{def:gadget-circuit}), the map diagram $X[\mathfrak{G}]$ is obtained by replacing every gadget call $g$ with the box $X_g$ and every code-block wire of $\mathfrak{G}$, which connects an output code of one call to an input code of another, with wires joining the output list of the first box to the input list of the second, entry by entry.
The free ports of all boxes, together with the lists for the code blocks at the boundary of $\mathfrak{G}$, are the boundary ports of $X[\mathfrak{G}]$.
\end{definition}
We consider three such assignments, each turning a gadget circuit into a map diagram whose contraction describes the composite gadget.
\begin{itemize}
    \item \emph{Detector maps} $\mathrm{DM}_g$~(Sec.~\ref{sec:detector-maps}, Def.~\ref{def:detector-map}): the list for a code is its virtual-syndrome port; measurement outcomes and detector values are free.
    The detector map of the composite gadget is defined as the contraction $\mathrm{DM}_G:=\interpret{\mathrm{DM}[\mathfrak{G}]}$, which substitutes output virtual syndromes into input virtual syndromes as in Sec.~\ref{sec:preliminaries-virtual-detectors} (Fig.~\ref{fig:extended-composed}).
    \item \emph{Raw fault-effect maps} $\mathrm{RFE}_g$~(Sec.~\ref{sec:gadget-rfe}): the list for a code is its boundary-error port; fault configurations and outcome flips are free.
    The raw fault-effect map of the composite gadget is defined as the contraction $\mathrm{RFE}_G:=\interpret{\mathrm{RFE}[\mathfrak{G}]}$, Eq.~\eqref{eq:rfe-composition}, after checking that the contraction is a valid raw fault-effect map (Fig.~\ref{fig:rfe-composisiton}).
    \item \emph{Extended detector error models} $\mathrm{EDEM}_g$~(Sec.~\ref{sec:gadget-dem-definition}, Def.~\ref{def:raw-extended-detector-error-model}): the list for a code consists of its virtual-syndrome-flip port and its boundary-error port, or the boundary-error coordinates selected by the code presentation; fault configurations, detector flips, and observable flips are free.
    Here the contraction is not a definition but a result: $\mathrm{EDEM}_G=\interpret{\mathrm{EDEM}[\mathfrak{G}]}$ is Theorem~\ref{thm:edem-composition} in Sec.~\ref{sec:gadget-dem-composition}, and Sec.~\ref{sec:symbolic-construction} evaluates such contractions symbolically.
\end{itemize}

\section{Decomposing detectors}
\label{sec:cutting-detectors}

This section develops tools for decomposing detectors across subcircuits of stabilizer circuits.
We first establish formal detector-decomposition results, then relate them to the virtual detectors of Sec.~\ref{sec:preliminaries-virtual-detectors}.
We conclude by extending the detector map of Sec.~\ref{sec:preliminaries-detectors}, which represents a list of detectors as an affine map, to gadgets.
For gadgets with open code boundaries, these maps carry virtual-syndrome inputs and outputs; contracting the detector-map diagram of a gadget circuit yields the detector map of its composite gadget.

\subsection{Formal foundations}
\label{sec:cutting-detectors-formal-foundations}

To decompose the detectors of a serial composition of two stabilizer circuits $C_1$ and $C_2$ at the boundary between them, we impose three requirements on the boundary.
(i) The connecting qubits carry no stabilizer of the output of $C_1$, including any with a deterministic sign.
(ii) The outcome-dependent signs of the stabilizers removed from the connecting qubits are instead recorded as classical bits $s$, which are affine functions of the outcomes $m_1$ of $C_1$.
(iii) Only independent such bits are recorded, so that no detector is supported on $s$.
An unencoder of the code at the boundary~(Sec.~\ref{sec:preliminaries-codes}) satisfies these requirements except in the two situations discussed in Sec.~\ref{sec:virtual-detectors}.
Definition~\ref{def:compressor} formalizes (i)--(iii) by introducing a \emph{compressor} and pairing it with an \emph{expander} that reverses its action.

\begin{definition}
\label{def:compressor}
Given a stabilizer circuit $C_1$ with $n$ output qubits and outcome vector $m_1$~(Fig.~\ref{fig:circuit-plain}),
a \emph{$C_1$-compressor} is a stabilizer circuit that applies a Clifford unitary $E^\dagger$ to the $n$ qubits and then measures $n_p+n_s$ of them destructively in the $Z$ basis, obtaining $n_p$ deterministic outcomes $0$ and an outcome bitvector $s\in\Ftwo^{n_s}$~(Fig.~\ref{fig:circuit-with-compressor}), such that the circuit consisting of $C_1$ followed by the compressor
has a trivial output stabilizer group~(Sec.~\ref{sec:stabilizer-maps}) on its $n_Q=n-n_p-n_s$ remaining qubits,
the compressor outcome is an affine function $s=\mathrm{CM}(m_1)$ of $m_1$,
and no nontrivial detector is supported on $s$.
The \emph{$C_1$-expander} is the compressor run in reverse: it takes $s$ as a declared classical input, prepares $n_p$ fresh qubits in $\ket{\boldsymbol 0}$ and $n_s$ fresh qubits in $\ket{s}$ using $X$ gates controlled by the bits of $s$, and applies $E$ to these qubits and the $n_Q$ remaining qubits~(Fig.~\ref{fig:circuit-with-compressor}).
\end{definition}

\begin{figure}[t]
    \centering
    \begin{subfigure}[b]{0.23\linewidth}
        \centering
        \raisebox{0.9324em}[8em][0pt]{\includegraphics[height=8em]{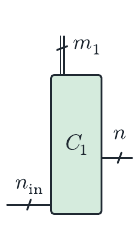}}
        \caption{}
        \label{fig:circuit-plain}
    \end{subfigure}
    \hfill
    \begin{subfigure}[b]{0.75\linewidth}
        \centering
        \includegraphics[height=8em]{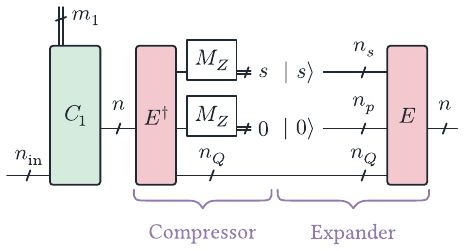}
        \caption{}
        \label{fig:circuit-with-compressor}
    \end{subfigure}
    \caption{(\subref{fig:circuit-plain}) A stabilizer circuit $C_1$ with $\Nin$ input qubits, $n$ output qubits, and outcome vector $m_1$.
    (\subref{fig:circuit-with-compressor}) An equivalent stabilizer circuit consisting of $C_1$ followed by a compressor and expander (Def.~\ref{def:compressor}, Lemma~\ref{lem:compressor-exists}).
    The compressor applies the Clifford unitary $E^\dagger$ and measures $n_p+n_s$ qubits in the $Z$ basis, obtaining the deterministic outcomes $0$ and the outcomes $s$; the expander prepares $n_p+n_s$ fresh qubits in $\ket{\boldsymbol 0}$ and $\ket{s}$ and applies $E$.}
    \label{fig:circuit-and-its-compressor}
\end{figure}

A compressor minimizes the classical and quantum information present at the circuit output without losing information.
Informally, the definition ensures that $n_Q$ is minimal and, subject to that condition, that $n_s$ is also minimal.
If the output stabilizer group on the $n_Q$ qubits were nontrivial, then we could decrease $n_Q$ by one and increase either $n_s$ or $n_p$ by one, depending on whether the sign of the additional stabilizer depends on $m_1$.
If there were a nontrivial detector supported on $s$, then we could decrease $n_s$ by one and increase $n_p$ by one.

\begin{lemma}
\label{lem:compressor-exists}
Every stabilizer circuit $C_1$ has a $C_1$-compressor~(Def.~\ref{def:compressor}), and $C_1$ followed by a $C_1$-compressor and its expander is equivalent to $C_1$~(Def.~\ref{def:circuit-equivalence}) under the outcome map that retains $m_1$ and discards the compressor outcomes.
\end{lemma}
\begin{proof}
The construction separates stabilizers with fixed signs from those whose signs must be recorded.
Let $G$ be the output stabilizer group of $C_1$ on its $n$ output qubits~(Sec.~\ref{sec:stabilizer-maps}), and let $G_0$ be its subgroup of stabilizers whose signs are constant across all attainable outcomes $m_1$.
Choose $n_p$ independent generators of $G_0$ and Paulis $P_1,\ldots,P_{n_s}$ whose cosets form a basis of $G/G_0$.
Take a Clifford unitary $E$ such that $E^\dagger$ maps the chosen generators of $G_0$, taken with their fixed signs, to $Z_1,\ldots,Z_{n_p}$ and maps $P_i$ to $Z_{n_p+i}$ for $i=1,\ldots,n_s$.
Measuring these $n_p+n_s$ qubits in the $Z$ basis yields the deterministic outcomes $0$ and the outcomes $s$, whose bits are the signs of the $P_i$ and hence affine functions of $m_1$.

On each outcome branch, the measured qubits are already in the state $\ket{\boldsymbol 0}\otimes\ket{\mathrm{CM}(m_1)}$ after applying $E^\dagger$ and are therefore uncorrelated with the remaining qubits.
Destructively measuring them and then preparing the same state leaves the branch unchanged.
Applying $E$ restores the original output, proving the claimed equivalence after discarding the redundant compressor outcomes.

It remains to check that neither the quantum output nor the classical record contains a redundant degree of freedom.
If a Pauli $P_Q$ stabilized the remaining $n_Q$ qubits with a sign affine in $(m_1,s)$, substituting $s=\mathrm{CM}(m_1)$ would make that sign affine in $m_1$.
Thus $E(I\otimes P_Q)E^\dagger$ would belong to $G$.
But $E^\dagger$ maps every element of $G$ to a product of $Z_1,\ldots,Z_{n_p+n_s}$, which acts trivially on the remaining qubits, so $P_Q=I$.
Likewise, a nontrivial detector $\langle w,s\rangle+c$ supported on $s$ would give the product $\prod_i P_i^{w_i}$ a fixed sign.
This product would then belong to $G_0$, contradicting the independence of the chosen cosets.
Hence all three compressor properties hold.
\end{proof}

Compressors and expanders are not unique, but Lemma~\ref{lem:compressor-exists} guarantees their existence.
We now use them to cut detectors across the serial composition of $C_1$ and $C_2$~(Fig.~\ref{fig:serial-open}).

\begin{figure*}[t]
    \centering
    \begin{subfigure}[b]{0.35\textwidth}
        \centering
        \raisebox{1.14486em}[10em][0pt]{\includegraphics[height=10em]{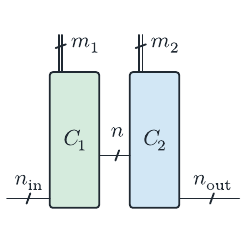}}
        \caption{}
        \label{fig:serial-open-circuit}
    \end{subfigure}
    \hfill
    \begin{subfigure}[b]{0.64\textwidth}
        \centering
        \includegraphics[height=10em]{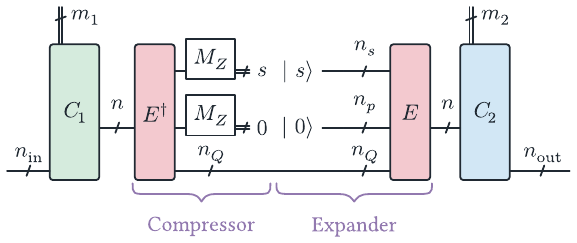}
        \caption{}
        \label{fig:serial-open-circuit-with-compressor}
    \end{subfigure}
    \caption{(\subref{fig:serial-open-circuit}) A two-part open stabilizer circuit with $\Nin$ input qubits and $n_{\mathrm{out}}$ output qubits, composed of stabilizer circuits $C_1$ and $C_2$ connected
    by $n$ qubits and having outcome vectors $m_1$ and $m_2$.
    (\subref{fig:serial-open-circuit-with-compressor}) An equivalent circuit with a $C_1$-compressor and $C_1$-expander inserted between the two parts.}
    \label{fig:serial-open}
\end{figure*}

Detectors are affine expressions in the outcomes that vanish on every attainable outcome~(Def.~\ref{def:detector}); for the expander followed by $C_2$, the free classical input $s$ counts as a variable.
One direction of the decomposition follows directly from the composition of boundary constraints~(Sec.~\ref{sec:stabilizer-maps}).
A detector of $C_1$ is a detector of the composed circuit, and a detector $\langle v_0,s\rangle+\langle v_2,m_2\rangle+b$ of the expander followed by $C_2$ becomes the detector $\langle v_0,\mathrm{CM}(m_1)\rangle+\langle v_2,m_2\rangle+b$ of the composed circuit, because the compressor records $s=\mathrm{CM}(m_1)$ and inserting it together with the expander is an equivalence~(Lemma~\ref{lem:compressor-exists}).

We now prove the converse, stated in Lemma~\ref{lem:cutting-detectors} below.
The proof relies on Theorem~\ref{thm:boundary-observable}, which factors detectors across a cut without using a compressor and requires only that $C_2$ be a stabilizer circuit.
In terms of channels and observables, the theorem states that an observable on the outcome registers that stabilizes the output of a serial composition factors through a Pauli $P$ at the cut, which the first channel prepares and the second measures.
The Pauli $P$ is obtained by pulling the observable back through the second channel using its adjoint~(Sec.~\ref{sec:stabilizer-maps}).

\begin{theorem}[Stabilized observables factor across a serial composition]
\label{thm:boundary-observable}
Consider a channel $\Phi_1$ from a register $A$ to registers $\mathsf{M}_1\otimes B$, a stabilizer channel $\Phi_2$ from $B$ to $\mathsf{M}_2\otimes C$, and their serial composition $\Phi:=(\operatorname{id}_{\mathsf{M}_1}\otimes\Phi_2)\circ\Phi_1$ from $A$ to $\mathsf{M}_1\otimes\mathsf{M}_2\otimes C$.
Let $Z_{\mathsf{M}_1}$ and $Z_{\mathsf{M}_2}$ be signed $Z$-type observables on $\mathsf{M}_1$ and $\mathsf{M}_2$ such that $Z_{\mathsf{M}_1}\otimes Z_{\mathsf{M}_2}\otimes I_C$ stabilizes the output of $\Phi$, that is, $(Z_{\mathsf{M}_1}\otimes Z_{\mathsf{M}_2}\otimes I_C)\Phi(\rho)=\Phi(\rho)$ for all operators $\rho$.
Then $P:=\Phi_2^\dagger(Z_{\mathsf{M}_2}\otimes I_C)$ is a Pauli observable on $B$ such that, for all operators $\rho$,
\begin{equation*}
\begin{aligned}
(Z_{\mathsf{M}_1}\otimes P)\Phi_1(\rho) &= \Phi_1(\rho), \\
(Z_{\mathsf{M}_2}\otimes I_C)\Phi_2(P\rho) &= \Phi_2(\rho).
\end{aligned}
\end{equation*}
\end{theorem}
\begin{proof}
Pulling the stabilized observable back to the input gives the identity.
Using Corollary~\ref{cor:dual-constraints}~(Appendix~\ref{app:further-math-preliminaries}) and the composition rule for adjoints, we obtain
\begin{equation*}
\begin{aligned}
I_A
&=\Phi^\dagger(Z_{\mathsf{M}_1}\otimes Z_{\mathsf{M}_2}\otimes I_C) \\
&=\Phi_1^\dagger\bigl(Z_{\mathsf{M}_1}\otimes
\Phi_2^\dagger(Z_{\mathsf{M}_2}\otimes I_C)\bigr) \\
&=\Phi_1^\dagger(Z_{\mathsf{M}_1}\otimes P).
\end{aligned}
\end{equation*}
The stabilizer assumption on $\Phi_2$ enters here: by the dual statement of Lemma~\ref{lem:pauli-covariance}, $P$ is either zero or a Pauli up to a sign.
The displayed identity excludes zero, and preservation of Hermiticity makes $P$ a Hermitian Pauli, hence a unitary.
We can therefore apply Corollary~\ref{cor:dual-constraints} to both
$\Phi_1^\dagger(Z_{\mathsf{M}_1}\otimes P)=I_A$ and
$\Phi_2^\dagger(Z_{\mathsf{M}_2}\otimes I_C)=P$.
These yield the first and second channel constraints in the statement, respectively.
\end{proof}

In the language of Sec.~\ref{sec:preliminaries}, let $\mathsf{M}_j$ be the declared classical outcome register of $\Phi_j$ and $Z_{\mathsf{M}_j}=Z_{\mathsf{M}_j}(a_j,b_j)$.
The hypothesis of the theorem then says that $\langle a_1,m_1\rangle+\langle a_2,m_2\rangle+b_1+b_2$ is a detector of $\Phi$~(Def.~\ref{def:detector}), and its conclusion says that $(I_A,P,a_1,b_1)$ and $(P,I_C,a_2,b_2)$ are boundary constraints~(Def.~\ref{def:boundary-constraint}): $P$ stabilizes the output of $\Phi_1$ with the sign given by $Z_{\mathsf{M}_1}$, and $\Phi_2$ measures $P$ with the sign given by $Z_{\mathsf{M}_2}$.
A compressor records the signs of exactly these boundary observables in its outcomes $s$, which yields the detector decomposition.
Declared classical boundary inputs of the composed circuit are counted among the outcomes $m_1$ of $C_1$.
By Def.~\ref{def:detector}, a detector must vanish for every value of such an input, so replacing the input by a fair coin~(Sec.~\ref{sec:preliminaries-circuits}) whose value is recorded in $m_1$ changes neither the detectors nor the argument below.

\begin{lemma}[Cutting detectors]
\label{lem:cutting-detectors}
Consider the serial composition of two stabilizer circuits $C_1$ and $C_2$ with outcome vectors $m_1$ and $m_2$~(Fig.~\ref{fig:serial-open-circuit}), and fix a $C_1$-compressor and its expander~(Def.~\ref{def:compressor}).
Every detector of the composed circuit is the sum of a detector $d_1$ of $C_1$ and a detector $d_2$ of the expander followed by $C_2$~(Fig.~\ref{fig:serial-open-circuit-with-compressor}), with the compressor outcome $\mathrm{CM}(m_1)$ substituted for the classical input $s$ of the expander.
If the detector is nontrivial, at least one of $d_1$ and $d_2$ is nontrivial.
The decomposition is unique.
\end{lemma}
\begin{proof}
The key point is that, after compression, the boundary Pauli supplied by Theorem~\ref{thm:boundary-observable} must be a parity of the classical bits $s$.
Let $d=\langle a_1,m_1\rangle+\langle a_2,m_2\rangle+b$ be a detector, and insert the compressor and expander after $C_1$.
This preserves $d$ by Lemma~\ref{lem:compressor-exists}.
Let $\Phi_1$ be $C_1$ followed by the compressor and $\Phi_2$ the expander followed by $C_2$.
At the cut, regard $s$ as the classical register $\mathsf{S}$ passed to $\Phi_2$ together with the remaining qubits $Q$, while $m_1$ is retained in the outcome register $\mathsf{M}_1$.
Apply Theorem~\ref{thm:boundary-observable} with $Z_{\mathsf{M}_1}(a_1,b)$ and $Z_{\mathsf{M}_2}(a_2,0)$.
It gives a Pauli
\[
P:=\Phi_2^\dagger(Z_{\mathsf{M}_2}(a_2,0)\otimes I_C)
\]
that $\Phi_1$ prepares with sign $(-1)^{\langle a_1,m_1\rangle+b}$ and $\Phi_2$ measures with outcome $\langle a_2,m_2\rangle$.

Because $\mathsf{S}$ is a declared classical input of $\Phi_2$, the image of $\Phi_2^\dagger$ is invariant under dephasing on $\mathsf{S}$~(Eq.~\eqref{eq:classical-interface}).
Thus $P$ has only $I$ and $Z$ factors on $\mathsf{S}$, so write $P=(-1)^cP_Q\otimes Z_{\mathsf{S}}^w$ with $P_Q$ Hermitian.
On a branch with outcomes $(m_1,s)$, the preparation constraint then says that $P_Q$ stabilizes $Q$ with sign
$(-1)^{\langle a_1,m_1\rangle+b+\langle w,s\rangle+c}$.
The compressor leaves $Q$ with a trivial output stabilizer group, so $P_Q=I$, absorbing any overall sign into $c$.
Consequently, the boundary observable is just $P=(-1)^cI_Q\otimes Z_{\mathsf{S}}^w$.

The preparation and measurement constraints now give the two detectors
\begin{align*}
d_1&:=\langle a_1,m_1\rangle
      +\langle w,\mathrm{CM}(m_1)\rangle+b+c, \\
d_2&:=\langle w,s\rangle+\langle a_2,m_2\rangle+c.
\end{align*}
Here $d_1$ is affine in $m_1$ because $\mathrm{CM}$ is affine, and $d_2$ vanishes for every classical input $s$ and every quantum input of $\Phi_2$ by its measurement constraint.
Substituting $s=\mathrm{CM}(m_1)$ cancels the shared boundary parity and gives $d_1+d_2=d$.
If $d$ is nontrivial, the two summands cannot both be zero, so at least one is nontrivial.

For uniqueness, the essential fact is that $s$ is a free input of the expander.
If $d=d_1+d_2=d_1'+d_2'$ are two decompositions, then $d_2-d_2'$ has no $m_2$ part, so it is a detector of the expander followed by $C_2$ that depends on $s$ alone.
Such a detector vanishes for every value of $s$~(Def.~\ref{def:detector}) and is therefore zero.
Hence $d_2=d_2'$ and $d_1=d_1'$.
\end{proof}

The proof is constructive.
Propagating $Z_{\mathsf{M}_2}(a_2,0)\otimes I_C$ backwards through $C_2$ and the expander by a stabilizer-tableau computation gives $P$.
Reading off its classical parity $w$ and sign $c$ then determines $d_1$ and $d_2$.
Thus a detector basis chosen for the composed circuit can be decomposed detector by detector while preserving the basis supplied to the decoder.
This matters because decoding performance depends on the choice of basis.
The first two compressor properties make the decomposition possible: no outcome-dependent stabilizer remains on $Q$, and the recorded boundary parity is affine in $m_1$.
The third property makes the classical interface minimal: since its attainable values form an affine code~(Prop.~\ref{prop:outcome-form}), the absence of detectors on $s$ means that $s$ ranges over all of $\Ftwo^{n_s}$.

\begin{corollary}[Completeness]
\label{cor:complete-basis}
Let $B_1$ be a complete detector basis of $C_1$ and $B_2$ a complete detector basis of the expander followed by $C_2$, and let $B_2'$ be obtained from $B_2$ by substituting $\mathrm{CM}(m_1)$ for $s$.
Then $B_1\cup B_2'$ is a complete detector basis of the composed circuit.
\end{corollary}
\begin{proof}
Each element of $B_1\cup B_2'$ is a detector of the composed circuit by the composition of boundary constraints noted before Lemma~\ref{lem:cutting-detectors}.
That lemma also establishes that the union spans the detector space: decompose any detector as $d_1+d_2$, expand the two parts in $B_1$ and $B_2$, and substitute $s=\mathrm{CM}(m_1)$.

For independence, any linear relation among the elements of $B_1\cup B_2'$ gives a decomposition of the zero detector into a sum from $B_1$ and a substituted sum from $B_2$.
Uniqueness in Lemma~\ref{lem:cutting-detectors}, compared with the decomposition $0+0$, forces both sums to be zero before substitution.
Independence of $B_1$ and $B_2$ then forces every coefficient in the relation to vanish.
\end{proof}

\subsection{Virtual syndromes and virtual detectors}
\label{sec:virtual-detectors}

In the approach to modularizing detectors introduced in Ref.~\cite{DEQ} and reviewed in Sec.~\ref{sec:preliminaries-virtual-detectors}, an unencoder does not always correspond to a compressor.
This correspondence can fail in two common ways.
First, in violation of requirement (iii) of Sec.~\ref{sec:cutting-detectors-formal-foundations}, nontrivial detectors may be supported entirely on the output virtual syndrome bits.
For example, in logical-zero state preparation, the output syndrome is commonly fixed to zero, so each output virtual syndrome bit $s_j$ defines a detector $s_j = 0$.
We handle this case by setting the expression for $s_j$ to zero in the output virtual syndrome map.
Second, in violation of requirement (i), the output stabilizer group~(Sec.~\ref{sec:stabilizer-maps}) may contain a logical operator of the code.
A logical-zero preparation circuit again provides an example.
This case is addressed by detectors in the logical circuit.
For example, a logical circuit for a $Z$-memory experiment prepares the logical-zero state and measures $Z$.
This circuit has a detector requiring the $Z$-measurement outcome to be zero.
Whenever a logical operation produces logical measurement outcomes, its outcome map~(Def.~\ref{def:gadget}, Eq.~\eqref{eq:detector-and-logical-affine-maps}) expresses them as affine functions of the physical outcomes.
Substituting the outcome map into a detector of the logical circuit turns it into a detector of the physical circuit.

Another challenge for this virtual-detector approach is multiport composition.
In this setting, not every choice of global detectors admits a simple decomposition.
Consider the transversal CNOT flanked by syndrome-extraction rounds in Fig.~\ref{fig:repetition-code-example}.
Section~\ref{sec:preliminaries-virtual-detectors} displays its detector $b_j + c_j + d_j$, which is supported on three syndrome-extraction rounds.
Applying Lemma~\ref{lem:cutting-detectors} to the boundary following the CNOT gives $d_2 = s^c_j + s^d_j + c_j + d_j$ for the part after the cut.
The parity $b_j + s^c_j + s^d_j$ recorded before the cut vanishes identically after substituting $s^c_j=a_j$ and $s^d_j=a_j+b_j$, so the detector $d_1$ of the part before the cut is trivial.
The detector $s^c_j + s^d_j + c_j + d_j$ is split across two parallel gadgets and would not appear if we resolved output virtual syndromes in terms of past outcomes within individual gadgets.

Some quantum error-correcting codes have a natural overcomplete set of code generators, as in various versions of $d$-dimensional toric codes.
The unencoder--encoder picture requires choosing a linearly independent set of generators.
However, we can still define the syndrome using an overcomplete set of generators and
specify virtual detectors in terms of the resulting overcomplete virtual syndromes,
as already highlighted in Ref.~\cite{DEQ}.
The syndromes of the extra generators can be regarded as convenient intermediate variables.

\subsection{Composing the detector basis using affine maps}
\label{sec:detector-maps}

To compose detectors across gadgets, we use the detector map of Eq.~\eqref{eq:detector-and-logical-affine-maps}, which represents a list of detectors
as an affine map from a bitvector of measurement outcomes to a bitvector of
detector values, and extend it with virtual-syndrome ports.

\begin{figure*}[ht]
    \centering
    \begin{subfigure}[b]{0.185\textwidth}
        \centering
        \includegraphics[height=8em]{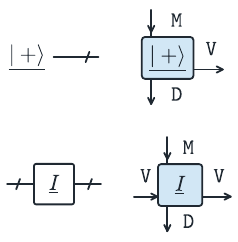}
        \caption{}
        \label{fig:extended-detector-prep-se}
    \end{subfigure}
    \hfill
    \begin{subfigure}[b]{0.175\textwidth}
        \centering
        \includegraphics[height=8em]{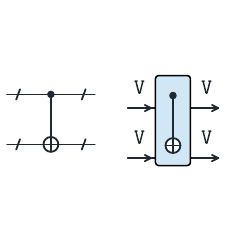}
        \caption{}
        \label{fig:extended-detector-cnot}
    \end{subfigure}
    \hfill
    \begin{subfigure}[b]{0.285\textwidth}
        \centering
        \includegraphics[height=8em]{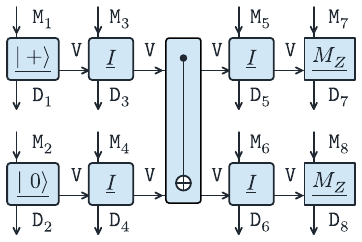}
        \caption{}
        \label{fig:extended-detector-wiring}
    \end{subfigure}
    \hfill
    \begin{subfigure}[b]{0.28\textwidth}
        \centering
        \includegraphics[height=8em]{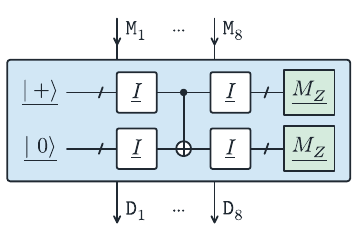}
        \caption{}
        \label{fig:extended-detector-composed}
    \end{subfigure}
    \caption{Detector maps with virtual-syndrome interfaces.
    Subfigures~\subref{fig:extended-detector-prep-se} and~\subref{fig:extended-detector-cnot} show gadgets and their detector maps $\mathrm{DM}_g$.
    Subfigure~\subref{fig:extended-detector-wiring} shows the detector-map diagram $\mathrm{DM}[\mathfrak{G}]$ of the Bell experiment in Fig.~\ref{fig:bell-pair},
    and subfigure~\subref{fig:extended-detector-composed} shows its contraction, the detector map $\mathrm{DM}_G$ of the composite gadget.}
    \label{fig:extended-composed}
\end{figure*}

A preliminary step in constructing an extended detector error model~(EDEM, Sec.~\ref{sec:gadget-dem-definition}) is to introduce the \emph{detector map}~(DM) of a gadget.

\begin{definition}
\label{def:detector-map}
The \emph{detector map} $\mathrm{DM}$ of a gadget extends the detector map of Eq.~\eqref{eq:detector-and-logical-affine-maps} with virtual-syndrome ports; it is an affine map with the following bitvector input and output ports.
The inputs are the gadget's measurement outcomes~(\texttt{M}) and the input virtual syndrome bitvectors~(\texttt{V}) for each input code block of the gadget.
The outputs are the detector values~(\texttt{D}) and the output virtual syndrome bitvectors~(\texttt{V}) for each output code block of the gadget.
Its linear part $\Delta\mathrm{DM}$ maps the corresponding measurement and virtual-syndrome flips to detector and output virtual-syndrome flips.
For a gadget without code boundaries, the virtual-syndrome ports are empty and the map reduces to the detector map of Eq.~\eqref{eq:detector-and-logical-affine-maps} from measurement outcomes to detector values.
\end{definition}

Figures~\ref{fig:extended-detector-prep-se} and~\ref{fig:extended-detector-cnot} show examples of detector maps for state preparation,
syndrome extraction, and a transversal CNOT.
For the detector map, the interface ports are the input and output virtual-syndrome ports, while measurement outcomes are free inputs and detector values are free outputs.
For a gadget circuit $\mathfrak{G}$, the detector map of its composite gadget is obtained by contracting the detector-map diagram $\mathrm{DM}[\mathfrak{G}]$~(Def.~\ref{def:map-diagram-of-gadget-circuit}),
in which the individual gadget detector maps are wired along the code blocks of $\mathfrak{G}$~(Fig.~\ref{fig:extended-detector-wiring}).
The contraction substitutes output virtual syndromes for input virtual syndromes and thereby turns every virtual detector into a detector of the composite gadget, as in Sec.~\ref{sec:preliminaries-virtual-detectors}.
We write $\mathrm{DM}_G=\interpret{\mathrm{DM}[\mathfrak{G}]}$.
After contraction, only the free ports and the virtual-syndrome ports at the boundary of $\mathfrak{G}$ remain~(Fig.~\ref{fig:extended-detector-composed}).
In the final detector map, detectors are naturally partitioned by the subcircuit in which they can be fully computed,
and measurement outcomes are partitioned by the subcircuit to which they belong.
When the partitions are indexed in temporal order, the matrix corresponding to the affine detector map is block lower triangular.

\paragraph{Success bits.}
\label{par:success-bits}
When considering fault-tolerant protocols with postselection, it is convenient to introduce
a special kind of detector called a \emph{success bit}.
A success bit is a declared deterministic parity whose violation causes the run to be discarded rather than decoded.
Success bits can be added to detector maps as free inputs and outputs to support postselection.
They are carried through the contraction as free ports rather than substituted along code blocks.

\section{Extended detector error model}
\label{sec:gadget-dems}

In this section we generalize detector error models to extended detector error models for gadgets with input and output qubits.
Section~\ref{sec:gadget-rfe} introduces boundary errors and the raw fault-effect map, which propagates faults and boundary errors through a gadget;
Sec.~\ref{sec:gadget-profiles} fixes the code presentations and gadget detector contracts that play the role of the detector basis;
and Sec.~\ref{sec:gadget-dem-definition} defines the extended detector error model.
We show that the extended detector error model of the composite gadget of a gadget circuit is the contraction of the EDEM diagram of that circuit,
in which the per-gadget extended detector error models are wired along the code blocks~(Sec.~\ref{sec:gadget-dem-composition}).
We also discuss optimizations that reduce the size of extended detector error models~(Secs.~\ref{sec:gadget-dem-definition} and~\ref{sec:gadget-dem-faults}).

\subsection{Boundary errors and raw fault-effect maps}
\label{sec:gadget-rfe}

Faults occur at different locations within a circuit, but a detector error model retains only their action on the circuit boundary.
To obtain this boundary description, pull the classical bit controlling each allowed fault $f_\alpha$~(Sec.~\ref{sec:preliminaries-faults}) back to an input port labelled by $\alpha\in\mathsf{F}$, so that the faulty circuit becomes a single stabilizer channel
\begin{equation}
    \begin{aligned}
    \Psi &:
    A\otimes\mathsf{F}
    \longrightarrow
    B\otimes\mathsf{M},
    \\
    \Phi_v(\rho)
    &=
    \Psi\bigl(\rho\otimes X_{\mathsf{F}}^{v}\proj{\boldsymbol{0}}_{\mathsf{F}}X_{\mathsf{F}}^{v}\bigr),
    \end{aligned}
\end{equation}
from the quantum input register $A$ and the fault ports to the quantum output register $B$ and the outcome register $\mathsf{M}$.
A fault configuration $v$ is then represented by the Pauli $X_{\mathsf{F}}^{v}$ on the fault ports.

Let $E_{\mathrm{in}}$ and $E_{\mathrm{out}}$ be the label spaces $\Ftwo^{2|A|}$ and $\Ftwo^{2|B|}$ of the phase-free Pauli groups on $A$ and $B$~(Sec.~\ref{sec:preliminaries-codes}), with Pauli representatives $P_A(e_{\mathrm{in}})$ and $P_B(e_{\mathrm{out}})$.
We call their elements \emph{boundary errors}; an output boundary error is interpreted relative to the ideal output on the same outcome branch.
\begin{definition}[Raw fault-effect map]
    \label{def:raw-fault-effect-map}
    A \emph{raw fault-effect map} compatible with $\Psi$ is a linear map
    \begin{equation}
        \mathrm{RFE}_{\Psi}:
        E_{\mathrm{in}}\oplus\Ftwo^{\mathsf{F}}
        \longrightarrow
        \Ftwo^{\mathsf{M}}\oplus E_{\mathrm{out}}
        \label{eq:raw-fault-effect-map}
    \end{equation}
    such that, whenever $\mathrm{RFE}_{\Psi}(e_{\mathrm{in}},v)=(\Delta m,e_{\mathrm{out}})$,
    \begin{equation}
        \Psi\circ
        \operatorname{Ad}_{P_A(e_{\mathrm{in}})\otimes X_{\mathsf{F}}^v}
        =
        \operatorname{Ad}_{P_B(e_{\mathrm{out}})\otimes X_{\mathsf{M}}^{\Delta m}}
        \circ\Psi.
        \label{eq:raw-fault-effect-compatibility}
    \end{equation}
    We call $\Delta m$ the \emph{outcome flip} and $e_{\mathrm{out}}$ the \emph{output boundary error} produced by $(e_{\mathrm{in}},v)$.
\end{definition}
Such a map exists by Lemma~\ref{lem:pauli-covariance}: for each basis element of $E_{\mathrm{in}}\oplus\Ftwo^{\mathsf{F}}$, the lemma supplies an output Pauli satisfying Eq.~\eqref{eq:raw-fault-effect-compatibility} up to a factor $Z_{\mathsf{M}}^{z}$.
This factor acts trivially on the declared classical output $\mathsf{M}$ by Eq.~\eqref{eq:classical-interface-stabilizers}.
Multiplying the chosen output Paulis extends the assignment linearly.
The \emph{raw fault-effect map of a gadget} $g$ is a raw fault-effect map $\mathrm{RFE}_g$ of its realization with the gadget's fault set $\mathcal{F}_g$, where $E_{\mathrm{in}}$ and $E_{\mathrm{out}}$ are the spaces of phase-free Pauli labels on its input and output code blocks.
For a gadget without code blocks, this map is the matrix $A_{\mathsf{M}}$ of Sec.~\ref{sec:preliminaries-dems}.

\begin{remark}[Choice of output representative]
    \label{rem:rfe-representatives}
    A raw fault-effect map is not unique.
    The output representative $P_B(e_{\mathrm{out}})\otimes X_{\mathsf{M}}^{\Delta m}$ may be changed by an output stabilizer~(Sec.~\ref{sec:stabilizer-channel-constraints}) or by a flip of an independently random outcome bit, since either acts trivially on the image of $\Psi$.
    These changes preserve the circuit semantics and the flips of all deterministic detector and observable parities.
    Dually, two input boundary errors that differ by a stabilizer of the code on that wire have the same downstream effect on codespace states, by Lemma~\ref{lem:encoder-pullback}.
    We call this freedom the \emph{choice of output representative}; every statement below holds for any fixed choice.
\end{remark}

Fault propagation through a gadget circuit is described by a composition of linear maps.
For raw fault-effect maps, the interface ports of a code are its boundary-error ports, and the fault and outcome-flip ports are free~(Def.~\ref{def:map-diagram-of-gadget-circuit}).
Applying Eq.~\eqref{eq:raw-fault-effect-compatibility} successively along the wires of a gadget circuit $\mathfrak{G}$ yields the same compatibility relation for its composite gadget $G$, with fault set $\bigsqcup_g\mathcal{F}_g$.
Thus, the contraction of the RFE diagram $\mathrm{RFE}[\mathfrak{G}]$ is a raw fault-effect map of $G$.
We define the raw fault-effect map of the composite gadget by this contraction,
\begin{equation}
    \mathrm{RFE}_G
    :=
    \interpret{\mathrm{RFE}[\mathfrak{G}]},
    \label{eq:rfe-composition}
\end{equation}
as illustrated in Fig.~\ref{fig:rfe-composisiton} of Sec.~\ref{sec:gadget-dem-composition} for the Bell-pair experiment of Fig.~\ref{fig:bell-pair}.

\subsection{Gadget detector contracts and code presentations}
\label{sec:gadget-profiles}

The standard detector error model is defined with respect to a choice of
detector basis, observables, and elementary faults.
The same three choices must be made to define the extended detector error model of a gadget, and gadget detector contracts and code presentations generalize the choice of detector basis.

\begin{definition}[Code presentation and gadget detector contract]
\label{def:code-profile}
A \emph{code presentation} of a stabilizer code is a possibly over-complete list of stabilizer generators together with the syndrome and logical-effect coordinates used at the code boundary:
for a boundary error $e$ with representative $P$, the \emph{syndrome coordinate} $s(e)$ is the commutation vector of $P$ with the listed generators, and the \emph{logical-effect coordinate} $l(e)$ is the phase-free label of the logical Pauli $\ell_P$ of $P$~(Lemma~\ref{lem:encoder-pullback}).
A \emph{gadget detector contract}, or \emph{gadget contract} for short, fixes the virtual detectors of the gadget, the affine expressions of its output virtual syndromes in terms of its measurement outcomes and input virtual syndromes, and the code presentations of its input and output codes; that is, it fixes a detector map~(Def.~\ref{def:detector-map}) together with the code presentations.
\end{definition}
Both coordinates depend only on the coset of $P$ modulo the stabilizer group, by Lemma~\ref{lem:encoder-pullback}, so they are well defined on boundary errors up to stabilizers.

For composition to work correctly, it is important that not only the input and output
stabilizer codes of connected ports match, but also that their code presentations match.

\paragraph{Contract correctness.}
\label{par:profile-correctness}
A gadget contract is \emph{correct} when every declared virtual detector and every declared output virtual syndrome relation is a deterministic parity of the realization sandwiched between the encoders and unencoders of its codes~(Fig.~\ref{fig:realization-sandwich}), that is, a check of the outcome code of the sandwiched circuit~(Prop.~\ref{prop:outcome-form}),
and when the declared detectors, output virtual syndrome relations, success bits~(Sec.~\ref{par:success-bits}) and deterministic logical outcomes together generate all deterministic parities of that circuit (rank condition).
Deterministic observables are the deterministic parities of the action outcomes, taken modulo detectors.
Both conditions are checked by Clifford simulation~\cite{Kliuchnikov2023stabilizercircuitverification}, as Sec.~\ref{sec:implementation-checks} describes.

\subsection{Extended detector error model}
\label{sec:gadget-dem-definition}

\begin{figure}
     \centering
        \includegraphics[height=7em]{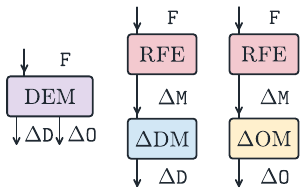}
        \caption{Bitmatrix part of the standard detector error model and its factorization through the raw fault-effect map~($\mathrm{RFE}$) and the linear parts $\Delta\mathrm{DM}$ and $\Delta\mathrm{OM}$ of the detector and outcome maps.}
        \label{fig:standard-detector-error-model}
\end{figure}

Recall that the bitmatrix part of the standard detector error model, Eq.~\eqref{eq:detector-error-model-map} in Sec.~\ref{sec:preliminaries-dems}, can be viewed as
a linear map from fault configurations to
detector and observable flips~(Fig.~\ref{fig:standard-detector-error-model})
for a gadget that has no inputs and no outputs.
It is obtained by composing the gadget's raw fault-effect map, which in this closed case is the matrix $A_{\mathsf{M}}$, with $\Delta\mathrm{DM}$ and $\Delta\mathrm{OM}$, as shown in Fig.~\ref{fig:standard-detector-error-model}.
EDEM composition propagates the full outcome map.
Deterministic observable-flip rows are subsequently extracted using the linear part of the composite logical action's detector map, as in Eq.~\eqref{eq:deterministic-outcome-map}.
An action outcome that is not deterministic, such as the random logical readout of Sec.~\ref{sec:implementation-examples}, is nevertheless a legitimate observable, whose row $\Delta\mathrm{OM}\,A_{\mathsf{M}}$ depends on the choice of output representative~(Remark~\ref{rem:rfe-representatives}).
Only the rows of deterministic parities are independent of that choice.

\begin{figure*}[ht]
    \centering
    \begin{subfigure}[b]{0.23\textwidth}
        \centering
        \includegraphics[height=13em]{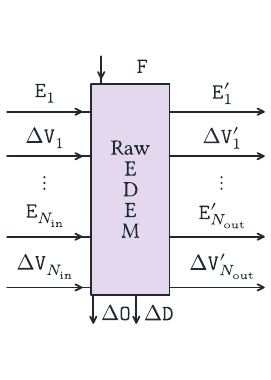}
        \caption{}
        \label{fig:extended-detector-error-model-diagram}
    \end{subfigure}
    \hfill
    \begin{subfigure}[b]{0.37\textwidth}
        \centering
        \includegraphics[height=13em]{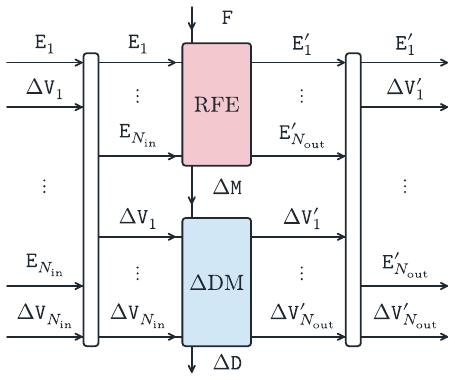}
        \caption{}
        \label{fig:extended-detector-error-model-detectors}
    \end{subfigure}
    \hfill
    \begin{subfigure}[b]{0.34\textwidth}
        \centering
        \includegraphics[height=13em]{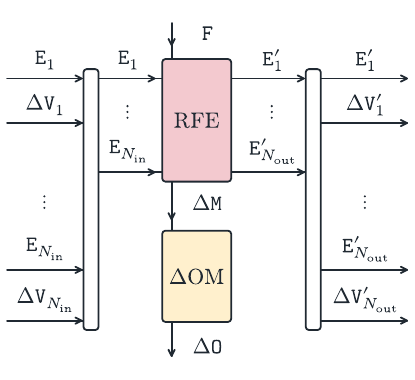}
        \caption{}
        \label{fig:extended-detector-error-model-observable}
    \end{subfigure}
    \caption{ Raw extended detector error model~(Def.~\ref{def:raw-extended-detector-error-model}).
    Subfigure~\ref{fig:extended-detector-error-model-diagram} shows a diagram with labelled inputs and outputs, subfigure~\ref{fig:extended-detector-error-model-detectors}
     relates the detector-flip output to the raw fault-effect map~($\mathrm{RFE}$, Def.~\ref{def:raw-fault-effect-map}) and the linear part $\Delta\mathrm{DM}$ of the detector map~(Def.~\ref{def:detector-map}), and
     subfigure~\ref{fig:extended-detector-error-model-observable}
     relates the observable-flip output to the raw fault-effect map and the linear part $\Delta\mathrm{OM}$ of the outcome map~(Def.~\ref{def:gadget}). Thin white blocks reorder wires. }
    \label{fig:raw-extended-detector-error-model}
\end{figure*}

The above viewpoint leads to a natural generalization to gadgets with inputs and outputs.
When composing gadgets in the diagrammatic representations such as Fig.~\ref{fig:raw-extended-detector-error-model}, we add horizontal input and output wires to pass virtual syndromes and boundary errors.
The virtual syndrome bits carry the signs of the stabilizers as expected by the providing gadget on the basis of its available classical information.
Boundary-error bits represent the residual Pauli effect of earlier faults on the output boundary, in the chosen output representative~(Remark~\ref{rem:rfe-representatives}).
\begin{definition}
\label{def:raw-extended-detector-error-model}
A raw extended detector error model for a gadget (with associated gadget contract and elementary faults) is a linear map with the following bitvector input and output ports~(Fig.~\ref{fig:raw-extended-detector-error-model}).
The inputs are a fault configuration~(\texttt{F}) and, for each gadget input code block, input virtual syndrome flips~($\Delta$\texttt{V}) and an input boundary error~(\texttt{E}).
The outputs are detector flips~($\Delta$\texttt{D}), declared observable flips~($\Delta$\texttt{O}), and, for each gadget output code block, output virtual syndrome flips~($\Delta$\texttt{V}) and an output boundary error~(\texttt{E}).
\end{definition}
Two structural properties are visible in Fig.~\ref{fig:raw-extended-detector-error-model}: the output boundary error, and hence its coordinates introduced below, depends only on the fault configuration and the input boundary error, because the raw fault-effect map has no virtual-syndrome input;
and the input virtual syndrome flips enter only the detector flips and the output virtual syndrome flips, through $\Delta\mathrm{DM}$.

The dimensions of the input and output wires of the raw extended detector error model can be further reduced by replacing boundary errors with coordinates selected by a code presentation.
\begin{definition}[Code-presentation map]
\label{def:code-profile-map}
The \emph{code-presentation map} $\mathrm{CPM}:\texttt{E}\to(\texttt{S},\texttt{L})$ of a code presentation is the linear map sending a boundary error $e$ to its syndrome and logical-effect coordinates $(s(e),l(e))$ of Definition~\ref{def:code-profile}.
A \emph{section} of $\mathrm{CPM}$ is a linear map $\sigma:\operatorname{im}(\mathrm{CPM})\to\texttt{E}$ with $\mathrm{CPM}\circ\sigma=\operatorname{id}$, which maps attainable coordinate pairs back to representative boundary errors.
\end{definition}
The code-presentation map is obtained by pulling physical $X$- and $Z$-Pauli errors back through the encoder~(Lemma~\ref{lem:encoder-pullback}) and mapping the resulting syndrome to the possibly over-complete generator list of the code presentation.
A convenient section sends the syndrome of each generator in an independent subset to a destabilizer of that generator, a Pauli anticommuting with it and commuting with the other generators, and each logical coordinate to the fixed logical representative of the code.
The \emph{extended detector error model} is the raw extended detector error model with its boundary-error ports compressed to boundary-error coordinates,
\begin{equation}
    \mathrm{EDEM}
    :=
    (\operatorname{id}\oplus\mathrm{CPM}_{\mathrm{out}})
    \circ
    \mathrm{EDEM}^{\mathrm{raw}}
    \circ
    (\operatorname{id}\oplus\sigma_{\mathrm{in}}),
    \label{eq:edem-compression}
\end{equation}
where $\mathrm{CPM}_{\mathrm{out}}$ and $\sigma_{\mathrm{in}}$ act on the boundary-error ports of the output and input codes and $\operatorname{id}$ on the remaining ports.
Inserting $\sigma\circ\mathrm{CPM}$ on a boundary-error wire leaves the contraction of a diagram unchanged up to the choice of output representative, because $e$ and $\sigma(\mathrm{CPM}(e))$ have the same coordinates and therefore differ by a stabilizer of the code on that wire~(Remark~\ref{rem:rfe-representatives}).
The raw form is simpler and more convenient for illustrating the key ideas; the implementation of Sec.~\ref{sec:examples-implementation} uses the compressed form.

For a gadget call $g$, write $b=(s,l)$ for the boundary coordinates on a wire and $\Delta v$ for the virtual-syndrome flips, with subscripts $\mathrm{in}$ and $\mathrm{out}$ collecting all input and output wires of $g$.
In block form, its EDEM is
\begin{equation}
    \begin{aligned}
        \Delta d_g &= D^g_F f_g + D^g_B b_{\mathrm{in}} + D^g_V \Delta v_{\mathrm{in}}, \\
        \Delta o_g &= O^g_F f_g + O^g_B b_{\mathrm{in}}, \\
        \Delta v_{\mathrm{out}} &= V^g_F f_g + V^g_B b_{\mathrm{in}} + V^g_V \Delta v_{\mathrm{in}}, \\
        b_{\mathrm{out}} &= B^g_F f_g + B^g_B b_{\mathrm{in}}.
    \end{aligned}
    \label{eq:edem-blocks}
\end{equation}
The virtual-syndrome interface $\Delta v$ records fault-induced changes in the stabilizer signs inferred from classical data, while the error interface $b=(s(e),l(e))$ records the syndrome and logical effect of the residual Pauli error $e$.
In particular, $\Delta v$ need not equal $s(e)$: a flipped syndrome-measurement bit can change $\Delta v$ without changing $e$, and a Pauli fault after the last check can change $e$ without changing $\Delta v$.
The absent blocks express the structural properties above: observable flips and outgoing boundary coordinates do not depend on incoming virtual-syndrome flips.

Tables~\ref{tab:gadget-dem-components} and~\ref{tab:gadget-dem-ports} summarize the linear maps and named port types used in this section.
As defined in Eq.~\eqref{eq:affine-map-linearization}, a label $\Delta\texttt{X}$ denotes the fault-induced change in the corresponding classical quantity $\texttt{X}$.
\begin{table*}[t]
    \centering
    \caption{Named classical bitvector port types, grouped by functionality.}
    \label{tab:gadget-dem-ports}
    \begingroup
    \footnotesize
    \setlength{\tabcolsep}{4pt}
    \begin{tabular}{@{}cp{0.18\textwidth}p{0.64\textwidth}@{}}
        \textcolor{quantumviolet}{\textbf{Abbrev.}}
        & \textcolor{quantumviolet}{\textbf{Full name}}
        & \textcolor{quantumviolet}{\textbf{Description}} \\
        \toprule
        \strut\textbf{\texttt{F}}
        & Fault configuration
        & A bitvector selecting a combination of elementary faults; the compressed variant uses a smaller coordinate set with the same attainable effects. \\
        \strut\textbf{\texttt{E}}
        & Boundary error
        & A phase-free Pauli label carrying the residual Pauli effect of earlier faults across a gadget boundary, up to the choice of output representative~(Remark~\ref{rem:rfe-representatives}). \\
        \strut\textbf{$(\texttt{S},\texttt{L})$}
        & Boundary-error\newline coordinates
        & The syndrome and logical-effect components selected by a code presentation; $(s,l)$ denotes a value of these ports. \\
        \strut\textbf{\texttt{V}, $\Delta$\texttt{V}}
        & Virtual\newline syndrome and flip
        & Stabilizer signs inferred from classical data, and their fault-induced changes. \\
        \strut\textbf{\texttt{M}, $\Delta$\texttt{M}}
        & Measurement\newline outcome and flip
        & A physical measurement outcome and its flip induced by faults and boundary errors. \\
        \strut\textbf{\texttt{D}, $\Delta$\texttt{D}}
        & Detector\newline value and flip
        & A declared detector parity and its fault-induced violation, obtained using $\mathrm{DM}$ and $\Delta\mathrm{DM}$, respectively. \\
        \strut\textbf{\texttt{O}, $\Delta$\texttt{O}}
        & Observable\newline value and flip
        & A declared logical-observable value and its fault-induced flip, obtained using $\mathrm{OM}$ and $\Delta\mathrm{OM}$, respectively. \\
        \bottomrule
    \end{tabular}
    \endgroup
\end{table*}

\begin{table*}[t]
    \centering
    \caption{Linear maps used to construct detector error models and extended detector error models.
    For an affine map $f$, $\Delta f$ denotes its linear part, as defined in Eq.~\eqref{eq:affine-map-linearization}.
    The final column lists input and output ports explicitly; ports subscripted $\mathrm{in}$ and $\mathrm{out}$ are interface ports, attached to input and output codes and wired along the code blocks of a gadget circuit~(Def.~\ref{def:map-diagram-of-gadget-circuit}), and the remaining ports are free.
    Badge colors match the corresponding diagram components, with green reserved for compression maps.}
    \label{tab:gadget-dem-components}
    \begingroup
    \footnotesize
    \setlength{\tabcolsep}{4pt}
    \begin{tabular}{@{}cp{0.18\textwidth}p{0.40\textwidth}p{0.23\textwidth}@{}}
        \textcolor{quantumviolet}{\textbf{Abbrev.}}
        & \textcolor{quantumviolet}{\textbf{Full name}}
        & \textcolor{quantumviolet}{\textbf{Description}}
        & \textcolor{quantumviolet}{\textbf{Ports}} \\
        \toprule
        \colorbox[HTML]{E8DFF0}{\strut\textbf{DEM}}
        & Detector error model
        & Bitmatrix part of a closed-experiment DEM~(Sec.~\ref{sec:preliminaries-dems}), mapping fault configurations to detector and observable flips.
        & \textbf{in:} \texttt{F}\newline
        \textbf{out:} $\Delta$\texttt{D}, $\Delta$\texttt{O} \\
        \colorbox[HTML]{E8DFF0}{\strut\textbf{Raw EDEM}}
        & Raw extended DEM
        & Open-boundary generalization of a DEM that retains boundary errors and virtual-syndrome flips.
        & \textbf{in:} \texttt{F}, $\Delta$\texttt{V}$_{\mathrm{in}}$, \texttt{E}$_{\mathrm{in}}$  \newline
        \textbf{out:} $\Delta$\texttt{D}, $\Delta$\texttt{O}, $\Delta$\texttt{V}$_{\mathrm{out}}$, \texttt{E}$_{\mathrm{out}}$ \\
        \colorbox[HTML]{E8DFF0}{\strut\textbf{EDEM}}
        & Extended DEM
        & Raw EDEM whose boundary-error ports are compressed to the boundary-error coordinates selected by the code presentations, Eq.~\eqref{eq:edem-compression}.
        & \textbf{in:} \texttt{F}, $\Delta$\texttt{V}$_{\mathrm{in}}$, $(\texttt{S},\texttt{L})_{\mathrm{in}}$\newline
        \textbf{out:}\,$\Delta\texttt{D},\Delta\texttt{O},\Delta\texttt{V}_{\mathrm{out}},(\texttt{S},\texttt{L})_{\mathrm{out}}$ \\
        \colorbox[HTML]{F4EBCF}{\strut$\boldsymbol{\Delta\mathrm{OM}}$}
        & Outcome map\newline (linear part)
        & Maps physical measurement-outcome flips to observable flips.
        & \textbf{in:} $\Delta$\texttt{M}\newline
        \textbf{out:} $\Delta$\texttt{O} \\
        \colorbox[HTML]{DDEAF4}{\strut$\boldsymbol{\Delta\mathrm{DM}}$}
        & Detector map\newline (linear part)
        & Maps measurement-outcome and input virtual-syndrome flips to detector and output virtual-syndrome flips.
        & \textbf{in:} $\Delta$\texttt{M}, $\Delta$\texttt{V}$_{\mathrm{in}}$\newline
        \textbf{out:} $\Delta$\texttt{D}, $\Delta$\texttt{V}$_{\mathrm{out}}$ \\
        \colorbox[HTML]{F4D9D9}{\strut\textbf{RFE}}
        & Raw fault-effect map
        & Propagates an input boundary error and a fault configuration to outcome flips and an output boundary error~(Def.~\ref{def:raw-fault-effect-map}).
        & \textbf{in:} \texttt{F}, \texttt{E}$_{\mathrm{in}}$\newline
        \textbf{out:} $\Delta$\texttt{M}, \texttt{E}$_{\mathrm{out}}$ \\
        \colorbox[HTML]{DCEFE1}{\strut\textbf{CPM}}
        & Code-presentation map
        & Converts a boundary error into the boundary-error coordinates selected by a code presentation.
        & \textbf{in:} \texttt{E}\newline
        \textbf{out:} $(\texttt{S},\texttt{L})$ \\
        \colorbox[HTML]{DCEFE1}{\strut\textbf{FCM}}
        & Fault-compression map
        & Expands compressed fault coordinates without changing the image of the raw fault-effect map.
        & \textbf{in:} compressed \texttt{F}\newline
        \textbf{out:} \texttt{F} \\
        \bottomrule
    \end{tabular}
    \endgroup
\end{table*}

\subsection{Composition}
\label{sec:gadget-dem-composition}

\begin{figure*}[ht]
    \centering
    \begin{subfigure}[b]{0.185\textwidth}
        \centering
        \includegraphics[height=8em]{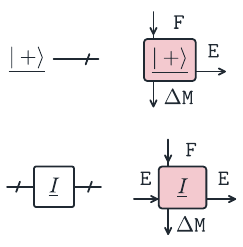}
        \caption{}
        \label{fig:rfe-prep-se}
    \end{subfigure}
    \hfill
    \begin{subfigure}[b]{0.17\textwidth}
        \centering
        \includegraphics[height=8em]{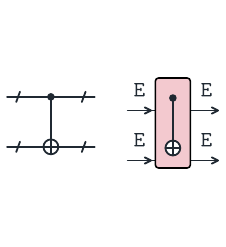}
        \caption{}
        \label{fig:rfe-cnot}
    \end{subfigure}
    \hfill
    \begin{subfigure}[b]{0.31\textwidth}
        \centering
        \includegraphics[height=8em]{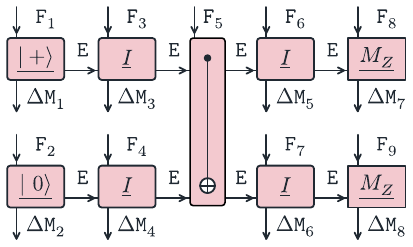}
        \caption{}
        \label{fig:rfe-wiring}
    \end{subfigure}
    \hfill
    \begin{subfigure}[b]{0.27\textwidth}
        \centering
        \includegraphics[height=8em]{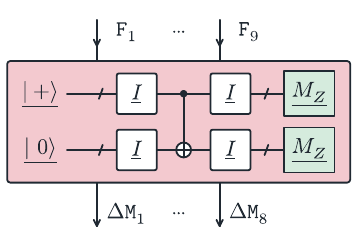}
        \caption{}
        \label{fig:rfe-composed}
    \end{subfigure}
    \caption{Raw fault-effect maps with boundary-error interfaces, the counterpart of Fig.~\ref{fig:extended-composed} for fault propagation.
    Subfigures~\subref{fig:rfe-prep-se} and~\subref{fig:rfe-cnot} show logical operations and their raw fault-effect maps $\mathrm{RFE}_g$~(Def.~\ref{def:raw-fault-effect-map}).
    Subfigure~\subref{fig:rfe-wiring} shows the RFE diagram $\mathrm{RFE}[\mathfrak{G}]$ of the Bell experiment in Fig.~\ref{fig:bell-pair}, in which the per-gadget maps are wired along the code blocks by boundary-error wires,
    and subfigure~\subref{fig:rfe-composed} shows its contraction, the raw fault-effect map $\mathrm{RFE}_G$ of the composite gadget, Eq.~\eqref{eq:rfe-composition}.}
    \label{fig:rfe-composisiton}
\end{figure*}

Let $\mathfrak{G}$ be a gadget circuit~(Def.~\ref{def:gadget-circuit}) in which connected ports carry matching codes and code presentations, and let $G$ be its composite gadget.
The gadget circuit induces a gadget contract, a fault set, and a raw fault-effect map for $G$:
the elementary faults are the disjoint union $\bigsqcup_{g}\mathcal{F}_g$ over the gadget calls,
the detector map is the contraction $\mathrm{DM}_G=\interpret{\mathrm{DM}[\mathfrak{G}]}$~(Sec.~\ref{sec:cutting-detectors}),
the raw fault-effect map is the contraction $\mathrm{RFE}_G=\interpret{\mathrm{RFE}[\mathfrak{G}]}$ of Eq.~\eqref{eq:rfe-composition},
and the code presentations at the boundary of $\mathfrak{G}$ are inherited from the corresponding gadget calls.
We call this the \emph{induced gadget contract} of $\mathfrak{G}$; all statements below refer to it.
For the extended detector error model of a gadget, the interface ports for a code are its virtual-syndrome-flip port together with its boundary-error port, or its boundary-error-coordinate port in the compressed form; the fault, detector-flip, and observable-flip ports are free.

\begin{theorem}[Composition of extended detector error models]
\label{thm:edem-composition}
The raw extended detector error model of the composite gadget is the contraction of the raw-EDEM diagram of the gadget circuit,
\begin{equation}
    \mathrm{EDEM}_G=\interpret{\mathrm{EDEM}[\mathfrak{G}]},
    \label{eq:edem-composition}
\end{equation}
in which each gadget call $g$ is replaced by $\mathrm{EDEM}_g$ and each code-block wire by a boundary-error wire and a virtual-syndrome-flip wire.
For compressed extended detector error models, with boundary-error-coordinate wires in place of boundary-error wires, Eq.~\eqref{eq:edem-composition} holds up to the choice of output representative~(Remark~\ref{rem:rfe-representatives}).
In particular, the detector and deterministic observable matrices of a closed experiment agree exactly; interface coordinates and random-observable rows may depend on the representative.
\end{theorem}

The proof assembles three facts about the diagram contractions.
First, $\mathrm{RFE}_G=\interpret{\mathrm{RFE}[\mathfrak{G}]}$ by Eq.~\eqref{eq:rfe-composition}; Fig.~\ref{fig:rfe-composisiton} shows the RFE diagram and its contraction for the Bell experiment.
Second, $\mathrm{DM}_G=\interpret{\mathrm{DM}[\mathfrak{G}]}$ holds by construction, and $\Delta\mathrm{DM}_G=\interpret{\Delta\mathrm{DM}[\mathfrak{G}]}$ follows because taking linear parts commutes with composition, Eq.~\eqref{eq:affine-map-linearization}.
Third, $\Delta\mathrm{OM}_G=\bigoplus_g\Delta\mathrm{OM}_g$, since the outcome map of the composite gadget is the direct sum of the outcome maps of its gadget calls, Eq.~\eqref{eq:composite-outcome-map}.
Now draw $\mathrm{EDEM}_G$ as in Fig.~\ref{fig:raw-extended-detector-error-model} in terms of $\mathrm{RFE}_G$, $\Delta\mathrm{DM}_G$, and $\Delta\mathrm{OM}_G$,
and substitute the diagrams $\mathrm{RFE}[\mathfrak{G}]$, $\Delta\mathrm{DM}[\mathfrak{G}]$, and $\bigoplus_g\Delta\mathrm{OM}_g$ for them.
The result has three vertices per gadget call, $\mathrm{RFE}_g$, $\Delta\mathrm{DM}_g$, and $\Delta\mathrm{OM}_g$,
and the only wires between different gadget calls are the boundary-error wires of $\mathrm{RFE}[\mathfrak{G}]$ and the virtual-syndrome-flip wires of $\Delta\mathrm{DM}[\mathfrak{G}]$, both inherited from $\mathfrak{G}$.
Grouping the three vertices of each gadget call into a single vertex reproduces $\mathrm{EDEM}_g$ as drawn in Fig.~\ref{fig:raw-extended-detector-error-model}.
The wiring that remains between the grouped vertices is that of $\mathfrak{G}$, so the grouped diagram is $\mathrm{EDEM}[\mathfrak{G}]$.

For extended detector error models one further fact is needed:
inserting a code-presentation map followed by a section~(Def.~\ref{def:code-profile-map}) on a boundary-error wire of $\mathrm{RFE}[\mathfrak{G}]$ leaves the contraction unchanged up to the choice of output representative,
because the two boundary errors differ by a stabilizer of the code on that wire, which acts trivially on the codespace state at the boundary~(Remark~\ref{rem:rfe-representatives}).
Applying Eq.~\eqref{eq:edem-compression} to every gadget call therefore compresses every boundary-error wire of $\mathrm{EDEM}[\mathfrak{G}]$ to a boundary-error-coordinate wire while preserving the represented fault effects up to this freedom.

For subsequent formulas involving compressed EDEMs, we choose the composite representative induced by contracting the fixed per-gadget compressed maps.
Regrouping these fixed maps, as in the window construction below, then gives literal matrix equality.

The theorem applies to any part of a gadget circuit that can itself be regarded as a single gadget call.
\begin{definition}[Window]
\label{def:window}
A set $W$ of gadget calls of $\mathfrak{G}$ is \emph{convex} if every directed path between two calls in $W$ stays in $W$.
The \emph{restriction} $\mathfrak{G}|_W$ is the sub-network with boxes $W$, the wires with both ends in $W$, and a boundary port for every wire entering or leaving $W$~(Sec.~\ref{sec:preliminaries-circuits}).
A \emph{window} is a convex set of gadget calls, identified with its restriction, and $\mathfrak{G}/W$ denotes the gadget circuit in which $W$ is replaced by a single call whose gadget is the composite gadget $G_W$ of $\mathfrak{G}|_W$.
\end{definition}
\begin{corollary}[Window EDEM]
\label{cor:window-edem}
For a window $W$ of $\mathfrak{G}$, $\mathrm{EDEM}_{G_W}=\interpret{\mathrm{EDEM}[\mathfrak{G}|_W]}$, and replacing the boxes of $W$ in $\mathrm{EDEM}[\mathfrak{G}]$ by the single box $\mathrm{EDEM}_{G_W}$ leaves the contraction unchanged.
\end{corollary}
We call $\mathrm{EDEM}_{G_W}$ the \emph{window EDEM} of $W$.
Section~\ref{sec:sliding-window-decoding} uses it as the decoding problem of a window of an adaptive computation.

\begin{corollary}[Causality]
\label{cor:causality}
Fix a topological order of the gadget calls of $\mathfrak{G}$ and assign every detector of the induced gadget contract to the gadget call in which it is fully computable, that is, to the box of $\mathrm{EDEM}[\mathfrak{G}]$ that outputs it.
Then the fault-to-detector block of $\mathrm{EDEM}_G$, and hence the channel-check matrix $A_{\mathsf{D}}$ of a closed experiment, is block lower triangular: a fault in a gadget call can flip only detectors assigned to that call or to gadget calls in its causal future.
\end{corollary}
\begin{proof}
By Theorem~\ref{thm:edem-composition}, the only wires between different gadget calls of $\mathrm{EDEM}[\mathfrak{G}]$ are the boundary-error and virtual-syndrome-flip wires inherited from the acyclic gadget circuit $\mathfrak{G}$, so a path from the fault port of a gadget call to a detector port of another exists only if the second call lies in the causal future of the first.
The block of a contracted diagram between two ports that no path connects is zero~(Sec.~\ref{sec:preliminaries-diagrams}), which is the claim.
\end{proof}

\subsection{Elementary-fault compression}
\label{sec:gadget-dem-faults}

We can further reduce the dimension of the fault-configuration wires by retaining a minimal set of coordinates with the same attainable effects.
\begin{definition}[Fault-compression map]
\label{def:fault-compression-map}
For the restriction $R:\Ftwo^{\mathsf{F}}\to Y$ of a raw fault-effect map to its fault ports, a \emph{fault-compression map} $\mathrm{FCM}:\Ftwo^{\mathsf{F}_{\mathrm{c}}}\to\Ftwo^{\mathsf{F}}$ is a linear map satisfying
\begin{equation}
    \operatorname{im}(R\circ\mathrm{FCM})=\operatorname{im}(R).
    \label{eq:fault-compression-map}
\end{equation}
\end{definition}
When the compressed coordinates select elementary faults, a minimal choice selects effect columns forming a basis of $\operatorname{im}(R)$.

This compression preserves attainable effects, but the compressed coordinates need not be independent stochastic faults.
To retain the original noise model, record each elementary fault $\alpha$ by a compressed coordinate vector $c_\alpha$ satisfying $R\,\mathrm{FCM}\,c_\alpha=R e_\alpha$, where $e_\alpha$ is its unit fault vector, together with its original probability $p_\alpha$.
The compressed maps propagate $c_\alpha$, while the probabilities remain attached to the original independent mechanisms.
If the compressed coordinates themselves are used as random variables, their induced joint distribution must be retained; a single original mechanism can flip several coordinates and thereby correlate them.

Despite this compression, elementary faults in
different gadgets can still have the same fault effect in the composite gadget.
For this reason, the post-processing step that merges elementary faults with the same
effect and updates the fault probabilities using Eq.~\eqref{eq:equivalent-fault-probability} is still necessary.

\section{Symbolic construction of extended detector error models}
\label{sec:symbolic-construction}

Theorem~\ref{thm:edem-composition} expresses the extended detector error model of a gadget circuit as the contraction of a map diagram, which gives us flexibility in how to contract it, that is, how to compute the linear map $\interpret{\mathfrak{D}}$ defined by a diagram $\mathfrak{D}$~(Sec.~\ref{sec:preliminaries-diagrams}).
In this section we explore a symbolic approach, in which contraction produces polynomial expressions for the blocks of the result
and the evaluation of these expressions is deferred and shared.

\subsection{Multi-matrix-variable polynomials}
\label{sec:symbolic-blocks}

\begin{figure*}[ht]
    \centering
    \begin{subfigure}[b]{0.44\textwidth}
        \centering
        \includegraphics[height=7.4em]{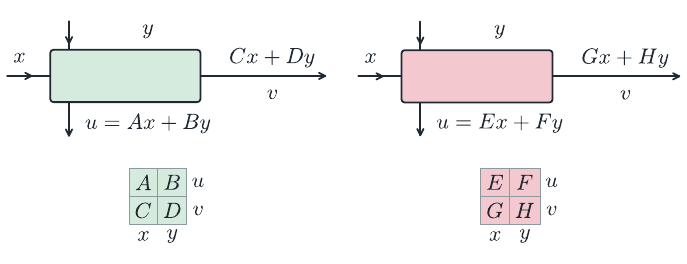}
        \caption{}
        \label{fig:diagrams-and-matrices-separate}
    \end{subfigure}
    \hfill
    \begin{subfigure}[b]{0.52\textwidth}
        \centering
        \includegraphics[height=7.4em]{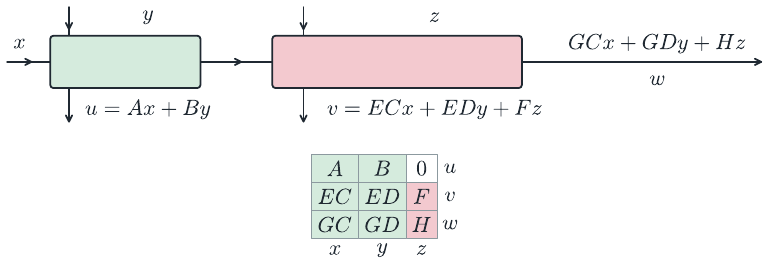}
        \caption{}
        \label{fig:diagrams-and-matrices-composed}
    \end{subfigure}
    \caption{ Linear map diagrams and matrices relating their input and output bitvectors. 
    Diagram composition results in matrix addition and multiplication. The blocks of the composition 
    contain multi-matrix-variable polynomials. 
    Note how an output being insensitive to an input, such as $u$ to $z$ is self evident both from the diagram composition structure and the 0 block in the block-matrix representation.}
    \label{fig:diagrams-and-matrices}
\end{figure*}

The key observation is that when we contract a diagram, the result is a block matrix
whose blocks are multi-matrix-variable polynomials in variables corresponding to the blocks
of the component maps~(Fig.~\ref{fig:diagrams-and-matrices}).
The matrix variables generally do not commute, every product must have compatible dimensions, and polynomial coefficients lie in $\Ftwo$.
When the blocks are sparse and many blocks are repeated, we can first evaluate the
blocks as multi-matrix-variable polynomials and then evaluate the matrix products.
This elucidates which blocks are repeated and ensures that the polynomial for each repeated block
is evaluated only once.

The symbolic evaluation of diagrams proceeds in three steps.
First, we analyze each distinct component diagram and assign a variable to each distinct matrix.
At this stage, in addition to assigning variables, we also identify zero and identity matrices.
Second, we contract the diagram in terms of the introduced variables and obtain a polynomial expression for each block.
Third, we evaluate the distinct polynomial expressions.
Interestingly, the second stage of symbolic evaluation can be the same across distances within the
same family of fault-tolerant protocols, and even across families of fault-tolerant protocols,
as we illustrate in the next subsection.

\subsection{Syndrome extraction repetition example}
\label{sec:symbolic-memory}

In this section we symbolically analyze the extended detector error model of repeated syndrome extraction rounds.
We consider a syndrome extraction gadget that, together with its gadget detector contract, satisfies the following assumptions:
\begin{enumerate}
    \item the gadget has the same input and output code with the same code presentation, and the code presentation has an independent generator list $G$~(Def.~\ref{def:code-profile})
    \item the action of the gadget's realization is equivalent to measuring the independent set of generators $G$; in particular, the sign of each measured generator
    is an affine function of the gadget's measurement outcomes
    \item the gadget logical action is the identity on all logical qubits
    \item the gadget has $|G|$ virtual detectors, each being the parity of the input syndrome bit and the measurement outcomes that determine the sign of the corresponding generator
    \item each output virtual syndrome bit equals the parity of the measurement outcomes that determine the sign of the corresponding generator
\end{enumerate}
Note that not every syndrome extraction gadget satisfies the above assumptions.
Many common syndrome extraction circuits for CSS codes do satisfy them, and thus serve as a good illustration of symbolic evaluation techniques.
Appendix~\ref{app:generalized-syndrome-extraction} gives the corresponding analysis for subsystem and Floquet syndrome extraction, where the group measured from a cycle's input can differ from the stabilizer group prepared at its output and the ideal output may carry an outcome-dependent Pauli frame.

\begin{figure}
     \centering
        \includegraphics[height=7em]{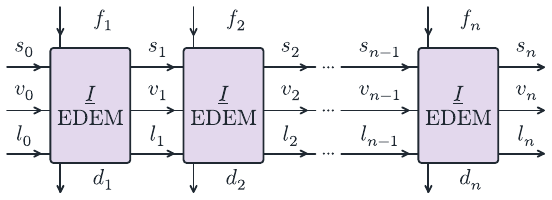}
        \caption{Extended detector error model repetition. The repeated diagram has the input and output structure of a syndrome extraction round.
        The variables $f_k$ are fault configurations,
        $s_k$ are the syndrome components of the boundary-error coordinates,
        $v_k$ are virtual syndrome flips,
        $l_k$ are the logical-effect components of the boundary-error coordinates,
        and $d_k$ are detector flips.
        Here and in the following derivation, we suppress $\Delta$ in $v_k$ and $d_k$.
    }
        \label{fig:extended-detector-model-repetition}
\end{figure}

Next we analyze the repetition of the extended detector error model for the syndrome extraction gadget in~Fig.~\ref{fig:extended-detector-model-repetition}; by Theorem~\ref{thm:edem-composition}, contracting this chained diagram yields the extended detector error model of the composite gadget of the consecutive rounds.
We use the following bitvector variables: 
\begin{align*}
    f_k & \text{-- fault configurations} \\ 
    l_k & \text{-- logical coordinate of the boundary error} \\
    s_k & \text{-- syndrome coordinate of the boundary error} \\ 
    v_k & \text{-- virtual syndrome flips} \\
    d_k & \text{-- detector flips} 
\end{align*}
For some bitmatrices $D_F$, $S_F$ and $L_F$ the variables above satisfy the following equations:
\begin{align}
l_k = & L_F f_k + l_{k-1}\label{eq:logical-effect-bpe} \\
s_k = & S_F f_k + s_{k-1}\label{eq:syndrome-bpe}\\
v_k = & D_F f_k + s_{k-1} \label{eq:virtual-syndrome-flip} \\
d_k = & D_F f_k + s_{k-1} + v_{k-1} \label{eq:detector-flip}
\end{align}
These follow from our assumption on the syndrome extraction gadget.
Since every relation is linear, we check the contribution of each input separately.
The matrices $D_F$, $S_F$, and $L_F$ capture how the gadget's faults affect detectors and the output boundary-error coordinates.

Logical effect~(\ref{eq:logical-effect-bpe}): a fault-free gadget carries the logical-effect component of the input boundary-error coordinates to the output unchanged, and this component does not depend on the input virtual syndrome flips—the latter holds for any gadget, because the output boundary-error coordinates of an extended detector error model depend only on the faults and the input boundary error (see the structural independences stated after Def.~\ref{def:raw-extended-detector-error-model}).

Syndrome~(\ref{eq:syndrome-bpe}): a fault-free gadget preserves the syndrome component of the input boundary-error coordinates because its action is to measure the code stabilizers.
As with the logical-effect component, the syndrome component does not depend on the input virtual syndrome flips or on the logical-effect component.

Output virtual syndrome flips~(\ref{eq:virtual-syndrome-flip}): by Def.~\ref{def:raw-extended-detector-error-model} and our choice of expressions for these flips in the gadget contract, a fault-free gadget outputs the syndrome component of the input boundary-error coordinates—the same $s_{k-1}$ dependence as Eq.~(\ref{eq:syndrome-bpe}), but carried forward as a virtual-syndrome flip rather than a boundary-error coordinate.
Faults enter here and in the detector flips through the same matrix $D_F$, because both are determined by the parities of the measurements used to fix the stabilizer signs.

Detector flips~(\ref{eq:detector-flip}): by the choice of virtual detectors, 
a fault-free gadget with zero incoming boundary-error coordinates reports the input virtual-syndrome flips, and a fault-free gadget with no input virtual-syndrome flips reports the incoming syndrome coordinate.

We can turn the matrix-vector recurrences into closed-form expressions and verify them symbolically~(Appendix~\ref{app:induction-script}), with the result given below:
\begin{align}
v_k = & D_F f_k + s_0 + S_F f^{\Sigma}_{k-1}, f^{\Sigma}_k = \sum_{j=1}^{k} f_j \label{eq:virtual-syndrome-flip-closed} \\
s_k = & s_0 + S_F f^{\Sigma}_k, \,\, l_k = l_0 + L_F f^{\Sigma}_k \label{eq:syndrome-bpe-closed} \\
d_k = & D_F f_k + \left( D_F + S_F \right) f_{k-1}, \quad (k \ge 2) \label{eq:detector-flip-closed}
\end{align}
We see in Eq.~(\ref{eq:detector-flip-closed}) that detector flips in step $k$ depend only on faults from steps $k$ and $k-1$ for $k \ge 2$.
This is what gives the repeated-extraction part of memory-experiment detector error models its banded structure; Sec.~\ref{sec:detector-span} recasts it as a detector-span certificate obtained from the gadget blocks alone~(Lemma~\ref{lem:detector-span-chain}).

A related use of repetition appears in Stim's \texttt{REPEAT} blocks~\cite{Gidney2021stimfaststabilizer}, which compactly represent repeated circuit and detector-error-model instructions.
Here, repeated gadget composition yields symbolic expressions for the model's matrix blocks.

The above example also illustrates that the block matrix corresponding to the EDEM of repeated syndrome extraction has substantial redundancy; 
many blocks are zero, many are identity and the rest repeat among $S_F$, $L_F$, $D_F$, $D_F + S_F$.
In a software implementation, we can simply store pointers to the distinct blocks when representing the extended detector error model, which makes this representation compact and reduces both its memory and its communication requirements.
For example, a decoder unit (such as a GPU, FPGA or ASIC) may have individual blocks pre-loaded\footnote{
    This technique is reminiscent of how sprites were used in early 2D video games to keep animations fluid without redrawing every pixel.
}.
When a decoding coordinator assigns it a specific decoding task, it need only describe the block structure, 
thus minimizing the communication requirement for specifying the EDEM.

\subsection{Caching and pre-computation}
\label{sec:symbolic-caching}

When analyzing many logical circuits implemented using the same fault-tolerant protocols, their detector error models
share blocks with the same polynomial expressions in the block matrices.
To speed up the instantiation of many such experiments, one can cache the results of evaluating the expressions
and amortize their cost across many circuits.
The expressions can be reused across parameter choices that preserve the block structure and identities.
Cached matrix values are specific to the block assignment and dimensions.
We can also prime the cache by evaluating polynomials that occur in common sub-circuits, such as various logical
operations surrounded by syndrome extraction rounds.

\section{Examples and software implementation}
\label{sec:examples-implementation}

We implemented the constructions of Secs.~\ref{sec:cutting-detectors}--\ref{sec:symbolic-construction} in a Rust prototype and used it to build and check extended detector error models for several code families and experiments.
The prototype covers gadgets and their verification, gadget detector contracts, per-gadget EDEM construction, EDEM composition by diagram contraction, and the symbolic evaluation of Sec.~\ref{sec:symbolic-construction}.
Its instruction set covers the repetition code, the rotated surface code, the two-dimensional toric code, and the four-dimensional toric code; for each code it provides state-preparation, syndrome-extraction, destructive-measurement, and transversal-CNOT gadget types with verified contracts, and a transversal-Hadamard gadget type where the code admits one.
Code presentations with over-complete generator lists~(Sec.~\ref{sec:gadget-profiles}) are supported and exercised by the four-dimensional toric code.

\subsection{Implementation overview}
\label{sec:implementation-overview}

The central data structure is the extended detector error model of Sec.~\ref{sec:gadget-dem-definition} in its compressed form: a block bitmatrix whose columns are grouped into the input boundary-error coordinates $(s,l)$, input virtual-syndrome flips, and elementary faults, and whose rows are grouped into detector flips and the corresponding output coordinates and virtual-syndrome flips.
The column groups correspond to the variables $s$, $v$, $l$, and $f$ of Sec.~\ref{sec:symbolic-memory}, and the detector rows to $d$.
The observable-flip rows of Def.~\ref{def:raw-extended-detector-error-model}, obtained through $\Delta\mathrm{OM}$, are stored in the same block bitmatrix; beyond the definition, the implementation also carries success-bit rows for postselection (Sec.~\ref{par:success-bits}).
Each block is stored once and shared by reference; zero and identity blocks require no storage.

Each gadget type is analyzed exactly once.
Fault propagation traverses the circuit once in reverse, recording which detectors and boundary outputs each elementary fault flips.

The analyzed gadget types form an instruction set, and an experiment is a gadget circuit built from calls to these types.
The experiment's DEM is assembled by wiring the per-call EDEMs along shared code blocks and contracting the syndrome, virtual-syndrome, and logical wires, as in Sec.~\ref{sec:gadget-dem-composition}.
The composed bitmatrix retains the partition of detectors and faults by gadget call, which is the block structure used in Sec.~\ref{sec:sliding-window-decoding}; fault probabilities are tracked per gadget and follow the fault columns through composition.

The resulting diagram can be evaluated by direct dense contraction over $\Ftwo$ or by the symbolic method of Sec.~\ref{sec:symbolic-construction}.
In the symbolic method, each distinct block is assigned a variable once per instruction set; contraction then produces a polynomial expression for each result block, and each distinct expression is evaluated once.
An optional cache stores the block matrix values of expressions that recur across experiments (Sec.~\ref{sec:symbolic-caching}).

\subsection{Correctness checks}
\label{sec:implementation-checks}

Gadgets and gadget detector contracts are verified by Clifford simulation, following Ref.~\cite{Kliuchnikov2023stabilizercircuitverification}.
The realization sandwiched between encoders and unencoders is compared against the action under the outcome map~(Def.~\ref{def:gadget}).
The declared detectors, output virtual syndrome relations, success bits and deterministic logical outcomes are checked for determinism and completeness using the contract-correctness conditions, including the rank condition, of Sec.~\ref{sec:gadget-profiles}.

A slower forward propagation of every fault provides an independent check of the reverse fault propagation; the two produce identical EDEMs on all gadget types of the instruction set.

Composition is tested directly.
For every code in the instruction set and every tested round count $n$, the EDEM of $n$ syndrome-extraction rounds built from the flattened circuit equals the contraction of $n$ single-round EDEMs.
The EDEMs of memory, Hadamard, and CNOT experiments built from gadget calls agree with contractions of the corresponding per-stage EDEMs.
The block-level decomposition of an EDEM into the raw fault-effect map, $\Delta\mathrm{DM}$, $\Delta\mathrm{OM}$, and the code-presentation maps (Fig.~\ref{fig:raw-extended-detector-error-model} and Eq.~\eqref{eq:edem-compression}) is likewise reproduced for the surface and both toric codes.

Direct dense contraction and symbolic evaluation produce bit-for-bit identical DEMs on a test suite covering memory experiments on all four codes and Hadamard, CNOT, and GHZ experiments on the surface and toric codes.

Finally, the DEMs obtained by symbolic EDEM composition are cross-checked against
Stim~\cite{Gidney2021stimfaststabilizer}, which constructs the DEM from the
complete physical circuit and its declared detectors and logical
observables, without using the gadget decomposition.
Table~\ref{tab:edem-correctness} specifies the tested code sizes, round
counts, noise parameters, Stim version, and comparison criterion.
The checks cover memory, transversal Hadamard and CNOT, and three- and
five-block encoded GHZ experiments, with all 1,836
code--size--round--experiment--basis--noise configurations passing.
The nonempty detector--observable effect sets agree exactly, and the
largest absolute difference between merged effect probabilities is
below $4.2\times10^{-14}$.

\begin{table*}[t]
\centering
\small
\setlength{\tabcolsep}{4pt}
\renewcommand{\arraystretch}{1.12}
\caption{Validation of symbolically composed EDEMs against Stim.
Each row specifies the Cartesian product of the listed sizes, round
counts, and experiments, tested in both bases under all three noise
presets. Fractions count passing/tested DEM comparisons; every comparison
checks the complete set of nonempty detector--observable effects and
every associated probability.}
\label{tab:edem-correctness}
\begin{tabular*}{\textwidth}{@{\extracolsep{\fill}}llclc@{}}
\toprule
Code & Size & SE rounds $r$ & Experiments & Agreement \\
\midrule
Repetition & $n=3,5,7$ & $2,10,30$ & M, C, $G_3$, $G_5$ & $216/216$ \\
 & $n=15,31,63,127,255$ & $10,100$ & M, C, $G_3$, $G_5$ & $240/240$ \\
 & $n=3,5$ & $100$ & $G_3$, $G_5$ & $24/24$ \\
\addlinespace[0.4em]
Rotated surface & $d=3,5,7$ & $2,10,30$ & M, H, C, $G_3$, $G_5$ & $270/270$ \\
 & $d=11,15,21,31$ & $2,10$ & M, H, C, $G_3$, $G_5$ & $240/240$ \\
 & $d=5$ & $100,300,1000$ & M, H, C & $54/54$ \\
 & $d=3,5$ & $100$ & $G_3$, $G_5$ & $24/24$ \\
\addlinespace[0.4em]
2D toric & $d=3,5,7$ & $2,10,30$ & M, H, C, $G_3$, $G_5$ & $270/270$ \\
 & $d=11,15,21$ & $2,10$ & M, H, C, $G_3$, $G_5$ & $180/180$ \\
 & $d=5$ & $100,300,1000$ & M, H, C & $54/54$ \\
 & $d=3,5$ & $100$ & $G_3$, $G_5$ & $24/24$ \\
\addlinespace[0.4em]
4D toric & $L=2,3,4$ & $2,10$ & M, H, C, $G_3$, $G_5$ & $180/180$ \\
 & $L=2$ & $30,100$ & M, H, C, $G_3$, $G_5$ & $60/60$ \\
\midrule
\multicolumn{4}{l}{\textbf{Total DEM comparisons}} & \textbf{1,836/1,836} \\
\multicolumn{4}{l}{Compared effect probabilities} & $496\,696\,358$ \\
\multicolumn{4}{l}{Largest absolute probability difference} & $4.16\times10^{-14}$ \\
\bottomrule
\end{tabular*}

\vspace{0.5em}
\begin{minipage}{\textwidth}
\footnotesize
\textit{Circuits and sizes.}
M, H, C, and $G_m$ denote memory, Hadamard, CNOT, and $m$-block GHZ
experiments. Here $r$ counts explicit SE rounds per block in each
segment: M has one segment, H and C have one before and one after the
gate, and $G_m$ has one after each rung of a nearest-neighbor CNOT
ladder, starting from $|+\rangle|0\rangle^{\otimes(m-1)}$.
Thus $G_5$ has $4r$ rounds per block.
M and C use matching preparation/readout bases; H uses dual readout.
GHZ readout checks adjacent logical $ZZ$ parities in the $Z$ basis
and the global logical $X^{\otimes m}$ parity in the $X$ basis,
for every encoded logical index.
Repetition size is the number $n$ of data qubits; $d$ is code distance
for surface and 2D toric codes. For 4D toric, $L$ is lattice side length,
with distance $L^2$ and $6L^4$ data qubits per block; the largest
five-block experiment contains 7,680 data qubits.

\smallskip
\textit{Noise and comparison.}
All runs use release builds and Stim~\cite{Gidney2021stimfaststabilizer}
1.16.dev0, with decomposition, loop folding, and gauge detectors disabled
and disjoint-error approximation threshold zero.
At $p=10^{-2}$, the presets are independent X/Z faults and independent
X/Y/Z faults, each at probability $p$ per surviving operation-target
qubit, and all nonidentity joint Pauli/readout patterns on an operation
with $q$ surviving targets and $b$ output bits, each independently at
$p/(4^q2^b-1)$.
Identical effects are merged as
$P_e=\tfrac12[1-\prod_{f:e_f=e}(1-2p_f)]$.
The acceptance bound is
$|P_e-P_e^{\rm Stim}|\leq10^{-12}(1+|P_e|+|P_e^{\rm Stim}|)$.
All effect sets agree exactly. The maximum discrepancy above is measured
over every compared effect probability.
\end{minipage}
\end{table*}

\subsection{Examples}
\label{sec:implementation-examples}

Memory and GHZ experiments illustrate the block structure of the composed DEMs.

\paragraph{Memory experiments.}
For a single code block undergoing preparation, $r$ syndrome-extraction rounds, and destructive measurement, the composed DEM reproduces, on the syndrome-extraction rounds, the banded structure derived in Sec.~\ref{sec:symbolic-memory}: each gadget call contributes nonzero blocks only for itself and the preceding call.


\paragraph{GHZ states on $n$ logical qubits.}
A ladder of transversal CNOTs across $n$ code blocks, with syndrome-extraction rounds after each rung and a destructive measurement of every block, prepares and verifies an $n$-qubit logical GHZ state using $n^2+2n-1$ gadget calls at one round per rung.
This experiment exercises parallel composition: each CNOT rung contributes blocks that relate faults in one code block to detectors in another.
For the tested values of $n$ up to 12 at distance 3, the composed DEM has lower block bandwidth $b=n+1$ ($b=n$ under $Z$-only noise) in the chosen topological order of gadget calls.
Here $b$ is the largest separation $t-s$ of a nonzero block relating detectors of call $t$ to faults of call $s$, so nonzero blocks satisfy $0\le t-s\le b$.
Since $\operatorname{dist}(s,t)\le t-s$ for causally related calls, this also bounds the graph-distance detector span of Eq.~\eqref{eq:bounded-detector-span} by $b$.
This bound need not be tight for a multiport gadget circuit.
Only about $2n^2$ of the $(n^2+2n-1)^2$ block-grid entries are nonzero, so the block representation becomes sparser as $n$ grows.
The GHZ correlations appear as deterministic observables, the shared readout bit as a random one.



\subsection{Prototype limitations}

The current version only supports independent stochastic noise over the elementary Pauli faults of Sec.~\ref{sec:preliminaries-faults}.
These include Pauli insertions on the individual gadget circuits and outcome bit flips.
Including cross-gadget correlated noise is an ongoing line of work.

The prototype composes gadgets along matching ports, with no cross-gadget feed-forward beyond parity-controlled Pauli corrections.
It currently produces monolithic DEMs.
Although the output retains the per-call partition needed for windowing, the prototype does not yet extract window EDEMs~(Cor.~\ref{cor:window-edem}) or automate the detector-span criterion of Lemma~\ref{lem:detector-span-chain}.
It also does not represent the interface dispositions and task ownership of Sec.~\ref{sec:decoder-tasking}.

\section{Window EDEMs and streaming decoding}
\label{sec:sliding-window-decoding}

This section connects compositional EDEM construction to the decoding problem of a long, adaptive computation.
It makes three points.
First, the object a windowed decoder consumes is the \emph{window EDEM}: the contraction of the EDEM diagram restricted to a convex set of gadget calls, which by Theorem~\ref{thm:edem-composition} is again an EDEM.
Second, the two structural properties windowing requires, causality and bounded detector span, can be read off the per-gadget EDEM blocks, so the symbolic analysis of Sec.~\ref{sec:symbolic-construction} also certifies them.
Third, decoding tasks need not follow the causal order of the calls; we distinguish their data dependencies from quantum execution and explain the interface conditions needed to express them compositionally.
We prescribe no decoding algorithm and do not analyze accuracy.

\subsection{Incremental construction and outcome interpretation}
\label{sec:streaming-interpretation}

In a fault-tolerant computation, decoded logical outcomes determine which logical gadgets the quantum processing unit (QPU) executes next.
The gadget circuit $\mathfrak{G}$~(Def.~\ref{def:gadget-circuit}) is thus unknown in advance; it grows one gadget call at a time, in a topological order fixed by the instruction stream.
Recall that a gadget call is a box of the gadget circuit, one use of an elementary gadget at a specific place in the computation; we write \emph{call} for short, and say a call is \emph{issued} when the control system queues its physical instruction to the QPU.

Two consequences follow.
The EDEM diagram $\mathrm{EDEM}[\mathfrak{G}]$ must be assembled incrementally, in step with the gadget calls issued to the QPU.
Decoding must begin before the computation ends, so that logical outcomes can drive control flow within an acceptable \emph{reaction time}~\cite{Gidney2019flexiblelayout}, at a throughput that keeps pace with the growing decoding problem~\cite{Terhal2015quantumerrorcorrection,Battistel2023realtimedecoding}.

The same per-gadget maps serve both needs.
When the control system chooses the next gadget call $g$, it queues the physical instruction to the QPU and the precomputed $\mathrm{EDEM}_g$, 
with its detector map $\mathrm{DM}_g$ and outcome map $\mathrm{OM}_g$, to the decoding unit.
Appending $g$ to $\mathrm{EDEM}[\mathfrak{G}]$ touches only these stored maps and the wires at the ports of $g$.

When $g$ completes, $\mathrm{DM}_g$ evaluated on its physical outcomes and on the virtual syndromes arriving at its input wires yields the detector values $g$ declares and the virtual syndromes it passes on its output wires; $\mathrm{OM}_g$ yields the raw logical outcomes $g$ declares.
The structure that builds the EDEM thus also interprets QPU outcomes as they stream in.
The decoding unit may receive either the physical outcomes of each call or only the detector values computed from them.
Since $\mathrm{DM}_g$ and $\mathrm{OM}_g$ are affine~(Def.~\ref{def:detector-map}), the latter is $\Ftwo$ linear logic that can run in an intermediate unit close to the QPU, which then forwards only detector values, output virtual syndromes, and observable values.
This reduces the communication load when a realization produces many more measurement outcomes than detectors, such as in Floquet~\cite{Hastings2021dynamicallygenerated} and fusion-based~\cite{Bartolucci2023fusionbased} fault tolerance.

\subsection{Window extended detector error models}
\label{sec:window-edems}

A gadget circuit is an acyclic network of gadget calls~(Def.~\ref{def:gadget-circuit}).
We use the directed graph its wires induce on the calls, with an edge from $s$ to $t$ whenever an output code of $s$ is wired to an input code of $t$.
We write $s\preceq t$ if $s=t$ or a directed path leads from $s$ to $t$, call $\preceq$ the \emph{causal order} on calls, and write $\operatorname{dist}(s,t)$ for the length of a shortest such path, or $\infty$ if there is none.

A window~(Def.~\ref{def:window}) is a convex set $W$ of calls, a gadget circuit in its own right and at the same time a single call of the coarser circuit $\mathfrak{G}/W$.
Its window $\mathrm{EDEM}_{G_W}$ is the contraction of the restricted diagram, and $W$ may be replaced by this single box inside $\mathrm{EDEM}[\mathfrak{G}]$ without changing the contraction~(Corollary~\ref{cor:window-edem}).
Its free ports are the fault vectors $\mathsf{F}_W=\bigsqcup_{g\in W}\mathsf{F}_g$, the detector flips $\mathsf{D}_W$, and the observable flips of the calls in $W$.
Its interface ports are a virtual-syndrome-flip port and a boundary-coordinate port for every wire entering or leaving $W$.

Two windows appear in each step of sequential decoding, and they play different roles.
The \emph{inference window} $W$ fixes the decoding problem: its detectors supply the evidence and its elementary faults define the inference variables.
The \emph{commit window} $C\subseteq W$ fixes what is passed on.
Its output boundary is the cut along which decoding composes, and this cut runs through the interior of $W$, between the committed calls and the \emph{buffer} $W\setminus C$ (Fig.~\ref{fig:commit-buffer-banded-dem}).
A committed region has a past and a future, and the interface ports of its EDEM carry what crosses into the buffer and beyond; this is why the right error model for it is an EDEM and not a closed DEM.
Per-gadget blocks supply the rows, columns, and interface maps of both windows, so no physical circuit is revisited and the DEM of the whole computation is never materialized.

\begin{figure*}[t]
    \centering
    \includegraphics[width=\textwidth]{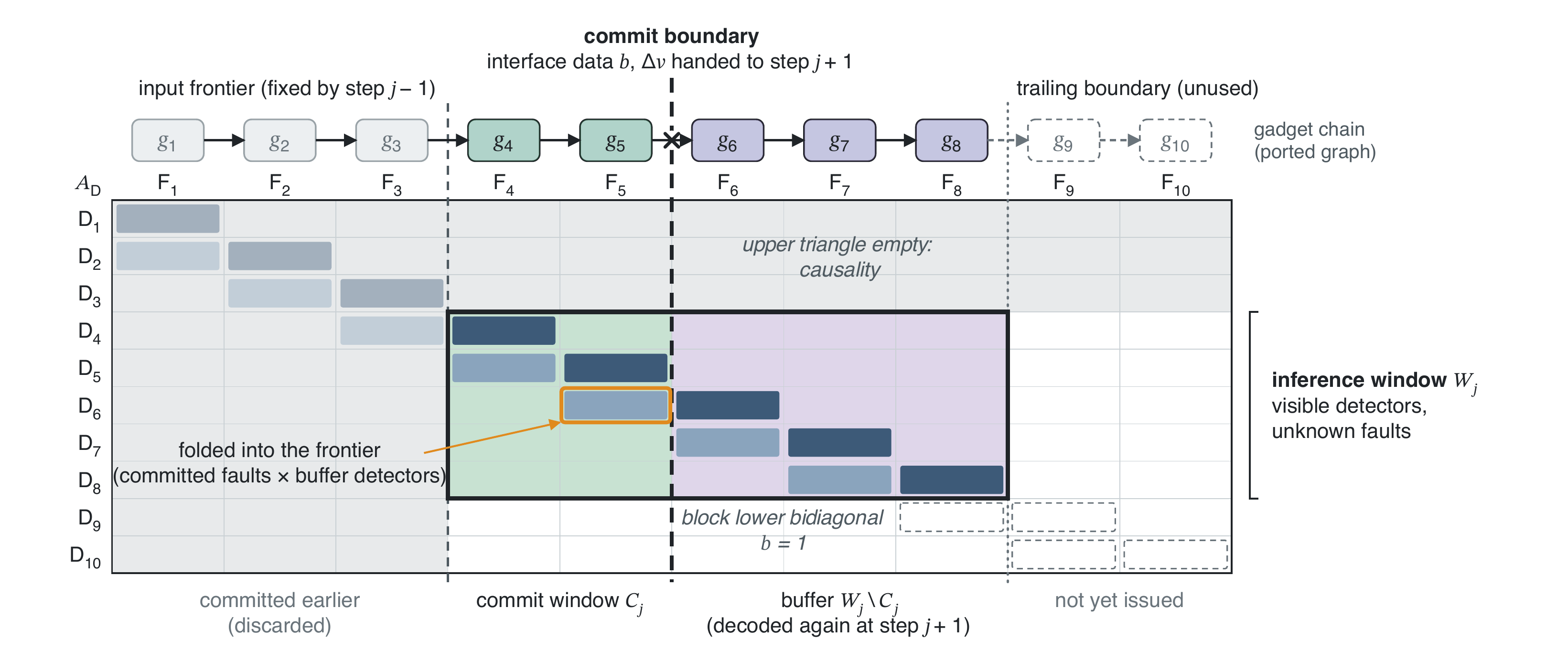}
    \caption{One step of sequential sliding-window decoding on the block-banded channel-check matrix $A_{\mathsf{D}}$ of a chain of gadget calls $g_1\to\cdots\to g_{10}$ (top strip).
    Column blocks are the fault partitions $\mathsf{F}_t$ and row blocks the detector partitions $\mathsf{D}_t$.
    Causality empties the upper triangle.
    The figure assumes detector span $w=1$, so nonzero blocks lie on the diagonal and the first subdiagonal; Sec.~\ref{sec:symbolic-memory} and App.~\ref{app:generalized-syndrome-extraction} provide locality conditions for the two-block band for repeated stabilizer, subsystem, and Floquet syndrome extraction, and Sec.~\ref{sec:detector-span} discusses gadgets that widen it.
    At step $j$ the inference window $W_j=\{g_4,\dots,g_8\}$ (black outline) supplies the visible detectors and the fault variables included in the inference problem.
    Its commit window $C_j=\{g_4,g_5\}$ (green) is separated from the buffer (purple) by the commit boundary (bold dashed line), which cuts the wire between $g_5$ and $g_6$ in the interior of $W_j$; the interface data $b$ and $\Delta v$ on that wire are handed to step $j+1$.
    The committed correction reaches later detectors only through the orange block $(\mathsf{D}_6,\mathsf{F}_5)$, which Proposition~\ref{prop:commit-and-fold} folds into the frontier.
    Calls $g_1$ to $g_3$ (grey) were committed at earlier steps and discarded; their block $(\mathsf{D}_4,\mathsf{F}_3)$ was folded at step $j-1$ and fixes the input frontier of $W_j$.
    The trailing boundary of $W_j$ plays no role in composition, and calls $g_9$ and $g_{10}$ (dashed) are not yet issued.}
    \label{fig:commit-buffer-banded-dem}
\end{figure*}

We use the block notation of Eq.~\eqref{eq:edem-blocks}, with $b=(s,l)$ the boundary coordinates and $\Delta v$ the virtual-syndrome flips.
The same block form holds for a window EDEM with $g$ replaced by $W$.

Turning an inference window into a decoding problem requires a \emph{disposition} for every wire that touches it.
For the sequential orientation considered here, there are three.
\begin{itemize}
    \item \emph{Stitched.}
    Both ends of the wire lie in $W$.
    The interface disappears in the contraction of Corollary~\ref{cor:window-edem}; this is composition.
    \item \emph{Closed.}
    The wire enters $W$ from a committed region.
    Its boundary coordinates $b_{\mathrm{in}}$ and virtual-syndrome flips $\Delta v_{\mathrm{in}}$ encode the interface effect of the previously committed correction and deterministically shift the detector values of $W$ (Proposition~\ref{prop:commit-and-fold} below).
    \item \emph{Open.}
    The wire leaves $W$, or enters it from a region that is neither in $W$ nor committed.
    Its interface data are unconstrained: fault configurations in $W$ that differ only there are indistinguishable from the detectors of $W$, so the inference stops at the wire.
    An open wire acquires a disposition of another kind when its far side is decoded later.
\end{itemize}
A preparation or readout gadget inside $W$ is the degenerate closed case: it has no code port on that side, so the window is closed there without boundary data from any other call.
In sequential decoding the input wires of $W$ are closed, its trailing wires are open, and all wires between its calls are stitched.
Given the fixed interface data $\Delta v_{\mathrm{in}}$ and $b_{\mathrm{in}}$ of its closed wires, the window decoder receives the observed detector flips $\Delta d_W$ and solves for a fault vector $f^\star_W$ with
\begin{equation}
    D^W_F f^\star_W
    =
    \Delta d_W + D^W_B b_{\mathrm{in}} + D^W_V\Delta v_{\mathrm{in}},
    \label{eq:window-decoding-problem}
\end{equation}
using the fault probabilities of $\mathsf{F}_W$ as prior.
This is an ordinary decoding problem for the channel-check matrix $D^W_F$, with the detector values shifted by known interface data.
The logical component $l$ of $b_{\mathrm{in}}$ is the accumulated Pauli frame, and $O^W_B b_{\mathrm{in}}$ is its contribution to the logical outcomes declared in $W$.

A window EDEM is thus consumed by a standard decoder through a closed DEM in the sense of Sec.~\ref{sec:preliminaries-dems}: the fault set $\mathsf{F}_W$ with its probabilities, the check matrix $D^W_F$, and the observable matrix $O^W_F$.
The interface blocks never reach the decoder.
Closed wires enter as the offsets $D^W_B b_{\mathrm{in}}+D^W_V\Delta v_{\mathrm{in}}$ on the detector values and $O^W_B b_{\mathrm{in}}$ on the observables; open outgoing wires contribute nothing, since their rows are not detectors; and an open past-facing wire is handled by an explicit boundary treatment, as discussed in Sec.~\ref{sec:decoder-tasking}.
Decoder-specific conversions, such as decomposing hyperedges into a graphlike model, apply to this DEM as to any other.
The decoder returns a fault vector on all of $\mathsf{F}_W$, which is restricted to the commit window and folded.
Faults of the window DEM with identical columns may be merged only if they agree on commit membership and on their interface and observable effects; otherwise the fold after decoding is wrong.

The decoder commits to a correction, or to an equivalence class of fault configurations with the same relevant effect; it does not claim to have identified the faults that occurred.
The next window needs only the effect of the committed correction on the interface, and evaluating the EDEM supplies it without any retained detector history.
\begin{proposition}[Commit and fold]
\label{prop:commit-and-fold}
Let $C$ be a window of $\mathfrak{G}$, let $\kappa\in\Ftwo^{\mathsf{F}_C}$ be a committed fault vector on $C$, and let $(b_{\mathrm{out}},\Delta v_{\mathrm{out}})$ be the interface outputs of $\mathrm{EDEM}_{G_C}$ evaluated on $\kappa$ with zero interface inputs.
Then for every gadget call $t\notin C$,
\begin{equation}
    \left(A_{\mathsf{D}}\right)_{\mathsf{D}_t,\mathsf{F}_C}\kappa
\end{equation}
equals the detector flip obtained by feeding $(b_{\mathrm{out}},\Delta v_{\mathrm{out}})$ into the wires leaving $C$ in $\mathrm{EDEM}[\mathfrak{G}/C]$, with all other inputs zero.
The same holds for observable flips.
\end{proposition}
\begin{proof}
By Corollary~\ref{cor:window-edem}, $\mathrm{EDEM}[\mathfrak{G}]$ and $\mathrm{EDEM}[\mathfrak{G}/C]$ with $C$ replaced by $\mathrm{EDEM}_{G_C}$ have the same contraction.
Faults propagate only forward along wires, and by convexity a wire leaving $C$ never re-enters it, so every path from $\mathsf{F}_C$ to a detector outside $C$ passes through a wire leaving $C$; by linearity, $\kappa$ contributes through its interface outputs alone.
\end{proof}
The proposition concerns the commit window, not the inference window.
The frontier passed to the next step is the output boundary of $C$; the trailing boundary of the inference window plays no role in composition.
To fold a committed correction into the frontier, add $b_{\mathrm{out}}$ and $\Delta v_{\mathrm{out}}$ to the interface data on the wires leaving $C$ and discard $C$; the interface inputs of $C$, fixed by earlier commits, are transported to the same wires by the fault-free part of $\mathrm{EDEM}_{G_C}$.
In matrix terms this is the residual update of Eq.~\eqref{eq:decoder-task-syndrome-update} below, computed through the interface instead of the global matrix.
The two halves of the handoff differ in what they carry.
If every fault pattern invisible to the detectors of $C$ is also invisible to its outgoing virtual syndromes, $\ker D^C_F\subseteq\ker V^C_F$, then a $\kappa$ that reproduces every detector of $C$ also reproduces the virtual syndromes $C$ hands on, and the boundary coordinates $b_{\mathrm{out}}$ are the only part of the handoff that carries the decoder's belief.
Syndrome-extraction gadgets whose virtual syndromes are the measured stabilizer signs satisfy this with $D^C_F=V^C_F$ (Sec.~\ref{sec:symbolic-memory}).
In the sequential schedule, a logical outcome declared by a call $g\in C$ is final once $C$ and its causal past have been committed, because every contributing fault then has a committed correction.

For a chain, with call $t$ immediately following the committed region $C$, the proposition gives the handoff factorization
\begin{equation}
    \left(A_{\mathsf{D}}\right)_{\mathsf{D}_t,\mathsf{F}_C}
    =
    \begin{bmatrix} D^{t}_B & D^{t}_V \end{bmatrix}
    \begin{bmatrix} B^{C}_F \\ V^{C}_F \end{bmatrix}.
    \label{eq:handoff-factorization}
\end{equation}
More distant calls require intervening interface transfers; general gadget circuits require the path sum described in Sec.~\ref{sec:detector-span}.
A decoding schedule may therefore pass any of three payloads across a cut.
Passing the committed fault vector $\kappa$ leaves both factors to the consumer and exposes the largest basis.
Passing the interface data $(b_{\mathrm{out}},\Delta v_{\mathrm{out}})$ lets the producer apply its own blocks and the consumer its own, so the two sides are coupled only through the code interface and either can be exchanged for another window EDEM with the same gadget detector contract.
Passing the detector delta $(A_{\mathsf{D}})_{\mathsf{D}_t,\mathsf{F}_C}\kappa$ uses the product, computed once per producer and consumer pair, and is the smallest payload, but is specialized to one consumer.
The three are equivalent only if the interface data include the boundary coordinates.
For the syndrome-extraction chain the handoff block is $D_F+S_F$, whereas the virtual-syndrome factors alone give $D_F$; the $S_F$ term travels through $b$, which is why the residual cannot be handed on as virtual syndromes only.

\subsection{Causality and detector span from gadget blocks}
\label{sec:detector-span}

Exactly one gadget call declares each detector, namely the call whose detector map outputs it~(Corollary~\ref{cor:causality}).
The detector-cutting construction of Sec.~\ref{sec:cutting-detectors} thus induces canonical partitions over gadget calls,
\begin{equation}
    \mathsf{D}
    =
    \bigsqcup_{t} \mathsf{D}_t,
    \qquad
    \mathsf{F}
    =
    \bigsqcup_{t} \mathsf{F}_t.
    \label{eq:windowed-detector-fault-partitions}
\end{equation}
Here $A_{\mathsf{D}}$ is the fault-to-detector block of the EDEM of the computation executed so far~(Sec.~\ref{sec:gadget-dem-definition}), which for a closed experiment is the channel-check matrix of Sec.~\ref{sec:preliminaries-dems}.
Every path in $\mathrm{EDEM}[\mathfrak{G}]$ from a fault port of $s$ to a detector port of $t$ follows wires of the gadget circuit, so by Corollary~\ref{cor:causality}
\begin{equation}
    \left(A_{\mathsf{D}}\right)_{\mathsf{D}_t,\mathsf{F}_s}
    =
    0
    \qquad
    \text{unless }s\preceq t.
    \label{eq:causal-channel-check-matrix}
\end{equation}
In any topological order of the gadget calls the channel-check matrix is block lower triangular: a fault cannot flip a detector outside its causal future.

Causality alone does not bound the size of a window, because a fault could flip detectors arbitrarily far downstream.
We say that $\mathfrak{G}$ with its induced gadget detector contract has \emph{detector span} $w$ if
\begin{equation}
    \left(A_{\mathsf{D}}\right)_{\mathsf{D}_t,\mathsf{F}_s}
    =
    0
    \qquad
    \text{whenever }\operatorname{dist}(s,t)>w.
    \label{eq:bounded-detector-span}
\end{equation}
Together, Eqs.~\eqref{eq:causal-channel-check-matrix} and~\eqref{eq:bounded-detector-span} make $A_{\mathsf{D}}$ block local with respect to the gadget circuit.
For a chain in its natural order, detector span $w$ bounds the lower block bandwidth $b\le w$, so the matrix occupies at most $w+1$ block diagonals, including the main diagonal, as in Fig.~\ref{fig:commit-buffer-banded-dem}.
Bandwidth depends on the chosen serialization of the causal order; detector span is intrinsic to the gadget graph and its detector contract.

Detector span is a property of the per-gadget blocks, not an assumption imposed from outside.
In the notation of Eq.~\eqref{eq:edem-blocks}, define the \emph{interface transfer} of a call $g$ as the map from its incoming to its outgoing interface data with faults set to zero,
\begin{equation}
    T^g
    =
    \begin{bmatrix}
        B^g_B & 0 \\
        V^g_B & V^g_V
    \end{bmatrix}
    :
    \begin{bmatrix} b_{\mathrm{in}} \\ \Delta v_{\mathrm{in}} \end{bmatrix}
    \longmapsto
    \begin{bmatrix} b_{\mathrm{out}} \\ \Delta v_{\mathrm{out}} \end{bmatrix}.
    \label{eq:interface-transfer}
\end{equation}
\begin{lemma}[Detector span of a chain]
\label{lem:detector-span-chain}
For a chain of gadget calls $g_1\to g_2\to\cdots$ and $s<t$,
\begin{equation}
    \left(A_{\mathsf{D}}\right)_{\mathsf{D}_t,\mathsf{F}_s}
    =
    \begin{bmatrix} D^{t}_B & D^{t}_V \end{bmatrix}
    T^{t-1}\cdots T^{s+1}
    \begin{bmatrix} B^{s}_F \\ V^{s}_F \end{bmatrix}.
    \label{eq:chain-block-entry}
\end{equation}
In particular, chains of identical calls $g$ have detector span $w$ uniformly in their length if and only if
\begin{equation}
    \begin{bmatrix} D^{g}_B & D^{g}_V \end{bmatrix}
    \left(T^g\right)^{k}
    \begin{bmatrix} B^{g}_F \\ V^{g}_F \end{bmatrix}
    =0
    \qquad\text{for all }k\ge w.
    \label{eq:chain-span-criterion}
\end{equation}
If the interface transfer $T^g$ acts on a $q$-dimensional space, the Cayley--Hamilton theorem reduces this condition to the finite checks $k=w,\ldots,w+q-1$.
\end{lemma}
\begin{proof}
Contracting the chain diagram multiplies matrices along the interface wires.
Detector and observable ports are free and do not feed forward.
The only path from $\mathsf{F}_s$ to $\mathsf{D}_t$ therefore enters the interface through the fault columns $B^{s}_F$ and $V^{s}_F$ of $s$, travels through the fault-free transfers of the intermediate calls, and exits through the interface columns $D^{t}_B$ and $D^{t}_V$ of the detector rows of $t$.
\end{proof}
Bounded detector span does not require the Pauli effect of a fault to vanish.
The interface transfer may carry the residual forward indefinitely; what Eq.~\eqref{eq:chain-span-criterion} demands is that the detector rows of later calls be blind to the transported residual, which happens when the boundary coordinates and the virtual syndromes report it in a way that cancels in $\begin{bmatrix} D^{g}_B & D^{g}_V \end{bmatrix}$.
Only the support of the residual on newly declared detector rows closes, and this is what lets the global EDEM grow without retained detector history.
Repeated stabilizer syndrome extraction is the simplest instance: the blocks of Sec.~\ref{sec:symbolic-memory} satisfy Eq.~\eqref{eq:chain-span-criterion} with $w=1$, and Eq.~\eqref{eq:chain-block-entry} reproduces the closed form of Eq.~\eqref{eq:detector-flip-closed}.
Appendix~\ref{app:generalized-syndrome-extraction} shows that subsystem and Floquet extraction cycles satisfy the same criterion under the locality condition of Eq.~\eqref{eq:generalized-se-locality-condition}, so the two-block band of Fig.~\ref{fig:commit-buffer-banded-dem} covers these cases as well.

A gadget that measures no stabilizers, such as a transversal Clifford gadget, declares no detectors and transports the incoming virtual-syndrome flips to its outputs through a nonzero block $V^g_V$, which relabels and may merge them.
Each such call extends the span by one edge until a later call measures the transported stabilizers; the GHZ experiment of Sec.~\ref{sec:implementation-examples} has measured lower block bandwidth $b=n+1$ for $n$ code blocks, which bounds its detector span.
Windowing remains possible as long as the span stays bounded.
For a general gadget circuit, the entry $(A_{\mathsf{D}})_{\mathsf{D}_t,\mathsf{F}_s}$ is a sum over directed paths from $s$ to $t$ of products of the form in Eq.~\eqref{eq:chain-block-entry}, and detector span $w$ follows if every path product of length exceeding $w$ vanishes.
Since the blocks in Eq.~\eqref{eq:chain-block-entry} are the matrix variables of Sec.~\ref{sec:symbolic-construction}, the span of a fault-tolerant instruction set can be verified symbolically, once per instruction set, with the machinery of Sec.~\ref{sec:symbolic-construction}.

\subsection{Sequential sliding-window decoding}
\label{sec:sequential-windows}

Sequential sliding-window decoding chooses the decoding order to agree with the causal order of the gadget calls.
This is a scheduling choice, additional to the EDEM construction.
Windowed decoding has a classical counterpart for convolutional codes~\cite{Iyengar2012windoweddecoding} and quantum variants including modular, parallel-window, and sandwich decoding~\cite{Bombin2023modulardecoding,Skoric2023parallelwindowdecoding,Tan2023scalablesurfacecodedecoders,Huang2024slidingwindownoisysyndrome}.

At step $j$, the inference window $W_j$ contains the commit window $C_j$ and a buffer $W_j\setminus C_j$.
The cumulative committed region is a down-set of the causal order: every strict causal predecessor of a call in $C_j$ lies either in $C_j$ or in an earlier commit window.
Consequently, every wire entering $C_j$ is fixed by an earlier commit or by a physical preparation boundary.
For a chain, the input frontier of $W_j$ is the commit boundary of $C_{j-1}$.
The decoder solves Eq.~\eqref{eq:window-decoding-problem}, commits the restriction $\kappa_j$ of its estimate to $C_j$, and folds its effect into the frontier using Proposition~\ref{prop:commit-and-fold}.
The buffer estimates remain provisional: their fault variables are included again in the next inference problem, together with newly available calls.
The next window EDEM is assembled from the retained and newly appended per-gadget blocks.
A corrected logical outcome is released once its measurement data are available and every fault contributing to its flip has a committed correction.

Detector span bounds the forward reach of a committed fault's direct detector effects.
Including the forward neighbourhood of $C_j$ through depth $w$ therefore includes every detector row that its faults can flip.
This gives a structural coverage condition for the inference window; the buffer needed for reliable commitment against multi-fault configurations also depends on the code, noise model, and decoder.
For the matchable examples of Ref.~\cite{Bombin2023modulardecoding}, sufficiently large buffers approach monolithic-decoder performance; qLDPC examples with single-shot redundancy can benefit from windows containing only a few rounds~\cite{Huang2024slidingwindownoisysyndrome}.

Streaming decoders such as \emph{Frontier}~\cite{Leverrier2026narrowfrontiers} and \emph{Snowflake}~\cite{Chan2026snowflake}, and the sliding-window wrapper of CUDA-Q QEC~\cite{NVIDIA2026cudaqqecslidingwindow}, motivate this interface between incrementally supplied error-model blocks and a retained decoding state.
The EDEM interface describes deterministic fault effects; it does not by itself specify a decoder's state or the probabilistic information needed for noise correlations across windows.

\subsection{Quantum execution and decoding order}
\label{sec:decoder-tasking}

The quantum circuit and the decoding tasks have different dependency graphs.
Quantum execution follows the causal order $\preceq$ of gadget calls.
Decoding follows the availability of measurement data and of interface information supplied by other decoding tasks.
Both orders are acyclic, but they need not agree: a task may use data from a physically later region before an earlier region has been decoded~\cite{Bombin2023modulardecoding,Skoric2023parallelwindowdecoding,Tan2023scalablesurfacecodedecoders}.

A decoding task $\tau$ has visible detectors $V_\tau\subseteq\mathsf{D}$, an \emph{inference fault set} $U_\tau\subseteq\mathsf{F}$, and a commit set $C_\tau\subseteq U_\tau$ that it owns.
The variables indexed by $U_\tau\setminus C_\tau$ provide buffer context; only the estimate on $C_\tau$ is committed.
Transfers of committed interface information define the edges of an acyclic task graph.
Write $\rho\leadsto\tau$ when $\tau$ depends on $\rho$ through a directed path in this graph.
A task can start once its measurement data and all required predecessor outputs are available.
Including another task's fault variables provisionally in a buffer does not create a dependency on that task's eventual commit.
For fixed global detector coordinates, predecessor commits enter as
\begin{equation}
    \sigma_\tau
    =\Delta d_{V_\tau}
    +\sum_{\rho\leadsto\tau}
    (A_{\mathsf{D}})_{V_\tau,C_\rho}\kappa_\rho.
    \label{eq:decoder-task-syndrome-update}
\end{equation}
Faults already accounted for by these offsets are excluded from $U_\tau$.
Every fault has one owner, and the schedule must arrange that each detector is finally satisfied after all commits that affect it.
At an open boundary, contributions not supplied by predecessors must be included in the local inference model or the affected detector constraints deferred.
These are requirements on the decoding decomposition; EDEM composition supplies fault-effect maps, not a guarantee of decoding accuracy.

The edge--vertex construction of Ref.~\cite{Bombin2023modulardecoding} gives a task graph with two layers.
Edge tasks first decode interface regions, using neighbouring block interiors as buffers.
They have no decoding predecessors and can run independently as soon as their required measurement data are available.
Their commits provide boundary conditions to vertex tasks, which then decode the remaining block interiors.
The committed interfaces separate these remaining problems, so the vertex tasks have closed boundaries and require no further buffers.
In a task-oriented interface description, the incoming wires of each vertex task are thus all closed, including those receiving information from physically later regions.
The parallel-window and sandwich constructions of Refs.~\cite{Skoric2023parallelwindowdecoding,Tan2023scalablesurfacecodedecoders} similarly perform independent buffered tasks followed by tasks that reconcile the intervening regions.
These schedules have two decoding layers, independently of the depth of the quantum circuit.
A general decoding task graph may have greater depth; its dependencies are determined by the chosen handoffs.
In every case, a task must also wait for its required measurement data, and adaptive quantum operations must wait for the decoded outcomes that control them.

EDEMs provide the local fault effects, virtual syndromes, and boundary-error coordinates needed for such handoffs.
The presentation in Secs.~\ref{sec:preliminaries-diagrams}--\ref{sec:gadget-dems}, however, assigns interface inputs and outputs according to the physical input and output codes of each gadget.
Section~\ref{sec:sequential-windows} uses this agreement between physical direction and decoding information flow.
For more general schedules, interface information should flow from the task that produces it to the task that consumes it, whichever physical side of a cut they occupy.
This concerns both virtual-syndrome information and the boundary-error effects represented in a committed correction.

We expect the EDEM perspective to extend to this separate task orientation when the required handoffs admit sufficient interface summaries and well-defined local maps, preserve the composed detector and observable effects, and have acyclic data dependencies.
For the physical orientation, Proposition~\ref{prop:commit-and-fold} establishes such a factorization; a different orientation requires its own compatible interface maps.
For each chosen task orientation, the exported interface data must be computable from the task's available measurements, incoming interface data, and committed correction.
The resulting handoffs must also form an acyclic dependency graph.
The required generalization is therefore to the choice of interface inputs and outputs and their maps, while retaining the compositional description of fault effects.
We leave its general conditions and implementation to future work.

\section{Conclusion}
\label{sec:conclusion}

We have shown how to construct the detector error model of a fault-tolerant circuit from the extended detector error models of its gadgets.
The detector-cutting lemmas of Sec.~\ref{sec:cutting-detectors} justify splitting detectors at gadget boundaries using virtual syndromes, and the extended detector error model of Sec.~\ref{sec:gadget-dems} packages each gadget's fault effects, together with virtual-syndrome interfaces and boundary-error coordinates, into a linear map over $\Ftwo$ that composes the way the gadgets do.
Because the global DEM is then a contraction of a diagram of linear maps, it can be evaluated symbolically: blocks become matrix variables, repeated blocks are evaluated once, and structural facts, such as the banded structure of repeated-syndrome-extraction DEMs, can be derived in closed form (Sec.~\ref{sec:symbolic-construction}).
A Rust prototype reproduces the DEMs that Stim derives from the corresponding flat circuits, for experiments spanning four code families, and realizes the practical benefits of analyzing each gadget type once (Sec.~\ref{sec:examples-implementation}).
Finally, the EDEM of any convex set of gadget calls is the decoding problem of a window; causality and detector span, the latter certified from the gadget blocks, make the channel-check matrix block local with respect to the gadget circuit, and committed corrections are folded into the frontier through the interface (Sec.~\ref{sec:sliding-window-decoding}).

The remaining limitations concern both the scope of the formalism and the capabilities of the prototype.
The formalism assumes stochastic Pauli noise on stabilizer circuits, and composition recovers the detector basis induced by the chosen gadget detector contracts rather than an arbitrary externally specified basis.
Noise correlations that cross gadget boundaries must be carried as explicit interface data.
On the practical side, the prototype does not yet extract window EDEMs and supports no cross-gadget feed-forward beyond parity-controlled Pauli corrections.
Quantum execution and decoding tasks can follow different dependency orders; expressing their handoffs through task-oriented EDEM interfaces requires compatible maps for virtual-syndrome information and boundary-error effects (Sec.~\ref{sec:decoder-tasking}).

A natural next step is to move beyond stabilizer logical action.
Other directions include richer noise models, implementing window extraction on top of the per-call partition the composition already produces, allowing EDEM interface inputs and outputs to follow decoding dependencies, together with an appropriate probabilistic treatment of open cuts, and sharing symbolic evaluations across protocol families.

\section*{Use of AI tools}

AI tools were used to generate the initial LaTeX project boilerplate and drafts of Secs.~\ref{sec:introduction}, \ref{sec:examples-implementation}, and~\ref{sec:conclusion}.
They also assisted with language and grammar editing throughout the manuscript, reorganizing material between the preliminaries and appendices, and identifying gaps in the presentation and unclear passages.
The proof exposition in Sec.~\ref{sec:cutting-detectors-formal-foundations} was revised with AI assistance.
The proofs in Appendix~\ref{app:further-math-preliminaries} were drafted with assistance from a language model and checked by the authors.

Many figures were generated from hand-drawn originals using AI tools.
The Rust prototype and SymPy verification scripts were developed with extensive assistance from AI coding tools.

\section*{Acknowledgments}

We thank members of the NVIDIA Quantum teams for many helpful discussions.

\bibliographystyle{quantum}
\bibliography{references}

@article{Gidney2021stimfaststabilizer,
  author    = {Gidney, Craig},
  title     = {{Stim}: a fast stabilizer circuit simulator},
  journal   = {Quantum},
  volume    = {5},
  pages     = {497},
  month     = jul,
  year      = {2021},
  doi       = {10.22331/q-2021-07-06-497},
}

@website{GidneyStimflow,
  author       = {Gidney, Craig},
  title        = {{stimflow}: annealed utilities for creating {QEC} circuits},
  note         = {Accessed 10 September 2026},
  url          = {https://github.com/quantumlib/Stim/tree/main/glue/stimflow},
}

@article{Terhal2015quantumerrorcorrection,
  author    = {Terhal, Barbara M.},
  title     = {Quantum error correction for quantum memories},
  journal   = {Reviews of Modern Physics},
  volume    = {87},
  number    = {2},
  pages     = {307--346},
  month     = apr,
  year      = {2015},
  publisher = {American Physical Society},
  doi       = {10.1103/RevModPhys.87.307},
  url       = {https://doi.org/10.1103/RevModPhys.87.307},
}

@article{Bombin2023logicalblocks,
  author    = {Bomb{\'i}n, H{\'e}ctor and Dawson, Chris and Mishmash, Ryan V. and Nickerson, Naomi and Pastawski, Fernando and Roberts, Sam},
  title     = {Logical blocks for fault-tolerant topological quantum computation},
  journal   = {PRX Quantum},
  volume    = {4},
  number    = {2},
  pages     = {020303},
  month     = apr,
  year      = {2023},
  publisher = {American Physical Society},
  doi       = {10.1103/PRXQuantum.4.020303},
  url       = {https://doi.org/10.1103/PRXQuantum.4.020303},
}

@article{Bartolucci2023fusionbased,
  author    = {Bartolucci, Sara and Birchall, Patrick and Bomb{\'i}n, H{\'e}ctor and Cable, Hugo and Dawson, Chris and Gimeno-Segovia, Mercedes and Johnston, Eric and Kieling, Konrad and Nickerson, Naomi and Pant, Mihir and Pastawski, Fernando and Rudolph, Terry and Sparrow, Chris},
  title     = {Fusion-based quantum computation},
  journal   = {Nature Communications},
  volume    = {14},
  number    = {1},
  pages     = {912},
  month     = feb,
  year      = {2023},
  publisher = {Springer Nature},
  doi       = {10.1038/s41467-023-36493-1},
  url       = {https://doi.org/10.1038/s41467-023-36493-1},
}

@article{Iyengar2012windoweddecoding,
  author    = {Iyengar, Aravind R. and Papaleo, Marco and Siegel, Paul H. and Wolf, Jack Keil and Vanelli-Coralli, Alessandro and Corazza, Giovanni E.},
  title     = {Windowed decoding of protograph-based {LDPC} convolutional codes over erasure channels},
  journal   = {IEEE Transactions on Information Theory},
  volume    = {58},
  number    = {4},
  pages     = {2303--2320},
  month     = apr,
  year      = {2012},
  publisher = {IEEE},
  doi       = {10.1109/TIT.2011.2177439},
  url       = {https://doi.org/10.1109/TIT.2011.2177439},
}

@article{Skoric2023parallelwindowdecoding,
  author    = {Skoric, Luka and Browne, Dan E. and Barnes, Kenton M. and Gillespie, Neil I. and Campbell, Earl T.},
  title     = {Parallel window decoding enables scalable fault tolerant quantum computation},
  journal   = {Nature Communications},
  volume    = {14},
  number    = {1},
  pages     = {7040},
  month     = nov,
  year      = {2023},
  publisher = {Springer Nature},
  doi       = {10.1038/s41467-023-42482-1},
  url       = {https://doi.org/10.1038/s41467-023-42482-1},
}

@article{Tan2023scalablesurfacecodedecoders,
  author    = {Tan, Xinyu and Zhang, Fang and Chao, Rui and Shi, Yaoyun and Chen, Jianxin},
  title     = {Scalable surface-code decoders with parallelization in time},
  journal   = {PRX Quantum},
  volume    = {4},
  number    = {4},
  pages     = {040344},
  month     = dec,
  year      = {2023},
  publisher = {American Physical Society},
  doi       = {10.1103/PRXQuantum.4.040344},
  url       = {https://doi.org/10.1103/PRXQuantum.4.040344},
}

@article{Huang2024slidingwindownoisysyndrome,
  author    = {Huang, Shilin and Puri, Shruti},
  title     = {Increasing memory lifetime of quantum low-density parity check codes with sliding-window noisy syndrome decoding},
  journal   = {Physical Review A},
  volume    = {110},
  number    = {1},
  pages     = {012453},
  month     = jul,
  year      = {2024},
  publisher = {American Physical Society},
  doi       = {10.1103/PhysRevA.110.012453},
  url       = {https://doi.org/10.1103/PhysRevA.110.012453},
}

@misc{Leverrier2026narrowfrontiers,
  author        = {Leverrier, Anthony and Urbanke, R{\"u}diger},
  title         = {Approximating optimal decoding of quantum {LDPC} codes with narrow frontiers},
  month         = jun,
  year          = {2026},
  eprint        = {2606.20513},
  archiveprefix = {arXiv},
  primaryclass  = {quant-ph},
  doi           = {10.48550/arXiv.2606.20513},
  url           = {https://doi.org/10.48550/arXiv.2606.20513},
}

@article{Chan2026snowflake,
  author    = {Chan, Tim},
  title     = {Snowflake: {A} distributed streaming decoder},
  journal   = {Quantum},
  volume    = {10},
  pages     = {2033},
  month     = mar,
  year      = {2026},
  publisher = {Verein zur F{\"o}rderung des Open Access Publizierens in den Quantenwissenschaften},
  doi       = {10.22331/q-2026-03-20-2033},
  url       = {https://doi.org/10.22331/q-2026-03-20-2033},
}

@misc{NVIDIA2026cudaqqecslidingwindow,
  author       = {{NVIDIA Corporation, CUDA-QX Development Team}},
  title        = {{CUDA-Q QEC C++ API}: Sliding Window Decoder},
  howpublished = {{CUDA-QX} documentation, version 0.7.0},
  year         = {2026},
  note         = {Accessed 27 August 2026},
  url          = {https://nvidia.github.io/cudaqx/api/qec/cpp_api.html#sliding-window-decoder},
}

@article{Battistel2023realtimedecoding,
  author    = {Battistel, Francesco and Chamberland, Christopher and Johar, Kauser and Overwater, Ramon W. J. and Sebastiano, Fabio and Skoric, Luka and Ueno, Yosuke and Usman, Muhammad},
  title     = {Real-time decoding for fault-tolerant quantum computing: progress, challenges and outlook},
  journal   = {Nano Futures},
  volume    = {7},
  number    = {3},
  pages     = {032003},
  month     = aug,
  year      = {2023},
  publisher = {IOP Publishing},
  doi       = {10.1088/2399-1984/aceba6},
  url       = {https://doi.org/10.1088/2399-1984/aceba6},
}

@misc{Gidney2019flexiblelayout,
  author        = {Gidney, Craig and Fowler, Austin G.},
  title         = {Flexible layout of surface code computations using {{AutoCCZ}} states},
  month         = may,
  year          = {2019},
  eprint        = {1905.08916},
  archiveprefix = {arXiv},
  primaryclass  = {quant-ph},
  doi           = {10.48550/arXiv.1905.08916},
  url           = {https://doi.org/10.48550/arXiv.1905.08916},
}

@misc{Bombin2023modulardecoding,
  author        = {Bomb{\'i}n, H{\'e}ctor and Dawson, Chris and Liu, Ye-Hua and Nickerson, Naomi and Pastawski, Fernando and Roberts, Sam},
  title         = {Modular decoding: parallelizable real-time decoding for quantum computers},
  month         = mar,
  year          = {2023},
  eprint        = {2303.04846},
  archiveprefix = {arXiv},
  primaryclass  = {quant-ph},
  doi           = {10.48550/arXiv.2303.04846},
  url           = {https://doi.org/10.48550/arXiv.2303.04846},
}

@article{Bombin2024unifyingflavorsof,
  author    = {Bombin, Hector and Litinski, Daniel and Nickerson, Naomi and Pastawski, Fernando and Roberts, Sam},
  title     = {Unifying flavors of fault tolerance with the {ZX} calculus},
  journal   = {Quantum},
  volume    = {8},
  pages     = {1379},
  month     = jun,
  year      = {2024},
  publisher = {Verein zur F{\"o}rderung des Open Access Publizierens in den Quantenwissenschaften},
  doi       = {10.22331/q-2024-06-18-1379},
  url       = {https://doi.org/10.22331/q-2024-06-18-1379},
}

@misc{Delfosse2023spacetimecodes,
  author        = {Delfosse, Nicolas and Paetznick, Adam},
  title         = {Spacetime codes of {Clifford} circuits},
  month         = apr,
  year          = {2023},
  eprint        = {2304.05943},
  archiveprefix = {arXiv},
  primaryclass  = {quant-ph},
  doi           = {10.48550/arXiv.2304.05943},
  url           = {https://doi.org/10.48550/arXiv.2304.05943},
}

@misc{Kliuchnikov2023stabilizercircuitverification,
  author        = {Kliuchnikov, Vadym and Beverland, Michael and Paetznick, Adam},
  title         = {Stabilizer circuit verification},
  month         = sep,
  year          = {2023},
  eprint        = {2309.08676},
  archiveprefix = {arXiv},
  primaryclass  = {quant-ph},
  doi           = {10.48550/arXiv.2309.08676},
  url           = {https://doi.org/10.48550/arXiv.2309.08676},
}

@misc{Beverland2024faulttolerance,
  author        = {Beverland, Michael E. and Huang, Shilin and Kliuchnikov, Vadym},
  title         = {Fault tolerance of stabilizer channels},
  month         = jan,
  year          = {2024},
  eprint        = {2401.12017},
  archiveprefix = {arXiv},
  primaryclass  = {quant-ph},
  doi           = {10.48550/arXiv.2401.12017},
  url           = {https://doi.org/10.48550/arXiv.2401.12017},
}

@article{Derks2025designingfault,
  author    = {Derks, Peter-Jan H.S. and Townsend-Teague, Alex and Burchards, Ansgar G. and Eisert, Jens},
  title     = {Designing fault-tolerant circuits using detector error models},
  journal   = {Quantum},
  volume    = {9},
  pages     = {1905},
  month     = nov,
  year      = {2025},
  publisher = {Verein zur F{\"o}rderung des Open Access Publizierens in den Quantenwissenschaften},
  doi       = {10.22331/q-2025-11-06-1905},
  url       = {https://doi.org/10.22331/q-2025-11-06-1905},
}

@misc{Rodatz2025faulttolerancebyconstruction,
  author        = {Rodatz, Benjamin and Po{\'o}r, Boldizs{\'a}r and Kissinger, Aleks},
  title         = {Fault tolerance by construction},
  month         = jun,
  year          = {2025},
  eprint        = {2506.17181},
  archiveprefix = {arXiv},
  primaryclass  = {quant-ph},
  doi           = {10.48550/arXiv.2506.17181},
  url           = {https://doi.org/10.48550/arXiv.2506.17181},
}

@article{Dennis2002topologicalmemory,
  author    = {Dennis, Eric and Kitaev, Alexei and Landahl, Andrew and Preskill, John},
  title     = {Topological quantum memory},
  journal   = {Journal of Mathematical Physics},
  volume    = {43},
  number    = {9},
  pages     = {4452--4505},
  month     = sep,
  year      = {2002},
  publisher = {AIP Publishing},
  doi       = {10.1063/1.1499754},
  url       = {https://doi.org/10.1063/1.1499754},
}

@article{Fowler2012surfacecodes,
  author    = {Fowler, Austin G. and Mariantoni, Matteo and Martinis, John M. and Cleland, Andrew N.},
  title     = {Surface codes: towards practical large-scale quantum computation},
  journal   = {Physical Review A},
  volume    = {86},
  number    = {3},
  pages     = {032324},
  month     = sep,
  year      = {2012},
  publisher = {American Physical Society},
  doi       = {10.1103/PhysRevA.86.032324},
  url       = {https://doi.org/10.1103/PhysRevA.86.032324},
}

@misc{Aliferis2006accuracythreshold,
  author        = {Aliferis, Panos and Gottesman, Daniel and Preskill, John},
  title         = {Quantum accuracy threshold for concatenated distance-3 codes},
  month         = apr,
  year          = {2005},
  eprint        = {quant-ph/0504218},
  archiveprefix = {arXiv},
  primaryclass  = {quant-ph},
  doi           = {10.48550/arXiv.quant-ph/0504218},
  url           = {https://doi.org/10.48550/arXiv.quant-ph/0504218},
  note          = {Published in Quantum Information and Computation 6, 97--165 (2006)},
}

@article{Horsman2012latticesurgery,
  author    = {Horsman, Dominic and Fowler, Austin G. and Devitt, Simon and Van Meter, Rodney},
  title     = {Surface code quantum computing by lattice surgery},
  journal   = {New Journal of Physics},
  volume    = {14},
  number    = {12},
  pages     = {123011},
  month     = dec,
  year      = {2012},
  publisher = {IOP Publishing},
  doi       = {10.1088/1367-2630/14/12/123011},
  url       = {https://doi.org/10.1088/1367-2630/14/12/123011},
}

@article{Litinski2019gameofsurfacecodes,
  author    = {Litinski, Daniel},
  title     = {{A Game of Surface Codes}: large-scale quantum computing with lattice surgery},
  journal   = {Quantum},
  volume    = {3},
  pages     = {128},
  month     = mar,
  year      = {2019},
  publisher = {Verein zur F{\"o}rderung des Open Access Publizierens in den Quantenwissenschaften},
  doi       = {10.22331/q-2019-03-05-128},
  url       = {https://doi.org/10.22331/q-2019-03-05-128},
}

@article{GoogleQuantumAI2025belowthreshold,
  author    = {{Google Quantum AI and Collaborators}},
  title     = {Quantum error correction below the surface code threshold},
  journal   = {Nature},
  volume    = {638},
  number    = {8052},
  pages     = {920--926},
  month     = feb,
  year      = {2025},
  publisher = {Springer Nature},
  doi       = {10.1038/s41586-024-08449-y},
  url       = {https://doi.org/10.1038/s41586-024-08449-y},
}

@article{Higgott2025sparseblossom,
  author    = {Higgott, Oscar and Gidney, Craig},
  title     = {{Sparse Blossom}: correcting a million errors per core second with minimum-weight matching},
  journal   = {Quantum},
  volume    = {9},
  pages     = {1600},
  month     = jan,
  year      = {2025},
  publisher = {Verein zur F{\"o}rderung des Open Access Publizierens in den Quantenwissenschaften},
  doi       = {10.22331/q-2025-01-20-1600},
  url       = {https://doi.org/10.22331/q-2025-01-20-1600},
}

@article{Delfosse2021unionfind,
  author    = {Delfosse, Nicolas and Nickerson, Naomi H.},
  title     = {Almost-linear time decoding algorithm for topological codes},
  journal   = {Quantum},
  volume    = {5},
  pages     = {595},
  month     = dec,
  year      = {2021},
  publisher = {Verein zur F{\"o}rderung des Open Access Publizierens in den Quantenwissenschaften},
  doi       = {10.22331/q-2021-12-02-595},
  url       = {https://doi.org/10.22331/q-2021-12-02-595},
}

@article{Panteleev2021degeneratequantumldpc,
  author    = {Panteleev, Pavel and Kalachev, Gleb},
  title     = {Degenerate quantum {LDPC} codes with good finite length performance},
  journal   = {Quantum},
  volume    = {5},
  pages     = {585},
  month     = nov,
  year      = {2021},
  publisher = {Verein zur F{\"o}rderung des Open Access Publizierens in den Quantenwissenschaften},
  doi       = {10.22331/q-2021-11-22-585},
  url       = {https://doi.org/10.22331/q-2021-11-22-585},
}

@article{Chamberland2018cliffordframe,
  author    = {Chamberland, Christopher and Iyer, Pavithran and Poulin, David},
  title     = {Fault-tolerant quantum computing in the {Pauli} or {Clifford} frame with slow error diagnostics},
  journal   = {Quantum},
  volume    = {2},
  pages     = {43},
  month     = jan,
  year      = {2018},
  publisher = {Verein zur F{\"o}rderung des Open Access Publizierens in den Quantenwissenschaften},
  doi       = {10.22331/q-2018-01-04-43},
  url       = {https://doi.org/10.22331/q-2018-01-04-43},
}

@article{Knill2005realisticallynoisy,
  author    = {Knill, Emanuel},
  title     = {Quantum computing with realistically noisy devices},
  journal   = {Nature},
  volume    = {434},
  number    = {7029},
  pages     = {39--44},
  month     = mar,
  year      = {2005},
  publisher = {Springer Nature},
  doi       = {10.1038/nature03350},
  url       = {https://doi.org/10.1038/nature03350},
}

@article{McEwen2023relaxinghardware,
  author    = {McEwen, Matt and Bacon, Dave and Gidney, Craig},
  title     = {Relaxing hardware requirements for surface code circuits using time-dynamics},
  journal   = {Quantum},
  volume    = {7},
  pages     = {1172},
  month     = nov,
  year      = {2023},
  publisher = {Verein zur F{\"o}rderung des Open Access Publizierens in den Quantenwissenschaften},
  doi       = {10.22331/q-2023-11-07-1172},
  url       = {https://doi.org/10.22331/q-2023-11-07-1172},
}

@article{Caune2026realtimedecoding,
  author    = {Caune, Laura and Skoric, Luka and Blunt, Nick S. and Ruban, Archibald and McDaniel, Jimmy and Valery, Joseph A. and Patterson, Andrew D. and Gramolin, Alexander V. and Majaniemi, Joonas and Barnes, Kenton M. and Bialas, Tomasz and Bu{\u g}dayc{\i}, Okan and Crawford, Ophelia and Geh{\'e}r, Gy{\"o}rgy P. and Krovi, Hari and Matekole, Elisha and Topal, Canberk and Poletto, Stefano and Bryant, Michael and Snyder, Kalan and Gillespie, Neil I. and Jones, Glenn and Johar, Kauser and Campbell, Earl T. and Hill, Alexander D.},
  title     = {Demonstrating real-time and low-latency quantum error correction with superconducting qubits},
  journal   = {Nature Communications},
  volume    = {17},
  pages     = {7383},
  month     = jun,
  year      = {2026},
  publisher = {Springer Nature},
  doi       = {10.1038/s41467-026-73331-6},
  url       = {https://doi.org/10.1038/s41467-026-73331-6},
}

@repository{DEQ,
  author  = {{Microsoft}},
  title   = {{deq: dynamic and generic QEC decoding system}},
  year    = {2026},
  code    = {https://github.com/microsoft/qdk-ec/tree/96dc7dee7393a5650db439c6e1ac192dc1101578/deq},
  version = {0.4.2},
  commit  = {96dc7dee7393a5650db439c6e1ac192dc1101578},
}

@article{Choi1974,
author = {Man-Duen Choi},
title = {{A schwarz inequality for positive linear maps on $C^{\ast}$-algebras}},
volume = {18},
journal = {Illinois Journal of Mathematics},
number = {4},
publisher = {Duke University Press},
pages = {565 -- 574},
year = {1974},
doi = {10.1215/ijm/1256051007},
URL = {https://doi.org/10.1215/ijm/1256051007}
}

@article{deGroot2022symmetryprotected,
  doi = {10.22331/q-2022-11-10-856},
  url = {https://doi.org/10.22331/q-2022-11-10-856},
  title = {Symmetry {P}rotected {T}opological {O}rder in {O}pen {Q}uantum {S}ystems},
  author = {de Groot, Caroline and Turzillo, Alex and Schuch, Norbert},
  journal = {{Quantum}},
  issn = {2521-327X},
  publisher = {{Verein zur F{\"{o}}rderung des Open Access Publizierens in den Quantenwissenschaften}},
  volume = {6},
  pages = {856},
  month = nov,
  year = {2022}
}

@article{Yashin2025,
   author = {Vsevolod I. Yashin and Maria A. Elovenkova},
   doi = {10.1007/s11128-025-04682-0},
   issn = {1573-1332},
   issue = {3},
   journal = {Quantum Information Processing},
   month = {3},
   pages = {99},
   title = {Characterization of non-adaptive Clifford channels},
   volume = {24},
   year = {2025}
}

@misc{Gidney2024magicstatecultivation,
  author        = {Gidney, Craig and Shutty, Noah and Jones, Cody},
  title         = {Magic state cultivation: growing {T} states as cheap as {CNOT} gates},
  month         = sep,
  year          = {2024},
  eprint        = {2409.17595},
  archiveprefix = {arXiv},
  primaryclass  = {quant-ph},
  doi           = {10.48550/arXiv.2409.17595},
  url           = {https://doi.org/10.48550/arXiv.2409.17595},
}

@article{Bombin2015gaugecolorcodes,
  author    = {Bomb{\'i}n, H{\'e}ctor},
  title     = {Gauge color codes: optimal transversal gates and gauge fixing in topological stabilizer codes},
  journal   = {New Journal of Physics},
  volume    = {17},
  number    = {8},
  pages     = {083002},
  month     = aug,
  year      = {2015},
  doi       = {10.1088/1367-2630/17/8/083002},
}

@article{Hastings2021dynamicallygenerated,
  author    = {Hastings, Matthew B. and Haah, Jeongwan},
  title     = {Dynamically Generated Logical Qubits},
  journal   = {Quantum},
  volume    = {5},
  pages     = {564},
  month     = oct,
  year      = {2021},
  publisher = {Verein zur F{\"o}rderung des Open Access Publizierens in den Quantenwissenschaften},
  doi       = {10.22331/q-2021-10-19-564},
  url       = {https://doi.org/10.22331/q-2021-10-19-564},
}

@unpublished{WuDeq,
  author = {Wu and others},
  title  = {{deq}: Universal Quantum Decoding at Runtime},
  note   = {In preparation}
}

@article{Lafont2003,
title = {Towards an algebraic theory of Boolean circuits},
journal = {Journal of Pure and Applied Algebra},
volume = {184},
number = {2},
pages = {257-310},
year = {2003},
issn = {0022-4049},
doi = {10.1016/S0022-4049(03)00069-0},
author = {Yves Lafont}
}

\appendix
\section{Further stabilizer preliminaries}
\label{app:further-stabilizer-preliminaries}

This appendix records the Choi-operator facts behind Sec.~\ref{sec:stabilizer-maps}, proves Lemma~\ref{lem:pauli-covariance} through a stabilizer dilation, verifies that the elementary operations of Sec.~\ref{sec:preliminaries-circuits} are stabilizer channels, and compares our stabilizer channels with the outcome-complete circuits of Ref.~\cite{Kliuchnikov2023stabilizercircuitverification}.

\subsection{Stabilizer Choi operators and dilations}
\label{app:stabilizer-channel-dilations}

For a superoperator $\Phi:\mathcal{L}(\Hilb{A})\rightarrow\mathcal{L}(\Hilb{B})$, the Choi operator of Sec.~\ref{sec:stabilizer-maps} is $J(\Phi):=(\Phi\otimes\operatorname{id}_{A'})(\proj{\Omega_A})$, where $\ket{\Omega_A}$ is the unnormalized maximally entangled vector between $A$ and its reference copy $A'$.
We call a nonzero $\Phi$ a \emph{stabilizer superoperator} if $J(\Phi)$ is a stabilizer operator, or equivalently if $\opket{J(\Phi)}$ is proportional to a pure stabilizer state.
An $n_{\mathrm{in}}$-qubit to $n_{\mathrm{out}}$-qubit stabilizer superoperator is therefore determined up to a scalar by $2(n_{\mathrm{in}}+n_{\mathrm{out}})$ independent Pauli stabilizer equations on $\opket{J(\Phi)}$, and after absorbing the computational-basis transposes arising from vectorization into the input Paulis each of them takes the four-slot form of Eq.~\eqref{eq:stabilizer-superoperator-constraint}.
Complete positivity and trace preservation are additional conditions, equivalent respectively to $J(\Phi)\geq 0$ and $\tr_B J(\Phi)=I_{A'}$.
A completely positive, trace-non-increasing stabilizer superoperator is a \emph{stabilizer map}, with the zero map included by convention, and a trace-preserving stabilizer map is a stabilizer channel in the sense of Sec.~\ref{sec:stabilizer-maps}.
Stabilizer maps include the selective projections $\rho\mapsto\Pi_S\rho\Pi_S$ and the stabilizer effects $\rho\mapsto\tr(\Pi_S\rho)$, which need not be channels.
Every stabilizer operator $Q$ induces the stabilizer superoperator $\operatorname{Ad}_Q:\rho\mapsto Q\rho Q^\dagger$ with the single Kraus operator $Q$, and a \emph{stabilizer isometry} is an isometry that is a stabilizer operator.
Tensor products and nonzero serial compositions of stabilizer superoperators are again stabilizer superoperators, since they correspond to tensor products and stabilizer contractions of their vectorized Choi operators.

The rest of this subsection compares two descriptions of a stabilizer channel.
The Choi description used in Sec.~\ref{sec:stabilizer-maps} is intrinsic, whereas a Stinespring dilation introduces a nonunique environment that records discarded information.
For completely positive, trace-preserving maps, the two descriptions define the same class~\cite[Theorem~1, items~(4)--(5)]{Yashin2025}; we include a self-contained proof below.

\begin{proposition}[Stabilizer dilation]
    \label{prop:stabilizer-choi-dilation}
    Let $\Phi:\mathcal{L}(\mathcal{H}_A)\rightarrow\mathcal{L}(\mathcal{H}_B)$ be completely positive and trace preserving.
    Then $J(\Phi)$ is a stabilizer operator if and only if there exist an environment $E$ and a stabilizer isometry
    \begin{equation}
        V:
        \mathcal{H}_A
        \longrightarrow
        \mathcal{H}_B\otimes\mathcal{H}_E
    \end{equation}
    such that
    \begin{equation}
        \Phi(\rho)
        =
        \operatorname{Tr}_E\!\left(V\rho V^\dagger\right).
        \label{eq:stabilizer-stinespring-dilation}
    \end{equation}
\end{proposition}

\begin{proof}
Suppose first that $V$ is a stabilizer isometry satisfying Eq.~\eqref{eq:stabilizer-stinespring-dilation}.
Its vectorization $\lvert V\rangle\!\rangle$ is a pure stabilizer state, and
\begin{equation}
    J(\Phi)
    =
    \operatorname{Tr}_E
    \left(
        \lvert V\rangle\!\rangle
        \langle\!\langle V\rvert
    \right).
    \label{eq:choi-from-stabilizer-dilation}
\end{equation}
To see the form of this reduced operator, write the density operator of a normalized pure stabilizer state with stabilizer group $S$ as $2^{-n}\sum_{s\in S}s$.
The partial trace removes every term acting nontrivially on $E$, leaving a scalar multiple of the projector onto the common eigenspace of the surviving stabilizers.
Thus, $J(\Phi)$ is proportional to a stabilizer-subspace projector and is therefore a stabilizer operator.

Conversely, suppose that $J(\Phi)$ is a stabilizer operator.
We first note that every positive stabilizer operator $K$ is proportional to a stabilizer-subspace projector.
Indeed, tracing one copy out of the pure stabilizer state $\lvert K\rangle\!\rangle$ gives $KK^\dagger$, which by the preceding stabilizer-group argument is proportional to such a projector $\Pi_S$.
Since $K\geq 0$, the positive square root is unique, so $K=\lambda\Pi_S$ for some $\lambda>0$.
Applying this observation to the positive Choi operator gives
\begin{equation}
    J(\Phi)
    =
    \lambda\Pi_S.
    \label{eq:positive-stabilizer-choi-projector}
\end{equation}

Let $C$ be a stabilizer encoding isometry whose image is the support of $\Pi_S$, and let $E$ be a reference copy of its logical input.
The encoded maximally entangled vector
\begin{equation}
    \lvert\Gamma\rangle
    :=
    \sqrt{\lambda}
    (C\otimes I_E)
    \lvert\Omega\rangle
    \label{eq:stabilizer-choi-purification}
\end{equation}
is a stabilizer purification of $J(\Phi)$.
After regrouping its systems as $(B\otimes E)\otimes A'$, write this vector as $\lvert V\rangle\!\rangle$ for an operator $V:A\rightarrow B\otimes E$.
Since $\operatorname{Tr}_{BE}(\lvert V\rangle\!\rangle\langle\!\langle V\rvert)=(V^\dagger V)^T$ and $\operatorname{Tr}_E(\lvert V\rangle\!\rangle\langle\!\langle V\rvert)=J(\Phi)$, trace preservation implies
\begin{equation}
    (V^\dagger V)^T
    =
    \operatorname{Tr}_B J(\Phi)
    =
    I_{A'},
    \label{eq:trace-preservation-implies-isometry}
\end{equation}
and hence $V^\dagger V=I_A$.
Therefore, $V$ is a stabilizer isometry.
Tracing out $E$ recovers $J(\Phi)$, so injectivity of the Choi representation gives Eq.~\eqref{eq:stabilizer-stinespring-dilation}.
\end{proof}

\begin{corollary}[Pauli covariance]
    \label{cor:stabilizer-channel-pauli-covariance}
    For every input Pauli $R$ of a stabilizer channel $\Phi$, there exists an output Pauli $Q$ such that
    \begin{equation}
        \Phi\circ\operatorname{Ad}_R
        =
        \operatorname{Ad}_Q\circ\Phi.
        \label{eq:appendix-stabilizer-channel-pauli-covariance}
    \end{equation}
\end{corollary}

\begin{proof}
Choose the stabilizer isometry $V:A\rightarrow B\otimes E$ supplied by Proposition~\ref{prop:stabilizer-choi-dilation}.
The input-reference marginal of $\lvert V\rangle\!\rangle$ is maximally mixed because
\begin{equation}
    \operatorname{Tr}_{BE}
    \left(
        \lvert V\rangle\!\rangle
        \langle\!\langle V\rvert
    \right)
    =
    (V^\dagger V)^T
    =
    I_{A'}.
    \label{eq:stabilizer-isometry-maximally-mixed-input}
\end{equation}
Consequently, the projection of the phase-free stabilizer group of $\lvert V\rangle\!\rangle$ onto the Pauli space of $A'$ is surjective.
Otherwise, nondegeneracy of the Pauli commutation form would give a nonidentity Pauli $T$ on $A'$ commuting with every projected stabilizer.
Then $I_{BE}\otimes T$ would commute with the complete stabilizer group and hence belong to it, because the stabilizer group of a pure stabilizer state is maximal.
This would make $T$ a stabilizer of the reduced state on $A'$, contradicting Eq.~\eqref{eq:stabilizer-isometry-maximally-mixed-input}.

It follows that for every input Pauli $R$, there is a Pauli $\widetilde W$ on $B\otimes E$ such that, after absorbing a possible sign into $\widetilde W$,
\begin{equation}
    (\widetilde W\otimes R^T)
    \lvert V\rangle\!\rangle
    =
    \lvert V\rangle\!\rangle.
    \label{eq:matched-pauli-stabilizer}
\end{equation}
The vectorization identity converts this equation into $VR=WV$, where $W:=\widetilde W^\dagger$.
Factor $W=Q\otimes S$ into Paulis on $B$ and $E$.
For every input operator $\rho$,
\begin{align*}
    \Phi(R\rho R^\dagger)
    &=
    \operatorname{Tr}_E
    \left[
        (Q\otimes S)V\rho V^\dagger
        (Q^\dagger\otimes S^\dagger)
    \right]
    \\
    &=
    Q\operatorname{Tr}_E(V\rho V^\dagger)Q^\dagger
    =
    Q\Phi(\rho)Q^\dagger,
\end{align*}
which proves Eq.~\eqref{eq:appendix-stabilizer-channel-pauli-covariance}.
\end{proof}

This proves Lemma~\ref{lem:pauli-covariance}; the dual statement is Theorem~1 of Ref.~\cite{Yashin2025}.

The trace-preserving assumption in Proposition~\ref{prop:stabilizer-choi-dilation} is essential.
A trace-nonincreasing stabilizer map may have a stabilizer Choi operator, but an isometry followed only by a partial trace is necessarily trace preserving.
Selective stabilizer projections and effects therefore require postselection or a corresponding effect in their dilation description.

Computational-basis dephasing illustrates why the Choi condition must allow mixed stabilizer operators.
With the unnormalized maximally entangled convention,
\begin{equation}
    \begin{aligned}
        J(\mathcal{T}_Z)
        &=
        \ket{00}\!\bra{00}
        +
        \ket{11}\!\bra{11},
        \\
        \lvert J(\mathcal{T}_Z)\rangle\!\rangle
        &=
        \ket{0000}+\ket{1111}.
    \end{aligned}
    \label{eq:dephasing-choi-and-vectorization}
\end{equation}
The Choi operator is mixed, while its vectorization is a pure GHZ stabilizer state.
An equivalent stabilizer dilation is
\begin{equation}
    V\ket{z}
    =
    \ket{z}_B\ket{z}_E,
    \qquad
    \mathcal{T}_Z(\rho)
    =
    \operatorname{Tr}_E(V\rho V^\dagger).
    \label{eq:dephasing-stabilizer-dilation}
\end{equation}
Thus, requiring the normalized Choi operator itself to be pure would exclude information-losing stabilizer channels such as dephasing.

\subsection{Stabilizer and outcome-complete stabilizer circuits}
\label{sec:stabilizer-versus-outcome-complete-circuits}

Our stabilizer channels with declared classical ports are, by design, very similar to the objects represented by the stabilizer circuits of Definition~2.3 in Ref.~\cite{Kliuchnikov2023stabilizercircuitverification}, which are built from zero-state and fair-coin allocations, Clifford unitaries, nondestructive Pauli measurements, Pauli corrections controlled by affine parities of earlier outcomes, and deallocation of wires promised to be in the zero state.
Such a circuit retains every measurement outcome and every allocated random bit in its outcome vector, and we therefore call it \emph{outcome-complete} in this subsection; the qualifier is ours rather than terminology of the reference.
With the outcome vector represented as a declared classical output register, the channel of an outcome-complete stabilizer circuit has the form
\begin{equation}
    \widehat{\Phi}(\rho)
    =
    \sum_{o}
    Q_o\rho Q_o^\dagger
    \otimes\proj{o},
    \label{eq:outcome-complete-stabilizer-channel}
\end{equation}
where every conditional map has a single stabilizer Kraus operator $Q_o$; in the language of quantum instruments, every conditional map has Kraus rank at most one, so for a pure input the quantum output conditioned on the complete outcome vector is pure whenever it is nonzero.
A stabilizer channel in the sense of Sec.~\ref{sec:stabilizer-maps} need not have this property, because it may discard a classical record or a quantum subsystem.
Moreover, the stabilizer maps of Appendix~\ref{app:stabilizer-channel-dilations} may be trace non-increasing, whereas the instrument represented by an outcome-complete circuit is trace preserving; a selective stabilizer map corresponds to one branch of such an instrument.
Below we verify that every elementary operation of Definition~2.3 is a stabilizer channel with the appropriate classical port declarations, and that conversely every stabilizer circuit built from these operations agrees with an outcome-complete one once the discarded records are accounted for.

We first record two basic constructions with classical interfaces.
A deterministic affine map $f(x)=Ax+b$ between classical bit strings is represented by
\begin{equation}
    \Phi_f(\rho)
    =
    \sum_{x\in\Ftwo^{m_{\mathrm{in}}}}
    \bra{x}\rho\ket{x}
    \proj{Ax+b}.
    \label{eq:affine-classical-channel}
\end{equation}
Indeed, $\opket{J(\Phi_f)}$ is the uniform superposition over an affine binary subspace and is therefore a stabilizer state.

A nondestructive projective $Z$ measurement is obtained by preparing a fresh ancilla in $\ket{0}$, applying a controlled-$X$ from the measured qubit to the ancilla, and declaring the ancilla output classical:
\begin{equation}
    \mathcal{M}_Z
    =
    (\operatorname{id}\otimes\mathcal{T}_Z)
    \circ\operatorname{Ad}_{\mathrm{CX}}
    \circ
    \left(
        \rho\mapsto\rho\otimes\proj{0}
    \right).
    \label{eq:z-measurement-channel}
\end{equation}
The ancilla is a quantum wire during the controlled-$X$ and becomes a classical carrier once $\mathcal{T}_Z$ is applied.
Thus, $\mathcal{M}_Z$ is a stabilizer channel with a declared classical outcome port.
The unencoder of Eq.~\eqref{eq:unencoder} is likewise a stabilizer channel with a declared classical syndrome output, since it is the Clifford unitary $E^\dagger$ followed by $\mathcal{T}_Z$ on every qubit of the syndrome register.


We now check that every elementary operation of Definition~2.3 in Ref.~\cite{Kliuchnikov2023stabilizercircuitverification} is represented by a stabilizer channel with the appropriate classical port declarations.
The allocation of a qubit initialized in the zero state is
\begin{equation}
    \mathcal{A}_0(1)
    :=
    \ket{0}\!\bra{0}.
    \label{eq:zero-state-allocation}
\end{equation}
Its output may be declared either quantum or classical because $\mathcal{T}_Z\circ\mathcal{A}_0=\mathcal{A}_0$.
The two declarations respectively describe a qubit initialized in $\ket{0}$ and a deterministic classical bit initialized to zero, and the preparation of any other single-qubit stabilizer state is $\mathcal{A}_0$ followed by a Clifford unitary.

The allocation of a classical random bit distributed as a fair coin is
\begin{equation}
    \mathcal{R}(1)
    :=
    \frac{I_C}{2}.
    \label{eq:fair-coin-allocation}
\end{equation}
The vectorization of its Choi operator is proportional to $\ket{00}+\ket{11}$, and its declared classical output satisfies $\mathcal{T}_Z\circ\mathcal{R}=\mathcal{R}$ so it can be interpreted as either classical or quantum.

Every Clifford unitary $U$, including every Pauli unitary, induces the stabilizer channel $\operatorname{Ad}_U$ without classical port declarations.
Indeed,
\begin{equation}
    J(\operatorname{Ad}_U)
    =
    \lvert U\rangle\!\rangle
    \langle\!\langle U\rvert,
    \label{eq:clifford-unitary-channel-choi}
\end{equation}
and $\lvert U\rangle\!\rangle$ is a stabilizer state.

For a Pauli observable $P$, its nondestructive measurement is the quantum-classical channel
\begin{equation}
    \begin{aligned}
        \mathcal{M}_P(\rho)
        &:={}
        \sum_{o\in\mathbb{F}_2}
        \Pi_o\rho\Pi_o
        \otimes\ket{o}\!\bra{o},
        \\
        \Pi_o
        &:={}
        \frac{I+(-1)^oP}{2}.
    \end{aligned}
    \label{eq:nondestructive-pauli-measurement}
\end{equation}
Choosing a Clifford unitary that maps $P$ to a single-qubit $Z$ reduces this channel to the construction in Eq.~\eqref{eq:z-measurement-channel}, so $\mathcal{M}_P$ is a stabilizer channel whose explicitly recorded outcome bit is a declared classical output port.

The deallocation of a quantum or classical two-level wire is represented directly by the partial trace $\operatorname{Tr}_w$.
For a discarded wire with no other input, $J(\operatorname{Tr}_w)=I_{w'}$, whose vectorization is a Bell state; discarding one part of a larger input is its tensor product with an identity channel.
Its input may be declared either quantum or classical because $\operatorname{Tr}_w\circ\mathcal{T}_Z=\operatorname{Tr}_w$.
Definition~2.3 permits deallocation only when the wire is known to be in the zero state, whereas our discard channel is defined on arbitrary inputs.
Tracing out a factor of a product state (which the promise guarantees) preserves the purity of the state, which is the invariant that outcome-complete stabilizer circuits maintain.
However, the stabilizer channels of Sec.~\ref{sec:stabilizer-maps}, and hence our stabilizer circuits, allow unconditional tracing out and dephasing of arbitrary subsystems and can therefore ``forget'' information.
On the promised zero-state input, the discard channel agrees with the stabilizer effect $\rho\mapsto\bra{0}\rho\ket{0}$, which is a purity-preserving stabilizer map but is not trace preserving without the promise.
A simple way to guarantee that a stabilizer circuit represents a trace-preserving stabilizer channel is to build it from trace-preserving elementary stabilizer channels only, as Definition~\ref{def:stabilizer-circuit} does.

Finally, let $c\in\mathbb{F}_2^m$ collect earlier classical outcomes, and let $p(c):=a^Tc+b$ be an affine parity.
The Pauli unitary $P$ conditioned on this parity is the channel
\begin{equation}
    \mathcal{C}_{P,p}
    \!\left(
        \ket{c}\!\bra{c}\otimes\rho_c
    \right)
    :=
    \ket{c}\!\bra{c}\otimes
    P^{p(c)}\rho_cP^{p(c)}.
    \label{eq:parity-controlled-pauli-channel}
\end{equation}
It is a stabilizer channel because affine classical control can be realized by Clifford operations on dephased classical controls and the quantum target.
Tensoring and composing these stabilizer channels therefore maps every outcome-complete stabilizer circuit to a stabilizer circuit with the same retained outcome register.

Conversely, a stabilizer circuit built from these operations can be refined to an outcome-complete stabilizer circuit by retaining, as an explicit outcome, every record that it discards.
For example, with $\Pi_o^Z:=(I+(-1)^oZ)/2$, dephasing is a nondestructive measurement whose outcome is forgotten,
\begin{equation}
    \mathcal{T}_Z(\rho)
    =
    \operatorname{Tr}_M
    \left[
        \sum_{o\in\mathbb{F}_2}
        \Pi_o^Z\rho\Pi_o^Z
        \otimes\ket{o}\!\bra{o}_M
    \right],
    \label{eq:dephasing-as-forgotten-measurement}
\end{equation}
and tracing out a wire is a destructive measurement of it in any Pauli basis whose result is forgotten.
If $M$ collects all records introduced in this way, the stabilizer channel and its outcome-complete refinement satisfy
\begin{equation}
    \Phi
    =
    \operatorname{Tr}_M\circ\widehat{\Phi}.
    \label{eq:outcome-complete-refinement}
\end{equation}
Thus the two channels are equal once the record-discarding map $\operatorname{Tr}_M$ is included; in the terms of Definition~\ref{def:circuit-equivalence}, the two circuits are equivalent under the outcome map that keeps the retained outcomes and discards $M$.
Together with the preceding paragraph, this translates each formalism into the other for the circuit class considered here; for a general stabilizer channel, Proposition~\ref{prop:stabilizer-choi-dilation} supplies the corresponding dilation by a stabilizer isometry followed by a discard.
This forgetting is more general than the instrument equivalence of Ref.~\cite{Kliuchnikov2023stabilizercircuitverification}, which only merges outcome branches whose conditional maps are proportional; the two branches of a dephasing measurement are not proportional and cannot be merged in that way.

For a finite unitary group $G$, the \emph{twirling channel} is $\mathcal{T}_G(\rho):=|G|^{-1}\sum_{U\in G}U\rho U^\dagger$, so that the dephasing channel $\mathcal{T}_Z$ of Sec.~\ref{sec:stabilizer-maps} is the twirl over $\langle Z\rangle$.
For a qubit register $A$, let $\mathcal{P}(A)$ denote its Pauli group; the twirl over $\mathcal{P}(A)$ is the \emph{fully depolarizing channel}
\begin{equation}
    \mathcal{T}_{\mathcal{P}(A)}(\rho)
    =
    \tr(\rho)\frac{I_A}{2^{|A|}},
    \label{eq:fully-depolarizing-channel}
\end{equation}
which is a stabilizer channel because $J(\mathcal{T}_{\mathcal{P}(A)})=(I_A/2^{|A|})\otimes I_{A'}$, whose vectorization is proportional to a product of Bell states.
Replacing dephasing by full depolarization in Eq.~\eqref{eq:classical-interface} gives stronger interface conditions: rather than discarding only coherence, they discard all information carried by the selected port.
For an input register $A_I$ and an output register $B_R$ of $\Phi$, with the identity action on the remaining ports understood, we call $A_I$ an \emph{irrelevant input} and $B_R$ an \emph{independently random output} when, respectively,
\begin{equation}
    \Phi\circ\mathcal{T}_{\mathcal{P}(A_I)}
    =
    \Phi,
    \qquad
    \mathcal{T}_{\mathcal{P}(B_R)}\circ\Phi
    =
    \Phi.
    \label{eq:depolarizing-interface}
\end{equation}
The first condition states that $\Phi$ does not depend on the input on $A_I$, and the second that $B_R$ is maximally mixed and uncorrelated with the input and with all other outputs, that is, $\Phi(\rho)=(I_{B_R}/2^{|B_R|})\otimes\widetilde{\Phi}(\rho)$ for a channel $\widetilde{\Phi}$ onto the remaining outputs.
Each condition is generated within $\mathsf{Stab}(\Phi)$ by the $X$- and $Z$-conjugation versions of the corresponding equation in Eq.~\eqref{eq:classical-interface-stabilizers}.

An independently random output can therefore be discarded and regenerated by allocating fresh fair coins, Eq.~\eqref{eq:fair-coin-allocation}, without changing the channel.
The condition is stronger than mere randomness: an outcome that is uniformly random but correlated with another output is not independently random, because applying $\mathcal{T}_{\mathcal{P}(B_R)}$ would destroy the correlation and change the channel.
For example, a fair coin allocated as in Definition~2.3 is an independently random output, whereas a physical wire that is randomized by measuring it, with the result retained in a distinct register $M$, becomes independently random only after $M$ is forgotten.

\section{Further mathematical preliminaries}
\label{app:further-math-preliminaries}

Here we provide a proof of the finite-dimensional multiplicative-domain theorem; the general version is due to Choi~\cite{Choi1974}.

\begin{theorem}[Multiplicative domain in finite dimensions]
\label{thm:multiplicative-domain}
Let
\[
\Psi:\operatorname{L}(\mathcal{Y})
\longrightarrow
\operatorname{L}(\mathcal{X})
\]
be a unital completely positive map, and let
\(A\in\operatorname{L}(\mathcal{Y})\). Suppose that
\[
\Psi(A^\dagger A)
=
\Psi(A)^\dagger\Psi(A)
\]
and
\[
\Psi(AA^\dagger)
=
\Psi(A)\Psi(A)^\dagger.
\]
Then, for every \(B\in\operatorname{L}(\mathcal{Y})\),
\[
\Psi(AB)=\Psi(A)\Psi(B)
\]
and
\[
\Psi(BA)=\Psi(B)\Psi(A).
\]
\end{theorem}

\begin{proof}
Choose a Kraus representation
\[
\Psi(B)=\sum_j K_j^\dagger B K_j,
\qquad
\sum_j K_j^\dagger K_j=I.
\]
Because \(\Psi\) is completely positive, it is Hermiticity
preserving, and hence
\[
\Psi(A^\dagger)=\Psi(A)^\dagger.
\]

Consider
\[
\sum_j
\bigl(AK_j-K_j\Psi(A)\bigr)^\dagger
\bigl(AK_j-K_j\Psi(A)\bigr).
\]
Expanding this expression gives
\begin{align*}
&\sum_j K_j^\dagger A^\dagger A K_j
-
\sum_j K_j^\dagger A^\dagger K_j\Psi(A) \\
&\qquad
-
\Psi(A)^\dagger\sum_j K_j^\dagger A K_j \\
&\qquad
+
\Psi(A)^\dagger
\left(\sum_j K_j^\dagger K_j\right)
\Psi(A) \\
&=
\Psi(A^\dagger A)-\Psi(A)^\dagger\Psi(A).
\end{align*}
By assumption, the right-hand side is zero. Every summand on
the left-hand side is positive semidefinite, so each summand
vanishes. Therefore,
\[
AK_j=K_j\Psi(A)
\]
for every \(j\). Consequently,
\begin{align*}
\Psi(BA)
&=
\sum_j K_j^\dagger BA K_j \\
&=
\sum_j K_j^\dagger B K_j\Psi(A) \\
&=
\Psi(B)\Psi(A).
\end{align*}

Apply the same argument to \(A^\dagger\). The second assumed
equality may be written as
\[
\Psi\bigl((A^\dagger)^\dagger A^\dagger\bigr)
=
\Psi(A^\dagger)^\dagger\Psi(A^\dagger).
\]
It follows that
\[
A^\dagger K_j=K_j\Psi(A^\dagger)
=
K_j\Psi(A)^\dagger.
\]
Taking adjoints gives
\[
K_j^\dagger A=\Psi(A)K_j^\dagger.
\]
Hence
\begin{align*}
\Psi(AB)
&=
\sum_j K_j^\dagger ABK_j \\
&=
\Psi(A)\sum_j K_j^\dagger BK_j \\
&=
\Psi(A)\Psi(B).
\end{align*}
\end{proof}

The following corollary shows how constraints on a dual channel translate into constraints on the primal channel.
This is similar to the argument around Eqs.~(6.17)--(6.18) in Ref.~\cite{deGroot2022symmetryprotected}.

\begin{corollary}
\label{cor:dual-constraints}
Let
\[
\Phi:\operatorname{L}(\mathcal{X})
\longrightarrow
\operatorname{L}(\mathcal{Y})
\]
be completely positive and trace preserving, and let
\(P\in\operatorname{L}(\mathcal{Y})\) and
\(Q\in\operatorname{L}(\mathcal{X})\) be Hermitian unitaries.
Then the following statements are equivalent:
\begin{enumerate}
    \item
    $
    \Phi^\dagger(P)=Q.
    $
    
    \item For every \(X\in\operatorname{L}(\mathcal{X})\),
    $
    P\Phi(QX)=\Phi(X).
    $
\end{enumerate}
\end{corollary}

\begin{proof}
Set
$
\Psi=\Phi^\dagger
$, and suppose first that
\[
\Psi(P)=Q.
\]
Since \(\Phi\) is completely positive and trace preserving,
\(\Psi\) is unital and completely positive. Because \(P\) and
\(Q\) are unitary,
\[
\Psi(P^\dagger P)
=
\Psi(I)
=
I
=
Q^\dagger Q
=
\Psi(P)^\dagger\Psi(P),
\]
and
\[
\Psi(PP^\dagger)
=
I
=
QQ^\dagger
=
\Psi(P)\Psi(P)^\dagger.
\]
The multiplicative-domain theorem~\ref{thm:multiplicative-domain} therefore gives
\[
\Psi(PB)=Q\Psi(B)
\]
for every \(B\in\operatorname{L}(\mathcal{Y})\).

Using the Hilbert--Schmidt inner product
\[
\langle A,B\rangle=\operatorname{Tr}(AB^\dagger),
\]
we obtain, for arbitrary
\(B\in\operatorname{L}(\mathcal{Y})\),
\begin{align*}
\langle \Phi(QX),B\rangle
&=
\langle QX,\Psi(B)\rangle \\
&=
\langle X,Q\Psi(B)\rangle \\
&=
\langle X,\Psi(PB)\rangle \\
&=
\langle \Phi(X),PB\rangle \\
&=
\langle P\Phi(X),B\rangle.
\end{align*}
Therefore,
\[
\Phi(QX)=P\Phi(X).
\]
Multiplying on the left by \(P\), and using \(P^2=I\), gives
\[
P\Phi(QX)=\Phi(X).
\]

Conversely, suppose that
\[
P\Phi(QX)=\Phi(X)
\]
for every \(X\in\operatorname{L}(\mathcal{X})\). Multiplying
on the left by \(P\) gives
\[
\Phi(QX)=P\Phi(X).
\]
Taking the trace and using that \(\Phi\) is trace preserving,
we obtain
\[
\operatorname{Tr}(QX)
=
\operatorname{Tr}(\Phi(QX))
=
\operatorname{Tr}(P\Phi(X)).
\]
By the definition of the dual map,
\begin{align*}
\operatorname{Tr}(P\Phi(X))
&=
\langle \Phi(X),P\rangle
=
\langle X,\Phi^\dagger(P)\rangle \\
&=
\operatorname{Tr}
\bigl(X\Phi^\dagger(P)^\dagger\bigr).
\end{align*}
Thus,
\[
\operatorname{Tr}(XQ)
=
\operatorname{Tr}
\bigl(X\Phi^\dagger(P)^\dagger\bigr)
\]
for every \(X\). By nondegeneracy of the trace pairing,
\[
\Phi^\dagger(P)^\dagger=Q.
\]
Taking adjoints and using \(Q^\dagger=Q\) yields
\[
\Phi^\dagger(P)=Q.
\]
\end{proof}

\section{Repeated subsystem and Floquet syndrome extraction}
\label{app:generalized-syndrome-extraction}

Section~\ref{sec:symbolic-memory} considers the particularly simple case in which every cycle measures the same stabilizer generators that define both of its code boundaries.
This appendix retains the requirement that each cycle prepares its declared output code, up to an outcome-dependent Pauli frame, while separating the Pauli group measured from the input from the stabilizer group prepared at the output.
This separation accommodates subsystem-code gauge fixing~\cite{Bombin2015gaugecolorcodes} and Floquet syndrome extraction~\cite{Hastings2021dynamicallygenerated} without changing the extended detector error model formalism~(Def.~\ref{def:raw-extended-detector-error-model}).

\subsection{Measured and prepared Pauli groups}

Consider a sequence of stabilizer channels
\begin{equation}
    \Phi_t:
    \mathcal{L}(\mathcal{H}_A)
    \longrightarrow
    \mathcal{L}(\mathcal{H}_A\otimes\mathcal{H}_{\mathsf{M}_t}),
    \label{eq:generalized-se-channel}
\end{equation}
where the quantum output of \(\Phi_t\) is the quantum input of \(\Phi_{t+1}\), and \(\mathsf{M}_t\) labels the outcomes declared by cycle \(t\).
The channels and their code presentations may depend periodically on \(t\).
For example, the phases of a Floquet protocol may be represented separately, or one full period may be packaged as a single cycle.

The boundary-constraint group $R_{\mathsf{M}_t}(\Phi_t)$ of Def.~\ref{def:boundary-constraint} distinguishes two phase-free Pauli groups,
\begin{align}
    \mathcal{G}^{\mathrm{meas}}_t
    &:={}
    \left\{
        \bar P\,\middle|\,
        \exists a,b:\ (\bar P,I,a,b)\in R_{\mathsf{M}_t}(\Phi_t)
    \right\},
    \label{eq:generalized-se-measured-group}
    \\
    \mathcal{S}_t
    &:={}
    \left\{
        \bar Q\,\middle|\,
        \exists a,b:\ (I,\bar Q,a,b)\in R_{\mathsf{M}_t}(\Phi_t)
    \right\}.
    \label{eq:generalized-se-prepared-group}
\end{align}
The group \(\mathcal{G}^{\mathrm{meas}}_t\) consists of input Paulis measured by the cycle, whereas \(\mathcal{S}_t\) consists of output Paulis prepared by the cycle; they are the measured group and the output stabilizer group of \(\Phi_t\) in the terminology of Sec.~\ref{sec:stabilizer-channel-constraints}.
Their signed versions \(\mathcal{G}^{\mathrm{meas}}_t(m_t)\) and \(\mathcal{S}_t(m_t)\) attach the sign \((-1)^{a^T m_t+b}\) supplied by the corresponding boundary constraint.

Let \(S_t^{\mathrm{code}}\) be the phase-free stabilizer group of the output code declared after cycle \(t\).
The preparation assumption used here is
\begin{equation}
    S_t^{\mathrm{code}}\leq \mathcal{S}_t.
    \label{eq:generalized-se-preparation-assumption}
\end{equation}
Equivalently, for every \(\bar S\in S_t^{\mathrm{code}}\), there are \(a_S,b_S\) such that
\begin{equation}
    \bigl(\bar S\otimes Z_{\mathsf{M}_t}(a_S,b_S)\bigr)
    \Phi_t(\rho)
    =
    \Phi_t(\rho)
    \qquad
    \text{for every }\rho.
    \label{eq:generalized-se-output-stabilizer}
\end{equation}
Thus the sign of every declared output stabilizer is determined by outcomes of the present cycle, without invoking an input stabilizer.
This condition is not equivalent to measuring all input-code stabilizers: measurement and preparation are different projections of the boundary-constraint group.

Relations with nontrivial Pauli components on both the input and output describe transported logical correlations and, for a subsystem presentation, may also retain gauge correlations.
These input--output relations determine the fault-free transformation of the remaining boundary information.

\subsection{Detector closure between unequal groups}

The stabilizer relations that close immediately across the boundary preceding cycle \(t\) belong to
\begin{equation}
    \mathcal{I}_t
    :=
    \mathcal{S}_{t-1}\cap\mathcal{G}^{\mathrm{meas}}_t,
    \label{eq:generalized-se-overlap-group}
\end{equation}
or to the corresponding subgroups selected by the adjacent code presentations.
For \(\bar H\in\mathcal{I}_t\), choose constraints
\begin{align}
    (I,\bar H,a^{\mathrm{S}},b^{\mathrm{S}})
    &\in R_{\mathsf{M}_{t-1}}(\Phi_{t-1}),
    \\
    (\bar H,I,a^{\mathrm{M}},b^{\mathrm{M}})
    &\in R_{\mathsf{M}_t}(\Phi_t).
\end{align}
Their composition gives the detector relation
\begin{equation}
    \begin{aligned}
        d_{t,\bar H}(m_{t-1},m_t)
        &={}
        (a^{\mathrm{S}})^T m_{t-1}
        +(a^{\mathrm{M}})^T m_t
        \\
        &\quad
        +b^{\mathrm{S}}+b^{\mathrm{M}}
        =0.
    \end{aligned}
    \label{eq:generalized-se-overlap-detector}
\end{equation}
for every fault-free realization.
By contrast, an element of \(\mathcal{S}_{t-1}\setminus\mathcal{G}^{\mathrm{meas}}_t\) need not yield a detector that closes in a later cycle.
It may instead be randomized when cycle \(t\) measures an anticommuting Pauli, as occurs during gauge fixing.
Any longer-range detector must follow from a boundary constraint of the composed channel rather than from an assumption that such a stabilizer persists.

The cycle may also have intrinsic detectors supported entirely on its own outcomes.
The profile can present intrinsic and boundary-closing detectors together by assigning a zero preceding-boundary component to each intrinsic row.

\subsection{Outcome-conditioned boundary coordinates}

Intrinsically random outcomes can select the signs of the prepared stabilizers and a known logical or gauge Pauli frame.
Consequently, the absolute output state is not generally a deterministic function of the preceding boundary state.
Let \(\widehat z_t(m_{\leq t})\) denote the ideal boundary coordinates obtained by propagating the preceding ideal branch through the code and Pauli-frame assignments declared for the outcome history \(m_{\leq t}\).
Here the Pauli-frame assignment is the outcome-dependent sign data \((a,b)\) of the output-stabilizer and transported-logical boundary constraints of Def.~\ref{def:boundary-constraint}, that is, the parities \(a^Tm_t+b\) that fix the signs of the prepared stabilizers and of the transported logical operators.
If \(z_t\) denotes the actual boundary coordinates in the same binary presentation, define the residual boundary error coordinates by
\begin{equation}
    \xi_t
    :=
    z_t+\widehat z_t(m_{\leq t}).
    \label{eq:generalized-se-relative-boundary-error}
\end{equation}
Thus \(\xi_t=0\) on every ideal outcome branch when the input boundary error is zero.
The syndrome and logical variables in Sec.~\ref{sec:symbolic-memory} are components of this outcome-conditioned residual, rather than absolute stabilizer signs or an absolute logical frame.
Any random gauge label that can affect a later cycle must likewise be retained by the profile as virtual boundary information instead of being averaged out.

Let \(f_t\) be the fault configuration of cycle \(t\)~(Sec.~\ref{sec:preliminaries-faults}), and let \(v_t\) be the output virtual-syndrome \emph{flip} in the presentation of \(\mathcal{S}_t\).
Let \(u_t\) collect the current-cycle parity flips appearing in the chosen detector rows, including measured-input parities for \(\mathcal{I}_t\) and intrinsic detector parities.
In the recurrence below, \(A_t\), \(B_t\), \(C_t\), \(D_t\), \(E_t\), \(P_t\), and \(Q_t\) denote blocks of the extended detector error model of cycle \(t\)~(Def.~\ref{def:raw-extended-detector-error-model}); they are unrelated to the matrices \(A_{\mathsf{M}}\), \(A_{\mathsf{D}}\), \(A_{\mathsf{O}}\) of Sec.~\ref{sec:preliminaries-dems} and to the encoder \(E\) of Sec.~\ref{sec:preliminaries-codes}.
The most general linear recurrence needed below has the form
\begin{align}
    \xi_t
    &={}
    A_t\xi_{t-1}+B_t f_t,
    \label{eq:generalized-se-boundary-recurrence}
    \\
    u_t
    &={}
    Q_t\xi_{t-1}+D_t f_t,
    \label{eq:generalized-se-measurement-recurrence}
    \\
    v_t
    &={}
    C_t\xi_{t-1}+E_t f_t,
    \label{eq:generalized-se-virtual-recurrence}
    \\
    d_t
    &={}
    u_t+P_t v_{t-1}
    =D_t f_t+Q_t\xi_{t-1}+P_t v_{t-1}.
    \label{eq:generalized-se-detector-recurrence}
\end{align}
Here \(P_t\) expresses the preceding prepared-stabilizer presentation in the detector-row presentation; its rows vanish for detectors intrinsic to cycle \(t\).
The matrices \(Q_t,D_t\) describe the measurement side of the detector relations, whereas \(C_t,E_t\) describe the independently chosen output-stabilizer presentation.
Unlike the specialization in Sec.~\ref{sec:symbolic-memory}, these pairs of matrices need not agree.
Known affine offsets have disappeared because Eqs.~\eqref{eq:generalized-se-boundary-recurrence}--\eqref{eq:generalized-se-detector-recurrence} describe flips relative to the outcome-conditioned ideal branch.

Define the ordered product
\begin{equation}
    A_{t:j}:=A_tA_{t-1}\cdots A_j,
    \label{eq:generalized-se-ordered-product}
\end{equation}
with an empty product equal to the identity.
Iteration of the boundary recurrence gives
\begin{equation}
    \xi_t
    =
    A_{t:1}\xi_0
    +\sum_{j=1}^{t} A_{t:j+1}B_j f_j,
    \label{eq:generalized-se-boundary-closed}
\end{equation}
and hence
\begin{equation}
    v_t
    =
    C_tA_{t-1:1}\xi_0
    +E_t f_t
    +\sum_{j=1}^{t-1}C_tA_{t-1:j+1}B_j f_j.
    \label{eq:generalized-se-virtual-closed}
\end{equation}
Substitution of the preceding-cycle recurrences into Eq.~\eqref{eq:generalized-se-detector-recurrence} gives, for \(t\geq2\),
\begin{align}
    d_t
    ={}&
    D_t f_t
    +\bigl(Q_tB_{t-1}+P_tE_{t-1}\bigr)f_{t-1}
    \notag\\
    &+\bigl(Q_tA_{t-1}+P_tC_{t-1}\bigr)\xi_{t-2}.
    \label{eq:generalized-se-detector-expanded}
\end{align}

A detector profile with one cycle of boundary memory satisfies the compatibility condition
\begin{equation}
    Q_tA_{t-1}+P_tC_{t-1}=0.
    \label{eq:generalized-se-locality-condition}
\end{equation}
This condition states that a boundary error present at the input of cycle \(t-1\) contributes identically to the two sides of each closure relation and therefore cancels.
Under this condition,
\begin{equation}
    d_t
    =
    D_t f_t
    +\bigl(Q_tB_{t-1}+P_tE_{t-1}\bigr)f_{t-1},
    \qquad t\geq2.
    \label{eq:generalized-se-detector-closed}
\end{equation}
Equation~\eqref{eq:generalized-se-locality-condition}, rather than equality of \(\mathcal{G}^{\mathrm{meas}}_t\) and \(\mathcal{S}_t\), is the algebraic condition responsible for the block lower bidiagonal structure of the detector error model, cf.\ Eq.~\eqref{eq:detector-flip-closed} and the detector span of Sec.~\ref{sec:sliding-window-decoding}.
If a chosen elementary Floquet phase has a longer detector window, one may either retain the corresponding additional virtual boundary data or group a full period into one syndrome extraction cycle.

\subsection{Reduction to the code-preserving example}

The equations in Sec.~\ref{sec:symbolic-memory} are recovered when
\begin{equation}
    \xi_t=
    \begin{bmatrix}
        s_t\\
        l_t
    \end{bmatrix},
    \qquad
    A_t=I,
    \qquad
    B_t=
    \begin{bmatrix}
        S_F\\
        L_F
    \end{bmatrix},
    \label{eq:generalized-se-special-boundary-matrices}
\end{equation}
and
\begin{align}
    Q_t=C_t
    &={}
    \begin{bmatrix}
        I&0
    \end{bmatrix},
    &
    P_t
    &={}I,
    \notag\\
    D_t
    &={}E_t=D_F.
    \label{eq:generalized-se-special-detector-matrices}
\end{align}
Equations~\eqref{eq:generalized-se-detector-recurrence} and \eqref{eq:generalized-se-virtual-recurrence} then become Eqs.~\eqref{eq:detector-flip} and \eqref{eq:virtual-syndrome-flip}, respectively.
Moreover, Eq.~\eqref{eq:generalized-se-locality-condition} reduces to \(I+I=0\), and Eq.~\eqref{eq:generalized-se-detector-closed} becomes Eq.~\eqref{eq:detector-flip-closed}.

For a protocol with period \(p\), grouping its phases into a full cycle gives the boundary matrices
\begin{align}
    A_{\mathrm{cyc}}
    &:={}
    A_pA_{p-1}\cdots A_1,
    \\
    B_{\mathrm{cyc}}
    &:={}
    \begin{bmatrix}
        A_p\cdots A_2B_1
        &\cdots
        &A_pB_{p-1}
        &B_p
    \end{bmatrix}.
    \label{eq:generalized-se-period-matrices}
\end{align}

\clearpage
\onecolumngrid
\section{Symbolic verification of the closed forms}
\label{app:induction-script}

Listing below verifies, with the SymPy computer-algebra system, the base case $k=1$ and the induction step $k-1\to k$ of the closed forms in Eqs.~\eqref{eq:virtual-syndrome-flip-closed}--\eqref{eq:detector-flip-closed}, starting from the recurrences~\eqref{eq:logical-effect-bpe}--\eqref{eq:detector-flip}.
The matrices $D_F$, $S_F$, $L_F$ and all bitvectors are symbolic and coefficients are reduced modulo~2, so the single induction-step check covers every $k\ge2$; the base case additionally records $d_1=D_Ff_1+s_0+v_0$, which is Eq.~\eqref{eq:detector-flip} at $k=1$.

\begin{lstlisting}[
    style=pythonsource,
    % caption={Symbolic verification of the closed-form recurrences.},
    % label={lst:induction-script}
]
"""Prove the closed forms of the recurrences by induction.

RECURRENCE (sections/05_symbolic_construction.tex):
    d_k = D_F f_k + s_{k-1} + v_{k-1}      (eq:detector-flip)
    s_k = s_{k-1} + S_F f_k                (eq:syndrome-bpe)
    l_k = l_{k-1} + L_F f_k                (eq:logical-effect-bpe)
    v_k = D_F f_k + s_{k-1}                (eq:virtual-syndrome-flip)

CLOSED FORMS:
    s_k = s_0 + S_F F_k
    l_k = l_0 + L_F F_k
    v_k = D_F f_k + s_0 + S_F F_{k-1}
    d_1 = D_F f_1 + s_0 + v_0
    d_k = D_F f_k + (D_F + S_F) f_{k-1}    (k >= 2:)
where F_k = f_1 + ... + f_k.
"""

from sympy import Add, MatrixSymbol, Mul, Symbol, ZeroMatrix

N_F, N_G, N_L = (Symbol(x, integer=True, positive=True)
                 for x in ('N_F', 'N_G', 'N_L'))
D_F = MatrixSymbol('D_F', N_G, N_F)        # faults -> detectors
S_F = MatrixSymbol('S_F', N_G, N_F)        # faults -> syndrome
L_F = MatrixSymbol('L_F', N_L, N_F)        # faults -> logical effect
ZG, ZL, ZF = ZeroMatrix(N_G, 1), ZeroMatrix(N_L, 1), ZeroMatrix(N_F, 1)

s_0 = MatrixSymbol('s_0', N_G, 1)
l_0 = MatrixSymbol('l_0', N_L, 1)
v_0 = MatrixSymbol('v_0', N_G, 1)


def mod2(expr):
    """Reduce integer coefficients of a matrix expression mod 2."""
    expr = expr.doit().expand()
    terms = []
    for t in (expr.args if isinstance(expr, Add) else (expr,)):
        head, *rest = t.args if isinstance(t, Mul) else (None,)
        if not (head is not None and head.is_Integer):
            terms.append(t)
        elif int(head) % 2:
            terms.append(Mul(*rest))
    return Add(*terms) if terms else ZeroMatrix(expr.rows, expr.cols)


def step(s, l, v, f_k):
    """One step of the RECURRENCE: state at k-1 and f_k -> state at k, d_k."""
    d_k = D_F * f_k + s + v
    v_k = D_F * f_k + s
    s_k = s + S_F * f_k
    l_k = l + L_F * f_k
    return s_k, l_k, v_k, d_k


def check(name, got, want, zero):
    assert mod2(got - want) == zero, f'{name} does not match'


# --- base case: k = 1, state at k-1 = 0 is the initial state, F_0 = 0 -------
f_1 = MatrixSymbol('f_1', N_F, 1)
s_1, l_1, v_1, d_1 = step(s_0, l_0, v_0, f_1)

check('s_1', s_1, s_0 + S_F * f_1, ZG)                  # F_1 = f_1
check('l_1', l_1, l_0 + L_F * f_1, ZL)
check('v_1', v_1, D_F * f_1 + s_0 + S_F * ZF, ZG)       # F_0 = 0
check('d_1', d_1, D_F * f_1 + s_0 + v_0, ZG)

# --- induction step: assume the closed forms at k-1, derive them at k -------
# k is symbolic: f_km1, f_k and the accumulator F_km2 = f_1 + ... + f_{k-2}
# are unconstrained, so this single check covers every k >= 2 at once.
f_km1 = MatrixSymbol('f_{k-1}', N_F, 1)
f_k = MatrixSymbol('f_k', N_F, 1)
F_km2 = MatrixSymbol('F_{k-2}', N_F, 1)
F_km1 = F_km2 + f_km1                                   # F_{k-1} = F_{k-2} + f_{k-1}
F_k = F_km1 + f_k                                       # F_k     = F_{k-1} + f_k

# induction hypothesis: the closed forms hold at k-1 (with k-1 >= 1)
s_km1 = s_0 + S_F * F_km1
l_km1 = l_0 + L_F * F_km1
v_km1 = D_F * f_km1 + s_0 + S_F * F_km2

s_k, l_k, v_k, d_k = step(s_km1, l_km1, v_km1, f_k)

check('s_k', s_k, s_0 + S_F * F_k, ZG)
check('l_k', l_k, l_0 + L_F * F_k, ZL)
check('v_k', v_k, D_F * f_k + s_0 + S_F * F_km1, ZG)
check('d_k', d_k, D_F * f_k + (D_F + S_F) * f_km1, ZG)

print('base case k = 1 and induction step k-1 -> k verified: '
      'the closed forms hold for all k >= 1')
\end{lstlisting}

\clearpage

\end{document}